\documentclass[11pt,a4paper,final]{article}
\usepackage[utf8]{inputenc}
\usepackage[english]{babel}
\usepackage{amsmath}
\usepackage{amsfonts}
\usepackage{amssymb}
\usepackage{appendix}
\usepackage{faktor}
\usepackage{yfonts}
\usepackage{upgreek}
\usepackage{lipsum}
\usepackage[hidelinks]{hyperref}
\hypersetup{
    colorlinks=true,
    urlcolor=blue,
    linkcolor=blue,
    citecolor=red,
    }
\usepackage{amsthm}
\usepackage{bbold}
\usepackage{mathrsfs}
\usepackage{wasysym}
\usepackage{mathtools}
\usepackage{bm}
\usepackage{cancel}
\usepackage{color}
\usepackage{enumitem,multicol}
\usepackage{multicol}
\usepackage{stackrel}
\usepackage{tikz}
\usepackage{tikz-cd}
\usepackage[Glenn]{fncychap}
\usepackage{subfig}
\usepackage{accents}
\usepackage{slashed}
\usepackage{psfrag}
\usepackage{todonotes}

\usepackage{stackengine}
	
\usetikzlibrary{babel}
\makeatletter
\newcommand*\owedge{\mathpalette\@owedge\relax}
\newcommand*\@owedge[1]{%
	\mathbin{%
		\ooalign{%
			$#1\m@th\bigcirc$\cr
			\hidewidth$#1\m@th\wedge$\hidewidth\cr
		}%
	}%
}
\makeatother

\makeatother
\makeatletter
\newcommand{\dotr}[1]{%
	\mathpalette\@dotr{#1}%
}
\newcommand*{\@dotr}[2]{%
	\sbox0{$\m@th#1#2$}%
	\usebox{0}%
	\raisebox{\dimexpr\ht0-\height}{$\m@th#1\@smallbullet#1\bullet$}%
	\kern\scriptspace
}
\newcommand*{\@smallbullet}[2]{%
	\scalebox{.25}{$\m@th#1#2$}%
}
\makeatother

\newcounter{mnotecount}

\newcommand{\mnote}[1]
{\protect{\stepcounter{mnotecount}}$^{\mbox{\tiny
						$\,\bullet$\themnotecount}}$ \marginpar{\hspace{-3mm}\color{red}
				\raggedright\tiny\em
				$\,\bullet$\themnotecount: #1} }

\newtheorem{teo}{Theorem}[section]
\newtheorem{cor}[teo]{Corollary}
\newtheorem{prop}[teo]{Proposition}
\newtheorem{lema}[teo]{Lemma}
\newtheorem{defi}[teo]{Definition}

\newtheorem{rmk}[teo]{Remark}
\newtheorem{con}[teo]{Conjecture}
\newtheorem{nota}[teo]{Notation}

\usepackage{makeidx}
\usepackage{graphicx}
\usepackage[left=2.50cm, right=2.50cm, top=2.45cm, bottom=2.50cm]{geometry}

\usepackage{scalerel}
\usepackage{stackengine,wasysym}

\newcommand{\uwh}[1]{%
	\mathpalette\douwidehat{#1}%
}

\makeatletter
\newcommand{\douwidehat}[2]{%
	\sbox0{$\m@th#1\widehat{\hphantom{#2}}$}%
	\sbox2{$\m@th#1x$}
	\sbox4{$\m@th#1#2$}
	\dimen0=\ht0
	\advance\dimen0 -.8\ht2
	\dimen2=\dp4
	\rlap{%
		\raisebox{\dimexpr\dimen0-\dimen2}{%
			\scalebox{1}[-1]{\box0}%
		}%
	}%
	{#2}%
}
\makeatother

\usepackage{color}
    \makeatletter
\renewcommand\part{%
	\if@openright
	\cleardoublepage
	\else
	\clearpage
	\fi
	\thispagestyle{empty}%
	\if@twocolumn
	\onecolumn
	\@tempswatrue
	\else
	\@tempswafalse
	\fi
	\null\vfil
	\secdef\@part\@spart}
\makeatother

\newcommand{\scri}{\mathscr{I}}

\newcommand{\wt}{\widetilde}
\newcommand{\wh}{\widehat}
\newcommand{\sch}{\operatorname{Sch}}

\newcommand{\scal}{\operatorname{Scal}}

\newcommand{\ol}{\overline}
\newcommand{\elltwo}{\ell^{(2)}}

\newcommand{\nablacero}{\accentset{\circ}{\nabla}}

\newcommand{\Id}{\operatorname{Id}}

\newcommand{\bY}{\textup{\textbf Y}}

\newcommand{\Y}{\textup{Y}}
\newcommand{\bcY}{\mathbb Y}
\newcommand{\bcr}{\mathbb r}
\newcommand{\bck}{\mathbb k}
\newcommand{\Q}{\mathcal{Q}}
\renewcommand{\L}{\mathcal{L}}

\newcommand{\bT}{\textup{\textbf T}}
\newcommand{\T}{\textup{T}}

\newcommand{\bU}{\textup{\textbf U}}
\newcommand{\U}{\textup{U}}

\newcommand{\bF}{\textup{\textbf F}}
\newcommand{\F}{\textup{F}}
\newcommand{\br}{\textup{\textbf r}}
\renewcommand{\r}{\textup{r}}

\renewcommand{\H}{\mc{H}}

\newcommand{\bs}{\textup{\textbf s}}
\newcommand{\s}{\textup{s}}

\newcommand{\bg}{\bm{\gamma}}

\newcommand{\mf}{\mathfrak}
\newcommand{\real}{\mathbb R}
\newcommand{\tr}{\operatorname{tr}}

\renewcommand{\div}{\operatorname{div}}

\newcommand{\st}{\stackrel}
\newcommand{\stb}{\stackbin}
\newcommand{\hess}{\operatorname{Hess }}

\newcommand{\para}{\parallel}
\newcommand{\Rcero}{\accentset{\circ}{R}}
\newcommand{\sph}{\mathbb S}
\newcommand{\eps}{\varepsilon}
\newcommand{\X}{\mathfrak{X}}
\renewcommand{\d}{\coloneqq}

\newcommand{\ric}{\operatorname{Ric}}

\newcommand{\lie}{\pounds}
\newcommand{\mc}{\mathcal}
\newcommand{\sto}{\rightarrow}
\renewcommand{\to}{\longrightarrow}

\title{\vspace*{-1.35cm}\textbf{Asymptotic expansion of the metric at null infinity}}
\author{Marc Mars$^1$\footnote{\href{mailto:marc@usal.es}{marc@usal.es}}\,\, and Gabriel Sánchez-Pérez$^{1,2}$\footnote{\href{mailto:gsperez@comillas.edu}{gsperez@comillas.edu}} \\ \\
$^1$ Departamento de Física Fundamental, Universidad de Salamanca\\
	Plaza de la Merced s/n, 37008 Salamanca, España\\ \\
$^2$ Departamento de Matemática Aplicada, ICAI, Universidad Pontificia de Comillas\\
Calle de Alberto Aguilera 25, 28015 Madrid, España.}
\date{\today}

\begin{document}
	\maketitle
	
\begin{abstract}
In this paper we analyze the conformal Einstein equations to all orders at null infinity without imposing any restriction on the spacetime dimension, the topology of $\scri$, or fall-off conditions for the Weyl tensor. In particular, we study how the equations constrain the geometry of null infinity when it is assumed to be foliated by cross-sections, not necessarily spheres. Our approach is coordinate-free and treats the conformal factor $\Omega$ as a dynamical variable. After identifying the free data at $\scri$, we show that any two asymptotically flat spacetimes sharing the same free data at null infinity are necessarily isometric to infinite order. In addition, we provide a detached definition of null infinity and prove an existence theorem for asymptotically flat spacetimes solving the field equations to infinite order at $\scri$ realizing the prescribed initial data. 
\end{abstract}

\section{Introduction}

The study of gravitational radiation is one of the central problems in mathematical relativity. It is also technically difficult because the metric that defines fall-off at infinity is itself a dynamical variable on which boundary conditions must be imposed. This difficulty is most naturally addressed within Penrose’s conformal framework \cite{Penrose1964LightCone,penrose1965zero}, where asymptotic flatness and radiation are encoded in the geometry of a null boundary $\scri$ attached to an unphysical spacetime $(\mc M,g,\Omega)$, allowing infinity to be treated as a genuine geometric object. Within this setting, several notions of asymptotic flatness have been proposed (see e.g. \cite{Penrose1968Structure,geroch1978asymptotically}), typically requiring the fulfillment of Einstein field equations in a neighbourhood of $\scri$, as well as some global assumptions such as (weak) asymptotic simplicity \cite{Penrose1968Structure}. In four spacetime dimensions, this forces null infinity to have topology $\real\times\sph^2$ \cite{newman1989global}, which in turn implies the vanishing of the Weyl tensor at $\scri$ \cite{kroon2017conformal,fernandez2022asymptotic}. As a consequence, most standard definitions of asymptotic flatness and radiation, as well as the associated notions of news \cite{geroch1977asymptotic,geroch1977asymptotic,fernandez2025news}, are inherently tied to spherical topology and four dimensions, and rely on the assumption that the rescaled Weyl tensor vanishes at null infinity. However, smooth solutions of the Einstein equations with non-spherical $\scri$ are known to exist, already in four dimensions, and the Weyl tensor need not vanish there \cite{schmidt1996vacuum}.\\

 While given an already compactified spacetime it is relatively straightforward to analyze its geometry induced at null infinity, the converse problem, namely constructing an asymptotically flat spacetime from prescribed data at $\scri$, is far more delicate and relies on a regular conformal formulation of the field equations \cite{friedrich1981asymptotic,friedrich1981regular}. Current existence results for the asymptotic characteristic problem \cite{Rendall,kannar1996existence,hilditch2020improved} are restricted to four dimensions and assume spherical topology and vanishing Weyl tensor at $\scri$. For the asymptotic hyperboloidal problem \cite{friedrich1983cauchy,friedrich1986existence,andersson1992regularity}, where spherical topology is likewise imposed, generic initial data lead to a non-vanishing Weyl tensor at null infinity, which is in turn associated with the appearance of logarithmic terms in the asymptotic expansion that spoil the smoothness of $\scri$ \cite{andersson1994hyperboloidal}. This behavior is consistent with other works on non-smooth null infinity \cite{winicour1985logarithmic,chrusciel1995gravitational,valiente2002polyhomogeneous,valiente2004new,kehrberger2022case}. \\

Motivated by the fact that comparatively fewer results are available in higher dimensions and non-spherical topologies, the aim of this paper is to analyze conformal null infinity in full generality, without assuming any restriction on the spacetime dimension, the topology of $\scri$ beyond admitting a foliation by cross-sections, or the vanishing of the Weyl tensor at infinity. More precisely, given a conformal manifold admitting a null infinity, we study how the geometry of $\scri$ is constrained by the requirement of asymptotic flatness, understood here as the validity of the conformal Einstein equations order by order at infinity. Analyzing these equations allows us to identify the free data at $\scri$, regardless of its dimension and the topology of its cross-sections. Our approach is entirely coordinate-free and treats the conformal factor $\Omega$ as a dynamical variable rather than as a fixed background quantity. Once the free data have been isolated, we prove that any two asymptotically flat spacetimes sharing the same free data at $\scri$ are necessarily isometric to infinite order. We also provide a detached definition of null infinity and prove that, given such free data, there exists an asymptotically flat spacetime in which the prescribed null hypersurface arises as its null infinity.\\

In order to study detached null hypersurfaces, it is convenient to employ the hypersurface data formalism developed in \cite{Marc3,Marc1,Marc2}. The basic building block of this formalism is the notion of \emph{metric hypersurface data} $\{\mc H,\bg,\bm\ell,\elltwo\}$, where $\bg$ is a symmetric tensor with one degenerate direction, $\bm\ell$ is a one-form, and $\elltwo$ is a scalar function on $\mc H$. When $\mc H$ is embedded in an ambient manifold, the collection $\{\bg,\bm\ell,\elltwo\}$ encodes respectively the fully tangent, tangent-transverse, and fully transverse components of the ambient metric $g$ at $\mc H$. Metric hypersurface data uniquely determine contravariant data $\{P,n\}$, where $P$ is a 2-contravariant symmetric tensor and $n$ is a vector, together with a torsion-free connection $\nablacero$ on $\mc H$. The vector field $n$ spans the kernel of $\bg$, i.e.~$\bg(n,\cdot)=0$.\\

While metric hypersurface data capture only zeroth-order information of the ambient metric, it is natural to consider higher-order derivatives at $\mc H$. Let $2\bY^{(k)}$ denote the $k$-th Lie derivative of $g$ along an arbitrary transverse vector field $\xi$. We refer to the collection $\{\bY^{(k)}\}_{k\ge1}$ as the \emph{transverse} or \emph{asymptotic expansion}. Informally, this expansion allows one to reconstruct the ambient metric order by order in the transverse direction. In a previous work \cite{Mio4}, we made this statement precise by proving that, given null metric hypersurface data $\{\mc H,\bg,\bm\ell,\elltwo\}$ and a prescribed collection of tensors $\{\bcY^{(k)}\}_{k\ge1}$ on $\mc H$, there exists an ambient manifold $(\mc M,g)$ such that $\bY^{(k)}=\bcY^{(k)}$ for all $k\ge1$. A priori, the resulting spacetime need not satisfy any field equations. In order to characterize those collections $\{\bcY^{(k)}\}_{k\ge1}$ for which $(\mc M,g)$ satisfies prescribed equations, such as the Einstein equations, we derived in \cite{Mio3} a system of identities relating the transverse derivatives of the ambient Ricci tensor at $\mc H$ to the tensors $\{\bY^{(k)}\}_{k\ge1}$.\\

The identities obtained in \cite{Mio3,Mio4} are well adapted to the Einstein equations in the bulk, but not to the study of null infinity. The reason is that the Einstein equations are naturally expressed in terms of the Ricci tensor, whereas the conformal Einstein equations involve additional structures. Consequently, the first step of the present analysis is to rewrite these identities in a form suitable for the conformal field equations. The resulting identities depend not only on the tensors $\{\bY^{(k)}\}_{k\ge1}$, but also on the transverse derivatives of the conformal factor $\Omega$ at $\scri$, which we denote by $\{\sigma^{(k)}\}_{k\ge1}$. \\

Once these identities have been established, we consider a conformal manifold $(\mc M,g,\Omega)$ with vanishing cosmological constant $\lambda=0$ and fix\footnote{In order to compute the transverse expansion at $\scri$, it is necessary to work with a specific representative of the conformal class, and hence to fix a conformal gauge.} the conformal gauge by imposing $|\nabla\Omega|^2=0$, which we refer to as a \emph{conformal geodesic gauge}. As shown in \cite{Mio6}, this gauge is uniquely determined once the value of the conformal factor $\Omega$ is prescribed on a hypersurface transverse to $\scri$. Analyzing the conformal Einstein equations order by order in this gauge leads to the following conclusions:

\begin{enumerate}
	\item At each order, the scalars $P^{ab}\Y^{(k)}_{ab}$, $\bY^{(k+1)}(n,n)$, and $\sigma^{(k+1)}$ satisfy a system of equations on $\scri$ which, except for one value of $k$, can be solved to determine these quantities in terms of lower-order data and the value of $\sigma^{(k+1)}$ on a chosen cross-section $\Sigma\hookrightarrow \scri$. This initial condition codifies the residual conformal freedom within the conformal geodesic gauge. For a very specific value of $k=m_1$, however, the system fails to be invertible, and instead we use another conformal equation, the so-called higher order Raychaudhuri equation, that allows us to determine $P^{ab}\Y^{(m_1)}_{ab}$, $\bY^{(m_1+1)}(n,n)$, and $\sigma^{(m_1+1)}$ provided an additional free function $\mf m$ on $\Sigma$ is prescribed. As pointed out recently in \cite{ciambelli2025asymptotic}, the null Raychaudhuri constraint gives rise to the Bondi mass-loss formula in four dimensions. Our result suggests that a similar conclusion holds in higher dimensions. 
	\item Once $P^{ab}\Y^{(k)}_{ab}$, $\bY^{(k+1)}(n,n)$, and $\sigma^{(k+1)}$ have been fixed, the one-form $\bY^{(k)}(n,\cdot)$ is completely determined by the conformal equations, except for a specific value of $k=m_2$ that we discuss later, at which an additional one-form $\bm\beta$ on $\Sigma$ must be prescribed as free data in order to determine $\bY^{(m_2)}(n,\cdot)$.
	\item Finally, after determining $P^{ab}\Y^{(k)}_{ab}$, $\bY^{(k+1)}(n,n)$, $\sigma^{(k+1)}$, and $\bY^{(k)}(n,\cdot)$, the remaining components of $\bY^{(k)}$ generically satisfy a transport equation along $n$, which can be integrated starting from an initial symmetric trace-free tensor $\mc Y^{(k)}$ on $\Sigma$. For even-dimensional spacetimes and for a specific value of $k=m_3$, however, $\bY^{(m_3)}$ satisfies no transport equation at all, and the components of $\bY^{(m_3)}$ not encoded in $P^{ab}\Y^{(m_3)}_{ab}$ or $\bY^{(m_3)}(n,\cdot)$ must be prescribed as additional free data.
\end{enumerate}
The number of degrees of freedom in the geometric and detached approach that we present agrees with the analysis in Bondi coordinates in dimension four \cite{sachs,compere2019lambda} and also in higher even dimensions \cite{capone2023phase,riello2024renormalization}. Furthermore, the fact that the tensor $\bY^{(m)}$ is only freely specifiable in even dimensions enforces the idea already pointed out in \cite{hollands2004conformal} that there are not smooth radiating odd dimensional spacetimes. This is why most of the definitions of asymptotic flatness in higher dimensions restrict to even dimensional spacetimes \cite{hollands2005asymptotic,tanabe2011asymptotic}.\\

The exceptional cases appearing in the second and third items give rise to two potential obstructions, analogous in spirit to the Fefferman--Graham obstruction tensor. These obstructions have already appeared in the literature in specific situations with other names (see e.g. \cite{riello2024renormalization} where they are denoted as the ``Coulombian'' and ``radiative'' anomalies). Let us see in more detail why these obstructions appear. The equation determining the one-form $\br^{(k)}\d \bY^{(k)}(n,\cdot)$ at each order takes the schematic form
\begin{equation}
	\label{eqintro1}
	\lie_n \br^{(k)} + \text{lower-order terms} = 0 ,
\end{equation}
where $\lie_n$ denotes the Lie derivative along $n$. The initial condition for this transport equation at a cross-section $\Sigma$ is obtained from another conformal equation, which reads
\begin{equation}
	\label{eqintro2}
	(\mf n - k)\,\br^{(k)}|_{\Sigma} + \text{lower-order terms} = 0 ,
\end{equation}
where $\mf n$ denotes the dimension of $\scri$. For $k\neq \mf n$, equation \eqref{eqintro2} uniquely determines $\br^{(k)}|_{\Sigma}$, which can then be used as initial data to integrate \eqref{eqintro1}. For $k=\mf n$, however, an initial condition for $\br^{(\mf n)}$ must be prescribed, which motivates the introduction of the free one-form $\bm\beta$. Moreover, if the lower-order terms in \eqref{eqintro2} do not vanish identically, the resulting spacetime cannot be smooth beyond this order. In this sense, equation \eqref{eqintro2} defines an obstruction tensor, which we refer to as the \emph{Coulombian obstruction tensor} and denote by $\mc O^{\Sigma}$. In Section~\ref{subsec_obs} we derive a necessary and sufficient condition for the vanishing of this obstruction in four spacetime dimensions and relate it to the vanishing of the Weyl tensor at $\scri$.\\

Concerning item~3, the recursive determination of $\bY^{(k)}$ is governed by another transport equation of the form
\begin{equation}
	\label{eqintro3}
	\left(\dfrac{\mf n-1}{2}-k\right)\lie_n \bY^{(k)} + \text{lower-order terms} = 0 .
\end{equation}
For $k\neq \frac{\mf n-1}{2}$, the tensor $\bY^{(k)}$ can be determined from an initial condition $\mc Y^{(k)}$ on $\Sigma$. When $k=\frac{\mf n-1}{2}$ (which occurs only when $\mf n$ is odd), $\bY^{(k)}$ satisfies no transport equation at all, and the components of $\bY^{(k)}$ not encoded in $P^{ab}\Y^{(k)}_{ab}$ or $\bY^{(k)}(n,\cdot)$ must be prescribed as free data. Furthermore, the remainder term in \eqref{eqintro3} defines another obstruction tensor, which we call the \emph{radiative obstruction tensor} and denote by $\mc O^{\scri}$. If this tensor is not identically zero, the conformal equations cannot be satisfied beyond this order. In Section~\ref{subsec_obs} we show that this obstruction vanishes identically in four spacetime dimensions and that, in six dimensions, it is closely related to the Fefferman--Graham obstruction tensor. We also conjecture that this behaviour emerges in higher dimensions (Conjecture \ref{conjecture}).\\

Once the free data $\mc D$ have been identified (we give a precise definition of free data in Remark \ref{remark_freedata}), we prove a uniqueness theorem (Theorem~\ref{uniqueness}), showing that any two asymptotically flat spacetimes sharing the same free data at $\scri$ are necessarily isometric to infinite order at their respective $\scri$. Our notion of asymptotic flatness is less restrictive than others commonly adopted in the literature such as the ones described above, as it only requires the conformal field equations to be satisfied to infinite order at $\scri$. An informal version of the uniqueness result is as follows.
\begin{teo}[Informal version, see Theorem~\ref{uniqueness}]
	Let $(\mc M,g,\Omega)$ and $(\mc M',g',\Omega')$ be two asymptotically flat spacetimes both written in a conformal geodesic gauge. Suppose that their respective null infinities have the same free data $\mc D$. Then $(\mc M,g,\Omega)$ and $(\mc M',g',\Omega')$ are isometric to infinite order at $\scri$.
\end{teo}

Having established that the free data fully characterize the geometry at null infinity, we prove the converse statement: given such free data on an abstract null hypersurface, together with the zeroth-order data that we call $\scri$-\emph{structure data} (see Definition~\ref{universal}), there exists an asymptotically flat conformal spacetime realizing them (Theorem~\ref{teorema}). An informal formulation of this existence theorem is the following.
\begin{teo}[Informal version, see Theorem~\ref{teorema}]
	Given $\scri$-structure data and free data $\mc D$ such that the radiative and Coulombian obstruction tensors vanish, there exists an asymptotically flat spacetime realizing these data.
\end{teo}

The construction of the ambient spacetime is technically involved, since the higher order conformal field equations are coupled order by order and, moreover, there are more equations than variables to be determined. A key part of the proof consists in showing that the redundant equations are automatically satisfied once the remaining ones are solved. At the core of this redundancy is the contracted Bianchi identity. While this redundancy is irrelevant when the ambient spacetime is already given, it becomes a central issue when the spacetime is to be constructed order by order, as it is the case in our existence theorem.\\

The structure of this manuscript is as follows. Sections~\ref{sec_hypersurfacedata} and \ref{sec_QE} provide an overview of the aspects of hypersurface data and conformal geometry needed in this work. In Section~\ref{section_transverse} we review the identities derived in \cite{Mio3} relating the transverse expansion to derivatives of the ambient Ricci tensor at a null hypersurface, and we derive new identities adapted to the conformal field equations. We also analyze the redundancy of the equations order by order and identify the free data. In Section~\ref{scri} we prove that these data completely characterize the geometry at null infinity and that any such data set can be realized by an asymptotically flat spacetime. Finally, in Section~\ref{subsec_obs} we study in detail the radiative and Coulombian obstruction tensors in the lowest spacetime dimensions in which they appear. The paper contains four appendices. Appendix~\ref{appendix} collects several identities used throughout the paper. Appendix~\ref{appendixB} presents auxiliary calculations that are not included in the main body for the sake of clarity. Appendix~\ref{appendixC} presents the conformal field equations in full generality. Finally, Appendix~\ref{appendixD} is devoted to the derivation of the higher-order Raychaudhuri equation.

\section*{Notation and conventions}

Throughout this paper $(\mc M,g)$ denotes a smooth $d$-dimensional semi-Riemannian manifold of some signature $(p,q)$ with both $p$ and $q$ different from zero. We employ both index-free and abstract index notation at our convenience. Ambient indices are denoted with Greek letters, abstract indices on a hypersurface are written in lowercase Latin letters, and abstract indices at cross-sections of a hypersurface are expressed in uppercase Latin letters. As usual, square brackets enclosing indices denote antisymmetrization and parenthesis are for symmetrization. The symmetrized tensor product is denoted with $\otimes_s$. By $\mc F(\mc M)$, $\X(\mc M)$ and $\X^{\star}(\mc M)$ we denote respectively the set of smooth functions, vector fields and one-forms on $\mc M$. The subset $\mc F^{\star}(\mc M)\subset\mc F(\mc M)$ consists of the nowhere vanishing functions on $\mc M$. The pullback of a function $f$ via a diffeomorphism $\Phi$ will be denoted by $\Phi^{\star}f$ or simply by $f$ depending on the context. A $(p,q)$-tensor refers to a tensor field $p$ times contravariant and $q$ times covariant. Given any pair of $(2,0)$ and $(0,2)$ tensors $A^{ab}$ and $B_{cd}$ we denote $\tr_A \bm B \d A^{ab}B_{ab}$. We employ the symbol $\nabla$ for the Levi-Civita connection of $g$, and our convention for the Riemann tensor is $$R(X,Y)Z= [\nabla_X,\nabla_Y]Z - \nabla_{[X,Y]}Z.$$ 

We use the notation $\lie_X^{(m)}T$ to denote the $m$-th Lie derivative of the tensor $T$ along $X$, and $X^{(m)}(f)$ for the $m$-th directional derivative of the function $f$ along $X$. When $m=1$ we also write $\lie_X T$ and $X(f)$, respectively, and when $m=0$ they are just the identity operators. All manifolds are assumed to be connected and smooth. Embedded hypersurfaces are assumed to be two-sided unless otherwise indicated.
\section{Review of hypersurface data formalism}
\label{sec_hypersurfacedata}

In this section we review the notions of the \textit{hypersurface data formalism} needed in this paper. Details can be found in \cite{Marc1,Marc2} (see also \cite{Marc3,tesismiguel}). We restrict from the beginning to the null case. \textbf{Null metric hypersurface data} is a set $\{\mc H,\bg,\bm\ell,\elltwo\}$ consisting of an $\mf n$-dimensional manifold $\mc H$, a symmetric, degenerate $(0,2)$-tensor field $\bg$ with just one degenerate direction at each point, a one-form $\bm\ell$, and a scalar function $\elltwo$ on $\mc H$, provided that the 2-covariant, symmetric tensor $\bm{\mc A}|_p$ on $T_p\mc H\times\real$ defined by $$\mc A|_p\left((W,a),(V,b)\right) \d \bg|_p (W,V) + a\bm\ell|_p(V)+b\bm\ell|_p(W)+ab\ell^{(2)}|_p$$ 

is non-degenerate at every $p\in\mc H$. A five-tuple $\{\mc H,\bg,\bm\ell,\elltwo,\bY\}$, where $\bY$ is a $(0,2)$ symmetric tensor field on $\mc H$, is called \textbf{null hypersurface data}. The non-degeneracy of $\bm{\mc A}$ allows one to define a 2-contravariant, symmetric tensor field $P$ and a vector $n$ on $\mc H$ by means of
\begin{multicols}{2}
	\noindent
	\begin{align}
		\gamma_{ab}n^b &=0,\label{gamman}\\
		\ell_an^a&=1,\label{ell(n)}
	\end{align}
	\begin{align}
		P^{ab}\ell_b+\ell^{(2)} n^a&=0,\label{Pell}\\
		P^{ac}\gamma_{cb} + \ell_b n^a &=\delta^a_b.\label{Pgamma}
	\end{align}
\end{multicols}

Given null hypersurface data we define the tensor fields\\

\begin{minipage}{0.5\textwidth}
	\noindent
	\begin{equation}
		\label{defK}
		\bU\d \dfrac{1}{2}\lie_n\bg,
	\end{equation} 
\end{minipage}
\begin{minipage}{0.5\textwidth}
	\noindent
	\begin{equation}
		\label{defF}
		\bF\d \dfrac{1}{2}d\bm\ell,
	\end{equation} 
\end{minipage}
\,\\
as well as the contractions $\bs\d \bF(n,\cdot)$, $\br\d \bY(n,\cdot)$ and $\kappa_n\d -\bY(n,n)$. Note that $\bs(n)=0$ because $\bF$ is a two-form. It is also useful to introduce a $(1,1)$-tensor $V$ defined by 
\begin{equation}
	\label{V}
V^b{}_a \d P^{bc}(\Y_{ac}+\F_{ac}) + \dfrac{1}{2}n^b\nablacero_a\elltwo,
\end{equation}
and write out its contraction with $n^a$, namely
\begin{equation}
	\label{Vn}
	V^b{}_a n^a = P^{bc}(\r_c+\s_c)+\dfrac{1}{2}n(\elltwo) n^b.
\end{equation}
Note that these definitions are fully detached from any ambient space. To connect with the standard concept of hypersurface we say that null metric hypersurface data $\{\mc H,\bg,\bm\ell,\ell^{(2)}\}$ is $(\Phi,\xi)$-embedded in a semi-Riemannian manifold $(\mc M,g)$ if there exists an embedding $\Phi:\mc H\hookrightarrow\mc M$ and a vector field $\xi$ along $\Phi(\mc H)$ everywhere transversal to $\Phi(\mc H)$, called rigging, such that 
\begin{equation}
	\Phi^{\star}(g)=\bg, \hspace{0.5cm} \Phi^{\star}\left(g(\xi,\cdot)\right) = \bm\ell,\hspace{0.5cm} \Phi^{\star}\left(g(\xi,\xi)\right) = \ell^{(2)}.
\end{equation}
 Furthermore, null hypersurface data $\{\mc H,\bg,\bm\ell,\ell^{(2)},\bY\}$ is embedded provided that, in addition,
 \begin{equation}
 	\dfrac{1}{2}\Phi^{\star}\left(\lie_{\xi} g\right) = \bY.
 \end{equation}
  For embedded data, $\bU$ coincides with the second fundamental form of $\Phi(\mc H)$ w.r.t. the unique normal one-form $\bm{\nu}$ satisfying $\bm{\nu}(\xi)=1$. Moreover, introducing a (local) basis $\{e_a\}$ of $\mc H$, the inverse metric $g^{\alpha\beta}$ at $\Phi(\mc H)$ can be then written in the basis $\{\xi, \wh{e}_a\d \Phi_{\star}e_a\}$ as
\begin{equation}
	\label{inversemetric}
	g^{\alpha\beta}\st{\mc H}{=}	P^{ab}\wh e_a^{\alpha}\wh e_b^{\beta} + n^{a}\wh e_a^{\alpha}\xi^{\beta} + n^{b}\wh e_b^{\beta}\xi^{\alpha} .
\end{equation}


In the embedded picture the notion of rigging vector is non-unique, since given a rigging $\xi$ any other vector of the form $\xi' = z(\xi+\Phi_{\star}V)$ with $(z,V)\in\mc{F}^{\star}(\mc H)\times\X(\mc H)$ is also transverse to $\Phi(\mc H)$. Translating this into the abstract setting, one defines the gauge transformed data by 

\begin{minipage}{0.4\textwidth}
	\noindent
	\begin{align}
\mc{G}_{(z,V)}\left(\bg \right)&\d \bg,\label{transgamma}\\
\mc{G}_{(z,V)}\left( \bm{\ell}\right)&\d z\left(\bm{\ell}+\bg(V,\cdot)\right),\label{tranfell}
	\end{align}
\end{minipage}
\begin{minipage}{0.6\textwidth}
	\begin{align}
\mc{G}_{(z,V)}\big( \ell^{(2)} \big)&\d z^2\big(\ell^{(2)}+2\bm\ell(V)+\bg(V,V)\big),\label{transell2}\\
\mc{G}_{(z,V)}\left( \bY\right)&\d z \bY + \bm\ell\otimes_s d z +\dfrac{1}{2}\lie_{zV}\bg.\label{transY}
	\end{align}
\end{minipage}
\,

These transformations induce the following gauge behaviour for $P$ and $n$ \cite{Marc1},
\begin{multicols}{2}
	\noindent
	\begin{equation}
		\label{gaugeP}
		\mc{G}_{(z,V)}\left(P \right) = P -2V\otimes_s n,
	\end{equation}
	\begin{equation}
		\label{transn}
		\mc{G}_{(z,V)}\left( n \right)= z^{-1}n.
	\end{equation}
\end{multicols}
Given null metric hypersurface data $\{\mc H,\bg,\bm\ell,\elltwo\}$ it is possible to define a torsion-free connection $\nablacero$ on $\mc H$ by means of \cite{Marc2}
\begin{multicols}{2}
	\noindent
	\begin{equation}
		\label{nablagamma}
		\nablacero_a\gamma_{bc} = -\ell_c\U_{ab} - \ell_b\U_{ac},
	\end{equation}
	\begin{equation}
		\label{nablaell}
		\nablacero_a\ell_b  = \F_{ab} - \elltwo\U_{ab},
	\end{equation}
\end{multicols}
which in the embedded case can be related with the Levi-Civita connection $\nabla$ of $g$ by 
\begin{equation}
	\label{connections}
	\nabla_{\Phi_{\star}X}\Phi_{\star}Y \st{\mc H}{=} \Phi_{\star}\nablacero_X Y - \bY(X,Y)\nu - \bU(X,Y)\xi,\qquad X,Y\in\X(\mc H).
\end{equation}
The action of $\nablacero$ on the contravariant data $\{P,n\}$ turns out to be \cite{Marc2}
\begin{align}
	\nablacero_c n^b & =\s_c n^b + P^{ba}\U_{ca},\label{derivadannull}\\
	\nablacero_c P^{ab} & = -\big(n^aP^{bd}+n^bP^{ad}\big) \F_{cd} - n^an^b (d\elltwo)_c.\label{derivadaP}
\end{align}
In the embedded case, the $\nabla$-derivative of $\xi$ along tangent directions to $\mc H$ is \cite{Marc1}
\begin{equation}
	\label{nablaxi}
	\wh{e}_a^{\rho}\nabla_{\rho}\xi^{\beta}\st{\mc H}{=} (\r-\s)_a\xi^{\beta} + V^b{}_a \wh e_b^{\beta} ,
\end{equation}
and as a consequence (cf. \eqref{inversemetric})
\begin{equation}
	\label{nablaxiup}
g^{\mu\rho}\nabla_{\rho}\xi^{\beta} \st{\mc H}{=} \big(P^{cd} \wh e_d^{\mu}+n^c\xi^{\mu}\big)\big((\r-\s)_c\xi^{\beta}+V^b{}_c \wh e_b^{\beta}\big) + \nu^{\mu}\xi^{\rho}\nabla_{\rho}\xi^{\beta}.
\end{equation}

A direct consequence of \eqref{derivadannull} is that for any one-form $\bm\omega$ \cite{tesismiguel},
\begin{equation}
	\label{nnablaomega}
	2n^b \nablacero_{(a}\omega_{b)} = \lie_n\omega_a + \nablacero_a(\bm\omega(n)) - 2\big(\bm\omega(n)\s_a + P^{bc}\U_{ac}\omega_b\big),
\end{equation}
while \eqref{derivadaP} implies that for every $(0,2)$ tensor $\T_{ab}$ the $P$-trace of its Lie derivative along $n$ is
\begin{equation}
	\label{lietrPY}
\tr_P\lie_n\bT\d 	P^{ab}\lie_n \T_{ab} = \lie_n(\tr_P\bT) + 4P^{ab}\T_{ac}\s_b n^c +2P^{ac}P^{bd}\U_{cd}T_{ab} + n(\elltwo)\bT(n,n). 
\end{equation}
Finally, we will frequently use that every $(0,2)$ tensor $\T_{ab}$ can be uniquely decomposed as \cite{MarcAbstract}
\begin{equation}
	\label{decomp_EM}
\T_{ab}= \wh{\T}_{ab}+ \frac{\tr_P{\T}+\T(n,n)\elltwo}{\mf n-1}\gamma_{ab}+2\ell_{(a}\T_{b)c}n^c +\T(n,n)\ell_a\ell_b,
\end{equation}
where $\wh{\T}_{ab}$ is a tensor lying on the kernel of the so-called energy-momentum map, i.e. a tensor satisfying $P^{ab}\wh{\T}_{ab}=0$ and $\wh{\T}_{ab}n^a=0$. We will often call $\wh \T_{ab}$ the transverse part of $\T_{ab}$, or when $\T_{ab}=\wh \T_{ab}$ we will say that $\T_{ab}$ is a transverse tensor. Similarly, for a one-form $\bm\omega$ we define 
\begin{equation}
	\label{decom2}
	\wh{\bm\omega}\d \bm\omega -\bm\omega(n)\bm\ell
\end{equation} 
and call $\wh{\bm\omega}$ the transverse part of $\bm\omega$, that satisfies $\wh{\bm\omega}(n)=0$.
\section{Quasi-Einstein manifolds}
\label{sec_QE}

In this section we review the basic aspects of conformal completions of manifolds. The results are well-known. However, our presentation puts the emphasis on the notion of ``quasi-Einstein'' manifold. This follows e.g. \cite{curry2018introduction,mars2025classification}. We start by recalling the definition of the Schouten tensor for a $(d\ge 3)$-dimensional metric $g$ in terms of the Ricci,
\begin{equation}
	\label{defi_Sch}
	\sch_g \d \dfrac{1}{d-2}\left(\ric_g-\dfrac{\scal_g}{2(d-1)}g\right).
\end{equation}
We will employ the symbols $L_{\alpha\beta}$ and $L$ for the Schouten in abstract index notation and for its trace, respectively. The Schouten and Ricci scalars are related by $L=\frac{R}{2(d-1)}$. From the second Bianchi identity it follows that\\

%
\begin{minipage}{0.4\textwidth}
	\noindent
\begin{equation}
	\label{bianchiL}
	\nabla^{\mu}L_{\mu\alpha}-\nabla_{\alpha}L=0,
\end{equation} 
\end{minipage}
\begin{minipage}{0.6\textwidth}
	\noindent
\begin{equation}
	\label{bianchiweyl}
	\nabla_{\mu}C^{\mu}{}_{\nu\alpha\beta} = (d-3)\left(\nabla_{\alpha}L_{\beta\nu}-\nabla_{\beta}L_{\alpha\nu}\right),
\end{equation} 
\end{minipage}
\,

$C^{\mu}{}_{\nu\alpha\beta}$ being the Weyl tensor. The Weyl, Schouten and Ricci tensors transform as follows under a conformal rescaling $\wh g= \omega^2g$ \cite{Wald},
\begin{align}
	\wh{C}^{\alpha}{}_{\beta\mu\nu} & = C^{\alpha}{}_{\beta\mu\nu},\\
	\wh{L}_{\alpha\beta} &= L_{\alpha\beta} - \dfrac{\nabla_{\alpha}\nabla_{\beta}\omega}{\omega} + \dfrac{2\nabla_{\alpha}\omega\nabla_{\beta}\omega}{\omega^2} - \dfrac{|\nabla\omega|^2_g}{2\omega^2} g_{\alpha\beta} ,\label{trans_schouten}\\
	\wh\nabla & = \nabla + \omega^{-1}\left(2\Id\otimes_s d\omega - \nabla\omega\otimes g\right),\label{transLC}\\
	\wh{L} & = \omega^{-2}\left(L - \dfrac{\square_g \omega}{\omega} - \dfrac{(d-4)|\nabla\omega|^2_g}{2\omega^2}\right).\label{trans_L}
\end{align}

Let $(\mc M,[g])$ be a $d$-dimensional conformal structure of arbitrary semi-Riemannian signature. For each $g\in[g]$ we construct the differential operator
\begin{equation}
	\label{Aoperator}
	A_g(f) \d (\hess_g f +f \sch_g)^{tf}, \qquad f\in\mc F(\mc M),
\end{equation} 
where ``$tf$'' denotes the trace free part w.r.t. $g$. One can easily check that $A_{\omega^2g}(\omega f) = \omega A_g(f)$, and therefore 
$$\Omega^{-2}g\in [g] \text{ is Einstein } \quad \Longleftrightarrow \quad (\hess_g\Omega + \Omega \sch_g)^{tf} = 0.$$ 

This condition can be rewritten in an equivalent way in terms of the scalar
\begin{equation}
	\label{def_s}
q\d \dfrac{\square_g \Omega + \Omega L}{d}
\end{equation} 
as $\hess_g\Omega + \Omega \sch_g - qg = 0$. This discussion motivates the following definition.

\begin{defi}
	\label{defi_CEE}
Let $(\mc M,g)$ be a semi-Riemannian manifold of dimension $d\ge 3$, $\Omega\in\mc F(\mc M)$ a non-identically zero function and $\mc T$ a trace-free, two-covariant tensor field. We say that the four-tuple $(\mc M,g,\Omega,\mc T)$ is a quasi-Einstein manifold provided that
	\begin{equation}
		\label{quasiEinsteinT}
\hess_g\Omega + \Omega\sch_g -qg = \mc T.
	\end{equation} 
A quasi-Einstein manifold is called \textit{vacuum} (or equivalently\footnote{The sense in which these equations are equivalent to the standard vacuum conformal field equations of Friedrich is discussed in Remark \ref{remarkCEE}.} is said to satisfy the conformal Einstein equations) when $\mc T=0$, 
	\begin{equation}
		\label{quasiEinstein}
		\hess_g\Omega + \Omega\sch_g -qg = 0.
	\end{equation}
Note that if $(\mc M,g,\Omega,\mc T)$ is a quasi-Einstein manifold, then $(\mc M,\omega^2g,\omega\Omega,\omega\mc T)$ is also a quasi-Einstein manifold for every $\omega\in\mc F^{\star}(\mc M)$.
\end{defi}
From the transformation $\wh{\Omega}=\omega\Omega$ and using \eqref{transLC} and \eqref{trans_L}, the behaviour of $q$ under conformal rescalings is 
\begin{equation}
	\label{transs}
	\wh{q} = \omega^{-1}\left( q + \dfrac{\langle\nabla\Omega,\nabla\omega\rangle_g}{\omega} + \dfrac{\Omega|\nabla\omega|^2_g}{2\omega^2}\right).
\end{equation}
A direct consequence of \eqref{quasiEinsteinT} and \eqref{bianchiweyl} is that for any quasi-Einstein manifold $(\mc M,g,\Omega,\mc T)$,
\begin{align}
\nabla_{\mu} q &=   \dfrac{1}{d-1}\nabla_{\rho}\mc T^{\rho}{}_{\mu}  -L^{\rho}{}_{\mu}\nabla_{\rho}\Omega ,	\label{nablas}\\
C^{\alpha}{}_{\beta\mu\nu}\nabla_{\alpha}\Omega &= \dfrac{\Omega}{d-3}\nabla_{\alpha} C^{\alpha}{}_{\beta\mu\nu}-2\nabla_{[\mu}\mc T_{\nu]\beta}+\dfrac{2}{d-1}g_{\beta[\mu}\nabla^{\rho}\mc T_{\nu]\rho}.\label{CnablaOmega}
\end{align}
Moreover, for any quasi-Einstein manifold satisfying $\Omega\nabla_{\beta}\mc T^{\alpha\beta} - (d-1)\mc T^{\alpha\beta}\nabla_{\beta}\Omega=0$ (this holds in particular for any vacuum quasi-Einstein manifold) one can define a conformally invariant constant
\begin{equation}
	\label{lambda}
	\lambda\d 2q\Omega-|\nabla\Omega|^2_g.
\end{equation}

Its geometric interpretation is as follows. Consider a vacuum quasi-Einstein manifold $(\mc M,g,\Omega)$ and define the metric $\wh g\d \Omega^{-2}g$ on the subset $\{\Omega\neq 0\}$. Plugging $\omega=\Omega^{-1}$ into \eqref{trans_schouten} and using the quasi-Einstein equation \eqref{quasiEinstein}, $$\wh{L}_{\alpha\beta} = L_{\alpha\beta}+\dfrac{\nabla_{\alpha}\nabla_{\beta}\Omega}{\Omega} - \dfrac{|\nabla\Omega|^2}{2\Omega^2}g_{\alpha\beta} = \dfrac{2q\Omega  -|\nabla\Omega|_g^2}{2\Omega^2}g_{\alpha\beta} = \dfrac{\lambda}{2}\wh g_{\alpha\beta}.$$ 

So, up to a numerical factor, $\lambda$ is the cosmological constant associated to the Einstein representative of $(\mc M,[g])$. Thus, it makes sense to talk about $\lambda$-vacuum quasi-Einstein manifolds.\\

At the points where $\Omega\neq 0$ one can define the rescaled Weyl tensor $\mf{D}^{\alpha}{}_{\beta\mu\nu}\d \Omega^{3-d}C^{\alpha}{}_{\beta\mu\nu}$. From \eqref{CnablaOmega},
\begin{equation}
	\label{divd}
	\nabla_{\alpha} \mf{D}^{\alpha}{}_{\beta\mu\nu} = 2(d-3)\Omega^{2-d}\left(\nabla_{[\mu}\mc{T}_{\nu]\beta} - \dfrac{1}{d-1}\nabla_{\alpha}\mc{T}^{\alpha}{}_{[\nu} g_{\mu]\beta}\right).
\end{equation}

It then follows that for $\lambda$-vacuum quasi-Einstein manifolds, the rescaled Weyl tensor satisfies a regular PDE on the closure of $\{\Omega\neq 0\}$. As we show now, $\{\Omega\neq 0\}$ is dense on $\mc M$, so $\mf{D}^{\alpha}{}_{\beta\mu\nu}$ satisfies a regular PDE everywhere on $\mc M$. This does not mean, however, that $\mf{D}^{\alpha}{}_{\beta\mu\nu}$ extends regularly to $\{\Omega=0\}$. The first part of the lemma is well-known, see e.g. \cite{krtouvs2004asymptotic}. The second is probably known to experts but we could not find an explicit proof in the literature.

\begin{lema}
	Let $(\mc M,g,\Omega)$ a $\lambda$-vacuum quasi-Einstein manifold with constant $\lambda$ and assume $\{\Omega=0\}\neq\emptyset$. Then,
	\begin{enumerate}
		\item If $\lambda\neq 0$, $\{\Omega=0\}$ is an embedded hypersurface with non-zero normal given by $\nabla\Omega$. Furthermore, it is spacelike when $\lambda>0$, and timelike when $\lambda<0$. 
		\item If $\lambda=0$, then except for a (possibly empty) collection of isolated points $\{p_i\}$, $\{\Omega=0\}$ is an embedded null hypersurface with nowhere zero normal $\nabla\Omega$.
	\end{enumerate}
	\begin{proof}
A point $p\in\mc M$ is called singular provided that $\Omega|_p=0$ and $\nabla\Omega|_p=0$. Then, from equation \eqref{lambda} the set $\{\Omega=0\}$ cannot admit any singular point in the case $\lambda\neq 0$, and therefore it is a smooth embedded hypersurface. The vector field $\nabla\Omega$ is normal to $\{\Omega=0\}$ and again by equation \eqref{lambda} we see that $\{\Omega=0\}$ is spacelike when $\lambda>0$ and timelike when $\lambda<0$.\\
		
For item 2. we first prove that singular points are necessarily isolated. Let $p\in\{\Omega=0\}$ be a singular point and define $q_p\d q(p)$. By \eqref{transs} the transformation of $q_p$ is $\hat{q}_p = \omega(p)^{-1} q_p$, so being zero/non-zero is a conformally invariant statement. Let us prove that $q_p\neq 0$. Let $\{e_a\}$ be an orthonormal basis of $T_p\mc M$, $\mc I\subseteq\real$ an interval, and let $\gamma(v)$, $v\in\mc I$, be an affinely parametrized geodesic starting at $p$. Let $\{e_a(v)\}$ the basis at $T_{\gamma(v)}\mc M$ obtained by parallel transport of $\{e_a\}$ along $\gamma(v)$. Define the following quantities $$\tilde\Omega(v)\d \Omega(\gamma(v)),\qquad \tilde\Omega_a(v)\d \langle \nabla\Omega|_{\gamma(v)},e_a(v)\rangle_g,\qquad \tilde{q}(v)\d q(\gamma(v)).$$ If $\eta_{ab}\d g(e_a,e_b)$ then $g(e_a(v),e_b(v)) = \eta_{ab}$ for all $v\in\mc I$. Observe that 
		\begin{equation}
			\label{aux1}
			\nabla\Omega|_{\gamma(v)} = \eta^{ab}\tilde\Omega_a(v) e_b(v),
		\end{equation} 
		where $\eta^{ab}\d \eta_{ab}$. Let us establish the following equations,
		\begin{equation}
			\label{system}
			\begin{aligned}
				\dfrac{d\tilde\Omega}{dv} & = \langle\dot\gamma,\nabla\Omega|_{\gamma(v)}\rangle_g = \eta^{ab}\tilde\Omega_a(v) \langle\dot\gamma(v),e_b(v)\rangle_g,\\
				\dfrac{d\tilde\Omega_a}{dv} & = (\nabla_{\dot\gamma}\nabla_{e_a}\Omega)|_{\gamma} = (\hess_g\Omega)|_{\gamma}(\dot\gamma,e_a) = \tilde{q}\langle\dot\gamma,e_a\rangle_g -\tilde\Omega \sch_g(\dot\gamma,e_a),\\
				\dfrac{d\tilde q}{dv} & = \nabla_{\dot\gamma}q|_{\gamma} = -\sch_g(\dot\gamma,\nabla\Omega) = -\eta^{ab} \sch_g(\dot\gamma,e_a)\tilde\Omega_b,
			\end{aligned}
		\end{equation}
		where the first one follows from \eqref{aux1}, the second from \eqref{quasiEinstein} and the last one from \eqref{nablas}. We also compute for later use
		\begin{equation}
			\label{aux2}
			\dfrac{d^2\tilde\Omega}{dv^2} = \nabla_{\dot\gamma}\nabla_{\dot\gamma}\Omega|_{\gamma} = \hess_g\Omega(\dot\gamma,\dot\gamma) = \tilde{q}g(\dot\gamma,\dot\gamma) - \tilde\Omega \sch_g(\dot\gamma,\dot\gamma).	
		\end{equation}
		Equations \eqref{system} constitute a linear first order system of homogeneous ODE for $\{\tilde\Omega,\tilde\Omega_a,\tilde q\}$. Since $p$ is a singular point we have $\tilde\Omega(0) =\tilde\Omega_a(0)=0$. If, moreover, $q_p=0$, then also $\tilde q(0)=0$ and thus by uniqueness of the solution, all the points on $\gamma(v)$ would be singular with $q=0$. Since $\gamma$ is arbitrary it follows that the set of singular points with $q=0$ is open (and obviously also closed). Since $\mc M$ is connected, this set is the whole of $\mc M$, which contradicts the fact that $\Omega$ is not identically zero. Hence, $q_p\neq 0$ at any singular point. Next we prove that singular points are isolated. Since $q_p\neq 0$ it follows from \eqref{quasiEinstein} that $\hess_g\Omega|_p\neq 0$ and thus $p$ is a non-degenerate critical point of $\Omega$. By Morse lemma \cite{bott1982lectures} there exists a neighbourhood $\mc U$ of $p$ and coordinates $\{x^{\alpha}\}_{\alpha=1}^d$ with $x^{\alpha}(p)=0$ such that $\Omega = -(x^1)^2 - \cdots - (x^{\sigma})^2 + (x^{\sigma+1})^2 + \cdots + (x^{d})^2$, the number $\sigma$ being the Morse index. Clearly the differential $d\Omega$ on $\mc U$ only vanishes at $p$, so $p$ is isolated and $\nabla\Omega$ is a non-zero normal to $(\{\Omega=0\}\setminus\{p\})\cup \mc U$, and thus by \eqref{lambda} this set is a smooth null hypersurface (it may be empty, e.g. when $g$ is Riemannian). 	
	\end{proof}
\end{lema}

We then define $\scri$ as the set of points where $\Omega=0$ and $\nabla\Omega\neq 0$. In the following remark we clarify the form of the conformal equations \eqref{quasiEinstein} used in this paper (which follows, e.g., \cite{curry2018introduction}), and compare it with Friedrich's original formulation.

\begin{rmk}
\label{remarkCEE}
Friedrich's original formulation of the vacuum conformal field equations can be found, for instance, in \cite{friedrich2002conformal,friedrich2015geometric}. This consists of the following system for the unknown fields $(g,\Omega,q,L,\mf D)$:

\begin{minipage}{0.6\textwidth}
	\noindent
	\begin{align}
R^{\mu}{}_{\nu\lambda\rho} = 2\left(g^{\mu}{}_{[\lambda}L_{\rho]\nu} - g_{\nu[\lambda} L_{\rho]}{}^{\mu}\right) + \Omega^{d-3}\mf{D}^{\mu}{}_{\nu\lambda\rho}\label{FE4}\\
\nabla_{\mu}\nabla_{\nu}\Omega = -\Omega L_{\mu\nu} + q g_{\mu\nu},\label{FE1}\\
\lambda = 2q\Omega -|\nabla\Omega|_g^2,\label{FE6}
	\end{align}
\end{minipage}
\begin{minipage}{0.4\textwidth}
	\noindent
	\begin{align}
\nabla_{\mu}\mf{D}^{\mu}{}_{\alpha\beta\gamma} = 0,\label{FE5}\\
\nabla_{\mu} q = - L_{\mu\nu}\nabla^{\nu}\Omega,\label{FE3}\\
2\nabla_{[\lambda}L_{\rho]\nu} = \Omega^{d-4}\mf{D}^{\mu}{}_{\nu\lambda\rho} \nabla_{\mu}\Omega.\label{FE2}
	\end{align}
\end{minipage}
\,\\
Equation \eqref{FE4} corresponds to the decomposition of the Riemann tensor into its irreducible parts, implying that $C^{\mu}{}_{\nu\lambda\rho} =\Omega^{d-3}\mf{D}^{\mu}{}_{\nu\lambda\rho}$ is the Weyl tensor of $g$ and $L_{\mu\nu}$ is its Schouten tensor. As a consequence, Equation \eqref{FE1} is precisely the quasi-Einstein equation \eqref{quasiEinstein}. Equations \eqref{FE6}, \eqref{FE5}, and \eqref{FE3} correspond exactly to \eqref{lambda}, \eqref{divd}, and \eqref{nablas} with $\mc{T}=0$, while \eqref{FE2} follows immediately upon substituting \eqref{bianchiweyl} into \eqref{CnablaOmega} and using $C^{\mu}{}_{\nu\lambda\rho} =\Omega^{d-3}\mf{D}^{\mu}{}_{\nu\lambda\rho}$. Consequently, the full Friedrich system can be derived from the quasi-Einstein equation \eqref{quasiEinstein}, justifying our terminology in Definition \ref{defi_CEE}.
\end{rmk}

Let $(\mc M,g,\Omega)$ be a $d=\mf n+1$ dimensional (vacuum) quasi-Einstein manifold with $\lambda=0$. Consider an embedding $\Phi:\scri\hookrightarrow\mc M$ and rigging $\xi$, and let $\{\scri,\bg,\bm\ell,\elltwo\}$ be the corresponding embedded null metric hypersurface data. Since $\nabla\Omega$ is non-vanishing, null and tangent to $\scri$, there must exist a non-vanishing function $\sigma$ such that $\nabla\Omega = \sigma \nu$ (and hence $\lie_\xi\Omega=\sigma$) on $\scri$. Moreover, by pulling-back \eqref{quasiEinstein} to $\{\Omega=0\}$ and noting that $e_a^{\alpha}e_b^{\beta} \nabla_{\alpha}\nabla_{\beta}\Omega = \frac{1}{2}e_a^{\alpha}e_b^{\beta}\big(\lie_{\nabla\Omega} g\big)_{\alpha\beta} = \sigma \U_{ab}$, it follows that $\bU = q\sigma^{-1}\bg$, so $\scri$ is totally umbilical. Furthermore, the conformal freedom $g' = \omega^2 g$ leaves a remnant conformal freedom on $\scri$ of the form $$\bg ' = \omega^2\bg,\qquad \bm\ell ' = \omega^2 \bm\ell,\qquad \elltwo{}' = \omega^2 \elltwo,\qquad P' = \omega^{-2}P,\qquad n' = \omega^{-2}n.$$ 

Since $\nabla'\Omega' = \omega^{-1} \nabla\Omega = \omega^{-1} \sigma \nu  = \omega \sigma \nu'$ on $\scri$, it follows that the function $\sigma$ scales by $\sigma' = \omega\sigma$ and $q|_{\scri}$ by (cf. \eqref{transs}) $q' \st{\scri}{=} \omega^{-1} q + \omega^{-2}\sigma \lie_n\omega$. The functions $q$ and $\omega$ are obviously independent of the rigging, so $q|_{\scri}$ and $\omega|_{\scri}$ are gauge invariant quantities. One the other hand $\sigma = \lie_{\xi} \Omega|_{\scri}$ has gauge behavior 
\begin{equation}
	\label{transsigma1}
\mc{G}_{(z,V)}	\sigma=z\sigma.
\end{equation}
Observe also that $q|_{\scri}$ and $\sigma|_{\scri}$ are not independent because from \eqref{lambda}, $$\nabla^{\alpha}\Omega\nabla_{\alpha}\nabla_{\beta}\Omega = \dfrac{1}{2}\nabla_{\beta}\big(|\nabla\Omega|_g^2\big) \st{\scri}{=} q\nabla_{\beta}\Omega,$$

 so $q|_{\scri}$ is the surface gravity of $\nabla\Omega = \sigma \nu$ at $\scri$. Given that $\nabla_{\nu}\nu = \kappa_n\nu$ (cf. \eqref{connections}) we also have $$\nabla_{\alpha\nu}(\sigma\nu) = (\lie_n\sigma+\sigma\kappa_n)\sigma\nu,$$ 

and hence $q\st{\scri}{=} \lie_n\sigma+\sigma \kappa_n$. With this relation at hand, it is easy to check using \eqref{box} in Appendix \ref{appendix} that the pullback of \eqref{def_s} to $\Omega=0$ is automatically fulfilled. This motivates the following definition of ``universal structure''. This notion is commonly used in the literature, see e.g. \cite{geroch1977asymptotic,ashtekar2015geometry,ciambelli2019carroll,herfray2022tractor}. Here we are just adapting it to the context of metric hypersurface data.


\begin{defi}
	\label{universal}
We say $\{\scri,\bg,\bm\ell,\elltwo,\sigma,\mf{q}\}$ is $\scri$-\textbf{structure data} provided $\{\scri,\bg,\bm\ell,\elltwo\}$ is null metric hypersurface data, $\sigma\in\mc{F}^{\star}(\scri)$, $\mf{q}\in\mc{F}(\scri)$ and $\bU = \mf q\sigma^{-1}\bg$. Moreover, the gauge transformations of $\sigma$ and $\mf q$ are $\mc{G}_{(z,V)}\sigma \d z\sigma$ and $\mc{G}_{(z,V)}\mf q\d\mf q$. Furthermore, we define the conformal transformation $\mc C_{\omega}$ of $\{\scri,\bg,\bm\ell,\elltwo,\sigma,\mf{q}\}$, where $\omega\in\mc{F}^{\star}(\scri)$ is a gauge-invariant function (called conformal factor), by $$\mc{C}_{\omega}\bg =\omega^2\bg, \qquad \mc{C}_{\omega}\bm\ell = \omega^2 \bm\ell,\qquad \mc{C}_{\omega}\elltwo = \omega^2 \elltwo,\qquad \mc{C}_{\omega}\sigma = \omega\sigma,\qquad \mc{C}_{\omega} \mf q = \omega^{-1}\mf q + \omega^{-2}\sigma \lie_n\omega.$$ 

It is straightforward to check that $\mc{C}_{\omega}\circ \mc{G}_{(z,V)} = \mc{G}_{(z,V)}\circ\mc{C}_{\omega}$.
\end{defi}

Note that given $\scri$-structure data and $X\in\X(\scri)$, the $\nabla$-derivative of $\nabla\Omega$ along $X$ at $\scri$ is, by \eqref{connections} and \eqref{derivadannull},
\begin{equation}
	\label{nablaXnablaOmega}
\nabla_X(\nabla\Omega) =  \sigma P^{ab}\U_{bc}X^c e_a + \big(\sigma (\bs-\br)+d\sigma\big)(X) n = \mf q X + \big(\sigma (\bs-\br)+d\sigma-\mf q\bm\ell\big)(X) n.
\end{equation}

Next we define the embedded version of the data to guarantee that $\sigma$ agrees with the proportionality function between $\nu$ and $\nabla\Omega$ at $\scri$, and also that $\mf q$ is the pullback of the scalar $q$ defined in \eqref{def_s}. 

\begin{defi}
	\label{emb_uni}
	We say that $\{\scri,\bg,\bm\ell,\elltwo,\sigma,\mf{q}\}$ is $(\Phi,\xi)$-embedded in $(\mc M,g,\Omega)$ provided that $\{\scri,\bg,\bm\ell,\elltwo\}$ is $(\Phi,\xi)$-embedded in $(\mc M,g)$, $\Omega=0$ in $\Phi(\scri)$, and in addition $\nabla\Omega\st{\scri}{=}\sigma\nu$ and $\mf q = \sigma \kappa_n + \lie_n\sigma$.
\end{defi}

From the conformal transformation of $\mf q$ in Def. \ref{universal} it follows that whenever $\scri$ admits cross-sections one can always find a conformal factor $\omega$ that satisfies $\omega^2\mf q + \sigma \lie_n\omega=0$ and, as a consequence, $\mc{C}_{\omega} \mf q=0$. The remaining conformal freedom is the function $\omega$ at some cross-section. In any of such conformal gauges $\scri$ is totally geodesic, $\bU=0$.
\section{Transverse expansion of the metric}
\label{section_transverse}

In previous works \cite{Mio3,Mio4} we computed the $m$-th transverse derivative of the ambient Ricci tensor in terms of transverse derivatives of the ambient metric $g$ at a general null hypersurface $\mc H$. In this section we quote the results that shall be needed in the rest of the paper. In order to simplify the notation one introduces the tensors $$\mc{K}^{(m)}\d\lie_{\xi}^{(m)}g,\qquad \bY^{(m)}\d \dfrac{1}{2}\Phi^{\star}\mc{K}^{(m)}, \qquad \br^{(m)}\d \bY^{(m)}(n,\cdot), \qquad \kappa^{(m)}\d -\bY^{(m)}(n,n),$$ 

as well as $$\mc{R}^{(m)}\d \Phi^{\star}\big(\lie_{\xi}^{(m-1)}\ric\big),\quad \dot{\mc R}^{(m)}\d \Phi^{\star}\big(\lie_{\xi}^{(m-1)}\ric(\xi,\cdot)\big),\quad \ddot{\mc R}^{(m)}\d \Phi^{\star}\big(\lie_{\xi}^{(m-1)}\ric(\xi,\xi)\big)$$

 for $m\ge 1$. We also denote $\mc{K}^{(0)}\d g$. Note that $\bY^{(1)},\br^{(1)},\kappa^{(1)}$ coincide with $\bY,\br,\kappa_n$, respectively. In what follows we refer to the tensors $\{\bY^{(1)},\bY^{(2)},...\}$ as the \textit{transverse} (or \textit{asymptotic) expansion}. The remaining derivatives of the metric, namely the tensors $\mc{K}^{(m)}(\xi,\cdot)$ at $\mc H$, are given by 
\begin{equation}
	\label{Kxi}
\mc{K}^{(1)}(\xi,e_a) \st{\mc H}{=} \dfrac{1}{2}\nablacero_a \elltwo,\qquad \mc{K}^{(1)}(\xi,\xi)=0,\qquad \mc{K}^{(m)}(\xi,e_a) \st{\mc H}{=} 0,\qquad \mc{K}^{(m)}(\xi,\xi) \st{\mc H}{=} 0 \qquad \forall m\ge 2
\end{equation}
under the assumption $\nabla_{\xi}\xi=0$ \cite{Mio3}. Note that \eqref{Kxi} together with \eqref{inversemetric} imply 
\begin{equation}
	\label{Kup}
	e_b^{\mu}\mc{K}_{\mu}{}^{\rho} \st{\mc H}{=} \left(2P^{cd}\Y_{bc}+\dfrac{1}{2}n^d \nablacero_b\elltwo\right)e_d^{\rho}+2\r_b\xi^{\rho},\qquad \xi^{\mu}\mc{K}_{\mu}{}^{\rho} \st{\mc H}{=} \dfrac{1}{2}P^{ab}\nablacero_a\elltwo e_b^{\rho}+\dfrac{1}{2}n(\elltwo)\xi^{\rho},
\end{equation}
and 
\begin{equation}
	\label{trazas}
\mc{K}^{\mu}{}_{\mu} \st{\mc H}{=} 2\tr_P\bY + n(\elltwo),\qquad \mc{K}^{(m)\mu}{}_{\mu} = 2\tr_P\bY^{(m)} \quad\forall m\ge 2.
\end{equation}
In particular,
\begin{equation}
	\label{Kupnu}
	\nu^{\mu}\mc{K}_{\mu}{}^{\rho} \st{\mc H}{=} \left(2P^{cd}\r_c+\dfrac{1}{2}n(\elltwo)n^d \right)e_d^{\rho}-2\kappa_n\xi^{\rho}.
\end{equation}
Recall also that for any two objects $S$ and $T$, any product of them $S\circledast T$ (including tensor contraction) and any derivative operator $\mc{D}$,
	\begin{equation}
		\label{derivada}
		\mc{D}^{(m)}\big(S\circledast T\big) = \sum_{i=0}^m\binom{m}{i}\big( \mc{D}^{(i)} S\big)\circledast\big( \mc{D}^{(m-i)}T\big).
	\end{equation}
In general, we employ the notation $A^{(m)}\d \lie_{\xi}^{(m-1)}A$, $m\ge 1$, for any tensor $A$. As shown in \cite{Mio3}, the commutator $[\lie_{\xi}^{(m)},\nabla]$ acting on $A$ is given by 
\begin{equation}
	\label{propMarc}
	\begin{aligned}
[\lie_{\xi}^{(m)},\nabla_{\gamma}]A^{\alpha_1\cdots\alpha_q}_{\beta_1\cdots\beta_p} =  \sum_{k=0}^{m-1}\binom{m}{k+1}&\left(\sum_{j=1}^{q}A^{(m-k)}{}^{\alpha_1\cdots\alpha_{j-1}\sigma\alpha_{j+1}\cdots\alpha_q}_{\beta_1\cdots\beta_p}\Sigma^{(k+1)}{}^{\alpha_j}{}_{\sigma\gamma}\right. \\
		&\left.- \sum_{i=1}^{p}A^{(m-k)}{}^{\alpha_1\cdots\alpha_q}_{\beta_1\cdots\beta_{i-1}\sigma\beta_{i+1}\cdots\beta_p}\Sigma^{(k+1)}{}^{\sigma}{}_{\beta_i\gamma}\right),
	\end{aligned}
\end{equation}
where $\Sigma^{\alpha}{}_{\mu\nu} =\dfrac{1}{2}g^{\alpha\beta}\left(\nabla_{\mu}\mc{K}_{\nu\beta} + \nabla_{\nu}\mc{K}_{\mu\beta} - \nabla_{\beta}\mc{K}_{\mu\nu}\right)$ and $\Sigma^{(m)}\d \lie_{\xi}^{(m-1)}\Sigma$. We also define the tensor $\uwh\Sigma_{\alpha\mu\nu}\d g_{\alpha\beta}\Sigma^{\beta}{}_{\mu\nu}$ and $\uwh\Sigma^{(m)}\d \lie_{\xi}^{(m-1)}\uwh\Sigma$, and note that $\uwh\Sigma^{(m)}_{\alpha\mu\nu} \neq g_{\alpha\beta}\Sigma^{(m)}{}^{\beta}{}_{\mu\nu}$. Applying \eqref{propMarc} to $A=df$ one arrives at the following.
	\begin{prop}
	\label{derivadashess}
	Let $\xi\in\X(\mc M)$ and $m\ge 1$ an integer. Then, given any function $f$ the following identity holds
	\begin{equation}
		\label{liehess}
		\lie_{\xi}^{(m)}\nabla_{\alpha}\nabla_{\beta}f = \nabla_{\alpha}\nabla_{\beta}\big(\lie_{\xi}^{(m)}(f)\big) - \sum_{k=0}^{m-1}\binom{m}{k+1} \nabla_{\sigma}\big(\lie_{\xi}^{(k)}(f)\big) \Sigma^{(m-k)}{}^{\sigma}{}_{\beta\alpha}.
	\end{equation}
	As a consequence,
	\begin{equation}
		\label{liebox}
		\begin{aligned}
			\lie_{\xi}^{(m)}\square_g f &= \sum_{i=0}^m \binom{m}{i}\big(\lie_{\xi}^{(i)}g^{\alpha\beta}\big)\left(\nabla_{\alpha}\nabla_{\beta} \big(\lie_{\xi}^{(m-i)}(f)\big) - \sum_{k=0}^{m-i-1}\binom{m-i}{k+1}\nabla_{\sigma}\big(\lie_{\xi}^{(k)}(f)\big) \Sigma^{(m-i-k)}{}^{\sigma}{}_{\beta\alpha}\right).
		\end{aligned}
	\end{equation}
	\begin{proof}
		Using \eqref{propMarc} with $A=df$ as well as $\lie_{\xi} d = d\lie_{\xi}$, so that $A^{(k)}=d(\lie_{\xi}^{(k-1)}(f))$, \eqref{liehess} follows. Identity \eqref{liebox} follows at once from \eqref{derivada} after inserting \eqref{liehess}, i.e. $$\lie_{\xi}^{(m)}\square_g f =\sum_{i=0}^{m}\binom{m}{i}\big(\lie_{\xi}^{(i)}g^{\alpha\beta}\big) \lie_{\xi}^{(m-i)}\nabla_{\alpha}\nabla_{\beta}f.$$ 
	\end{proof}
\end{prop}

One of the main results in \cite{Mio3,Mio4} is the computation of $\mc{R}_{ab}^{(m)}$, $\dot{\mc R}^{(m)}_a$ and $\ddot{\mc{R}}^{(m)}$ to the leading order, namely $\dot{\mc R}^{(m)}_a$ and $\ddot{\mc{R}}^{(m)}$ up to order $m+1$ and $\mc{R}_{ab}^{(m)}$ up to order $m$. In this paper we also need $\dot{\mc R}^{(m)}_a$ up to order $m$. To avoid overloading the body of the paper we postpone the computation to Appendix \ref{appendixB}. The result is as follows.
\begin{prop}[\cite{Mio3} and Prop. \ref{propliericxitang}]
	\label{prop_derivadas}
	Let $\{\mc H,\bg,\bm\ell,\elltwo\}$ be null metric hypersurface data $(\Phi,\xi)$-embedded on $(\mc M,g)$ and extend $\xi$ off $\Phi(\mc H)$ by $\nabla_{\xi}\xi=0$. Let $m\ge 2$ be an integer. Then,
\begin{align}
\ddot{\mc R}^{(m)} &= - \tr_P\bY^{(m+1)} + \mc{O}^{(m)},\label{ddotR}\\
\dot{\mc R}^{(m)}_a &= \r^{(m+1)}_a + P^{bc} \nablacero_b\Y^{(m)}_{ac} -2m P^{bc}\r_b \Y^{(m)}_{ac} - \nablacero_a\big(\tr_P\bY^{(m)}\big)  + \big(\tr_P\bY^{(m)}\big)(\r_a-\s_a) \nonumber\\
&\quad\,  + \big(P^{bc}\Y_{ab}- 3V^c{}_a\big)\r^{(m)}_c + \left(\tr_P\bY-\dfrac{m}{2}n(\elltwo)\right)\r^{(m)}_a +\mc{O}_a^{(m-1)},\label{dotR}\\
\mc{R}^{(m)}_{ab} &= -2\lie_n\Y^{(m)}_{ab} - \left(2m\kappa_n + \tr_P\bU\right)\Y^{(m)}_{ab}- (\tr_P\bY^{(m)})\U_{ab} + 4P^{cd}\U_{c(a}\Y^{(m)}_{b)d} \nonumber\\
&\quad\,    +4(\s-\r)_{(a} \r^{(m)}_{b)}+ 2\nablacero_{(a}\r^{(m)}_{b)} -2\kappa^{(m)}\Y_{ab} + \mc{O}^{(m-1)}_{ab},\label{R}
\end{align}	
where $\mc{O}^{(m)}$, $\mc{O}^{(m)}_a$ and $\mc{O}^{(m)}_{ab}$ are, respectively, a scalar, a one-form and a (0,2) symmetric tensor depending only on metric data $\{\bg,\bm\ell,\elltwo\}$ and $\{\bY,...,\bY^{(m)}\}$. Moreover,
		\begin{align}
		\hskip-0.3cm	\mc{R}^{(m)}_{ab} n^b &= -\lie_n\r^{(m)}_a - \big(2(m-1)\kappa_n + \tr_P\bU\big)\r^{(m)}_a-\nablacero_a\kappa^{(m)} + \mc{O}^{(m-1)}_{ab}n^b,\label{elierictangn}\\
		\hskip-0.3cm		P^{ab}\mc{R}^{(m)}_{ab} &= -2\lie_n\big(\tr_P\bY^{(m)}\big) - 2\left(m\kappa_n + \tr_P\bU\right)\tr_P\bY^{(m)} + 2\kappa^{(m)}\big(n(\elltwo)-\tr_P\bY\big) \nonumber\\
				&\quad\, -4P\big(\br+\bs,\br^{(m)}\big) + 2\div_P\br^{(m)} + P^{ab}\mc{O}^{(m-1)}_{ab},	\label{ePcontractioneq}\\
	\hskip-0.3cm\dot{\mc R}_a^{(m)}n^a &= -\kappa^{(m+1)} + \div_P\br^{(m)} - P^{ab}P^{cd} \U_{bd}\Y^{(m)}_{ac}  - \lie_n\big(\tr_P\bY^{(m)}\big)- \kappa_n \tr_P\bY^{(m)}\nonumber\\
					& \quad\, - 2P\big((m+1)\br+2 \bs,\br^{(m)}\big) -\left(\tr_P\bY-\dfrac{m+3}{2}n(\elltwo)\right)\kappa^{(m)} + \mc{O}_a^{(m-1)} n^a.\label{dotRn0}
			\end{align}
\end{prop}
The corresponding expressions in the case $m=1$, i.e. for $\mc{R}^{(1)}_{ab}$, $\dot{\mc R}^{(1)}_a$ and $\ddot{\mc R}^{(1)}$, are as follows \cite{miguel3,Mio3}
\begin{align}
\mc R^{(1)}_{ab}& =\mc R_{ab}= \accentset{\circ}{R}_{(ab)} -2\lie_n \Y_{ab} - (2\kappa_n+\tr_P\bU)\Y_{ab} + \nablacero_{(a}\left(\s_{b)}+2\r_{b)}\right)\nonumber\\
&\qquad\qquad -2\r_a\r_b + 4\r_{(a}\s_{b)} - \s_a\s_b - (\tr_P\bY)\U_{ab} + 2P^{cd}\U_{d(a}\left(2\Y_{b)c}+\F_{b)c}\right),\label{constraint}\\
\mc{\ddot R}^{(1)} & =\mc{\ddot R}= -P^{ab}\Y^{(2)}_{ab} + P^{ab}P^{cd}(\Y+\F)_{ac}(\Y+\F)_{bd} + P^{ab}(\r-\s)_a \nablacero_b\elltwo,\label{ricxixi}\\
\mc{\dot R}^{(1)}_c  & = \mc{\dot R}_c=\r^{(2)}_c - P^{ab}A_{abc}  - P^{ab}(\r+\s)_a(\Y+\F)_{cb} +\dfrac{1}{2}\kappa_n \nablacero_c\elltwo -\dfrac{1}{2} n(\elltwo)(\r-\s)_{c} ,\label{ricxiX}
\end{align}
where $\Rcero_{ab}$ is the Ricci tensor of $\nablacero$, while $A$ is given by $$A_{bcd}\d 2\nablacero_{[d}\F_{c]b} + 2\nablacero_{[d}\Y_{c]b} + \U_{b[d}\nablacero_{c]}\elltwo + 2\Y_{b[d}(\r-\s)_{c]}$$ 

and satisfies (the third equality is a result established in \cite{miguel3})
\begin{align*}
	P^{ab}A_{abc}n^c & = \dfrac{1}{2}P^{ab} (A_{abc}+A_{bac})n^c=-\dfrac{1}{2}P^{ab} (A_{acb}+A_{bca})n^c\\
	&= P^{ab}\left(\lie_n\Y_{ab}-\nablacero_a(\r+\s)_b+\kappa_n\Y_{ab}+(\r-\s)_a(\r-\s)_b-\dfrac{1}{2}n(\elltwo)\U_{ab}-P^{cd}\Y_{ac}\U_{bd}\right)\\
	&\st{\eqref{lietrPY}}{=} \lie_n\big(\tr_P\bY\big) +P^{ab}P^{cd}\U_{ac}\Y_{bd}  -\div_P(\br+\bs) +\big(\tr_P\bY-n(\elltwo)\big)\kappa_n \\
	&\quad\, +P(\br+\bs,\br+\bs) - \dfrac{1}{2}n(\elltwo)\tr_P\bU.
\end{align*}
Therefore, the contraction of \eqref{ricxiX} with $n^c$ is
\begin{align}
\dot{\mc R}^{(1)}_a n^a & = -\kappa^{(2)} - P^{ab}A_{abc}n^c - P(\br+\bs,\br+\bs)+ n(\elltwo)\kappa_n\nonumber\\
&=- \kappa^{(2)} - \lie_n\big(\tr_P\bY\big)+\div(\br+\bs) +\big(2n(\elltwo)-\tr_P\bY\big)\kappa_n + \dfrac{1}{2}\big(\tr_P\bU\big)  n(\elltwo)\nonumber\\
&\quad\, - P^{ad}P^{bc}\U_{ab}\Y_{cd} -4P(\br,\bs) - 2P(\br,\br) - 2P(\bs,\bs). \label{dotRn}
\end{align}
The contractions of $\mc{R}^{(1)}_{ab}$ with $P$ and $n$ are \cite{miguel3}
\begin{align}
	\hspace{-0.5cm}	\mc R_{ab}^{(1)}n^a & = -\lie_n(\r_b-\s_b) - \nablacero_b \kappa_n - (\tr_P\bU) (\r_b-\s_b)  - \nablacero_b (\tr_P\bU) + P^{cd}\nablacero_c\U_{bd} - 2P^{cd}\U_{bd}s_c,\label{constraintn}\\
	\hspace{-0.5cm}	\mc R^{(1)}_{ab}n^an^b & = -n(\tr_P\bU) + (\tr_P\bU)\kappa_n - P^{ab}P^{cd}\U_{ac}\U_{bd},\label{constraintnn}\\
	\hspace{-0.5cm}	P^{ab}{\mc R}^{(1)}_{ab} &= \tr_P \Rcero -2\lie_n(\tr_P\bY)  - 2\big(\kappa_n+\tr_P\bU\big)\tr_P\bY+ \div_P(\bs + 2\br) -2P(\br,\br)\nonumber\\
	\hspace{-0.5cm}	&\quad\, -4P(\br,\bs)-P(\bs,\bs)+2\kappa_n n(\elltwo).\label{trPR}
\end{align}
Expressions \eqref{trPR} and \eqref{dotRn} are all one needs to compute the ambient scalar curvature at $\mc H$. Indeed, using \eqref{inversemetric}, $R^{(1)} =R\st{\mc H}{=} \tr_P\mc{R}  +2\dot{\mc R}_a n^a$, and inserting \eqref{trPR} and \eqref{dotRn} one concludes
\begin{equation}
	\label{scal}
	\begin{aligned}
		\hskip -0.2cm R^{(1)}=R& = -2\kappa^{(2)} - 4\lie_n(\tr_P\bY) -2\big(2\kappa_n+\tr_P\bU\big)\tr_P\bY +3\div\bs + 4\div\br- 5P(\bs,\bs)\\
		&\quad\,  -6P(\br,\br)-12P(\br,\bs) -2P^{ab}P^{cd}\U_{ac}\Y_{bd}  +\big(\tr_P\bU+6\kappa_n\big) n(\elltwo) + \tr_P\Rcero .
	\end{aligned}
\end{equation}
Before computing the higher order derivatives of the scalar curvature $R^{(m)}$, $m\ge 2$, we recall the following notation from \cite{Mio3,Mio4}.
\begin{nota}
	\label{nota1}
Let $(\mc M,g)$ be a semi-Riemannian manifold and $\mc H$ a hypersurface. Given two ambient tensors $T$ and $S$, the notation $T \st{[m]}{=} S$ means that $T-S$ does not depend on derivatives of $g$ of order $m$ or higher. When $T$ and $S$ are tensors on $\mc H$, we use $T \st{(m)}{=} S$ to denote that $T-S$ does not depend on \textbf{transverse} derivatives of $g$ at $\mc H$ of order $m$ or higher. Clearly, for two ambient tensors $T$ and $S$ satisfying $T \st{[m]}{=} S$, their pullbacks to $\mc H$ satisfy $\Phi^{\star}T \st{(m)}{=} \Phi^{\star} S$.
\end{nota}
Applying $\lie_{\xi}^{(m-1)}$ to $g^{\alpha\beta}R_{\alpha\beta}$ for $m\ge 2$ and using identity \eqref{derivada} and $\lie_{\xi}g^{\alpha\beta} = -g^{\alpha\mu}g^{\beta\nu}\mc{K}_{\mu\nu}$ one obtains $$R^{(m)} \st{[m]}{=} g^{\alpha\beta} R^{(m)}_{\alpha\beta} - (m-1) g^{\alpha\mu}g^{\beta\nu}\mc{K}_{\mu\nu} R^{(m-1)}_{\alpha\beta} \st{[m]}{=} P^{ab}\mc{R}^{(m)}_{ab} + 2\dot{\mc R}^{(m)}_a n^a - (m-1) g^{\alpha\mu}g^{\beta\nu}\mc{K}_{\mu\nu} R^{(m-1)}_{\alpha\beta},$$ 

where in the second equality we inserted \eqref{inversemetric}. The first two terms are \eqref{ePcontractioneq} and \eqref{dotRn0}. Concerning the third one, we note that \eqref{R} implies that terms of the form $\bY^{(m)}$ and $\bY^{(m+1)}$ can only appear when the tensor $R^{(m-1)}_{\alpha\beta}$ is contracted with $\xi$ at least once, i.e. (cf. \eqref{inversemetric}) 
\begin{align*}
g^{\alpha\mu}g^{\beta\nu}\mc{K}_{\mu\nu} R^{(m-1)}_{\alpha\beta} &\st{(m)}{=} \big(2P^{ac}e_c^{\mu}\nu^{\nu}e_a^{\alpha}\xi^{\beta} + 2\xi^{\mu}\nu^{\nu}\nu^{\alpha}\xi^{\beta} + \nu^{\mu}\nu^{\nu}\xi^{\alpha}\xi^{\beta}\big)\mc{K}_{\mu\nu} R^{(m-1)}_{\alpha\beta}\\
& \st{(m)}{=} 4P^{ab}\r_a\dot{\mc R}_b^{(m-1)}+n(\elltwo) \dot{\mc R}^{(m-1)}_a n^a - 2\kappa_n \ddot{\mc R}^{(m-1)}\\
& \st{(m)}{=} 4P^{ab}\r_a\r^{(m)}_b - n(\elltwo) \kappa^{(m)} +2\kappa_n\tr_P\bY^{(m)},
\end{align*}
where in the second line we used $\mc{K}_{ab}=2\Y_{ab}$ and $\mc{K}_{\mu\nu}\xi^{\mu}e^{\nu}_a = \frac{1}{2}\nablacero_a\elltwo$ (see \eqref{Kxi}), and in the third line \eqref{dotR}, \eqref{dotRn0} and \eqref{ddotR}. Finally, inserting this, \eqref{ePcontractioneq} and \eqref{dotRn0} into the expression of $R^{(m)}$, we arrive at
\begin{equation}
	\label{scalar}
	\begin{aligned}
		R^{(m)}& \st{(m)}{=}  -2\kappa^{(m+1)} -4\lie_n\big(\tr_P\bY^{(m)}\big) - 2\big(2m\kappa_n + \tr_P\bU\big)\tr_P\bY^{(m)} - 2P^{ab}P^{cd}\U_{ac}\Y^{(m)}_{bd}\\
		&\quad\, +4\div_P\br^{(m)}-4P\big((2m+1)\br+3\bs,\r_b^{(m)}\big)+2\big((m+2)n(\elltwo)-2\tr_P\bY\big)\kappa^{(m)}.
	\end{aligned}
\end{equation}
An immediate consequence is 
\begin{equation}
	\label{scalarm+1}
	R^{(m)} \st{(m+1)}{=}  -2\kappa^{(m+1)}.
\end{equation}
\begin{rmk}
	\label{remarkricci}
	With the notation introduced in Notation \ref{nota1} it is clear from \eqref{ddotR}, \eqref{dotR}, \eqref{R} and \eqref{scalar} that $\mc{R}_{ab}^{(m)}\st{(m+1)}{=}0$, $\dot{\mc R}_a^{(m)}\st{(m+2)}{=}0$, $\ddot{\mc R}^{(m)}\st{(m+2)}{=}0$ and $R^{(m)}\st{(m+2)}{=}0$.
\end{rmk}

Next we recall the general existence and uniqueness results of \cite{Mio3,Mio4} that will be needed below. Proposition \ref{prop_diffeo} quotes the result in \cite{Mio3} where we constructed a diffeomorphism between neighbourhoods of two diffeomorphic hypersurfaces. Proposition \ref{teo_iso0} quotes a result in \cite{Mio3} where this diffeomorphism was used to derive a set of sufficient conditions on two diffeomorphic hypersurfaces such that there exists an isometry between two neighbourhoods of them. Finally, in Theorem \ref{borel} we recall a result from \cite{Mio4} where an ambient spacetime was constructed from the would-be expansion at a null hypersurface.

\begin{prop}
	\label{prop_diffeo}
	Let $\Phi:\mc H\hookrightarrow\mc M$ and $\Phi':\mc H'\hookrightarrow\mc M'$ be two embedded hypersurfaces in ambient manifolds $(\mc M,g)$ and $(\mc M',g')$ and let $\xi$, $\xi'$ be respectively riggings of $\Phi(\mc H)$, $\Phi'(\mc H')$ extended geodesically. Assume that there exists a diffeomorphism $\chi:\mc H\to\mc H'$. Then, there exist open neighbourhoods $\mc U\subset \mc M$ and $\mc U'\subset\mc M'$ of $\Phi(\mc H)$ and $\Phi'(\mc H')$ and a unique diffeomorphism $\Psi:\mc U\to\mc U'$ satisfying $\Psi_{\star}\xi = \xi'$ and $\Phi'\circ \chi = \Psi\circ \Phi$. 
\end{prop}
\begin{prop}
	\label{teo_iso0}
	Let $\{\H,\bg,\bm\ell,\elltwo\}$ (respectively $\{\H',\bg',\bm\ell',\elltwo{}'\}$) be null metric hypersurface data $(\Phi,\xi)$-embedded in $(\mc M,g)$ (resp. $(\Phi',\xi')$-embedded in $(\mc M',g')$) with $\xi$ and $\xi'$ extended geodesically. Assume that there exists a diffeomorphism $\chi:\H\to\H'$ such that $$\chi^{\star}\{\H',\bg',\bm\ell',\elltwo{}'\}\d \{\chi^{\star}\mc H',\chi^{\star}\bg',\chi^{\star}\bm\ell',\chi^{\star}\elltwo{}'\} = \{\H,\bg,\bm\ell,\elltwo\}$$
	
	 and $\chi^{\star}\bY^{(k)}{}' = \bY^{(k)}$ for every $k\ge 1$. Then, there exist neighbourhoods $\mc U\subset \mc M$ and $\mc U'\subset\mc M'$ of $\Phi(\mc H)$ and $\Phi'(\mc H')$ and a diffeomorphism $\Psi:\mc U\to\mc U'$ satisfying $\Psi_{\star}\xi = \xi'$ and $\Phi'\circ \chi = \Psi\circ \Phi$ (as in Proposition \ref{prop_diffeo}) such that 
	\begin{equation}
		\label{isometry0}
		\Psi^{\star} \lie_{\xi'}^{(i)}g' \st{\mc H}{=} \lie_{\xi}^{(i)}g
	\end{equation} 
	for every $i\in\mathbb{N}\cup\{0\}$.
\end{prop}
\begin{teo}
	\label{borel}
	Let $\{\mc H,\bg,\bm\ell,\elltwo\}$ be null metric hypersurface data and $\{\bcY^{(k)}\}_{k\ge 1}$ a sequence of $(0,2)$ symmetric tensor fields on $\mc H$. Then there exists a semi-Riemannian manifold $(\mc M,g)$, an embedding $\Phi:\mc H\hookrightarrow\mc M$ and a rigging vector $\xi$ satisfying $\nabla_{\xi}\xi = 0$ on $\mc M$ such that (i) $\{\mc H,\bg,\bm\ell,\elltwo\}$ is null metric hypersurface data $(\Phi,\xi)$-embedded in $(\mc M,g)$ and (ii) $\{\bcY^{(k)}\}_{k\ge 1}$ is the transverse expansion of $g$ at $\Phi(\mc H)$, i.e. $\bcY^{(k)} = \bY^{(k)}\d \frac{1}{2}\Phi^{\star}\big(\lie^{(k)}_{\xi}g\big)$ for every $k\ge 1$.
\end{teo}

We emphasize that the ambient spacetime $(\mc M,g)$ in this theorem need not to be analytic (i.e. the series need not to converge). An important ingredient in the proof of this theorem is Borel's Lemma \cite{golubitsky2012stable}, which we shall also need later in the following specific form.

\begin{lema}[Borel]
\label{Borel_lemma_for_functions}
Let $\mc M$ be a smooth manifold, $\mc H\hookrightarrow\mc M$ an embedded smooth hypersurface with rigging $\xi$ and $\{\sigma^{(k)}\}_{k\ge 0}$ a collection of functions on $\mc H$. Then, there exists a function $\Omega$ in a neighbourhood of $\mc H$ in $\mc M$ such that $\lie_{\xi}^{(k)}\Omega |_{\mc H} = \sigma^{(k)}$ for every $k\ge 0$.
\end{lema}

The main idea of the present paper is to analyze how the quasi-Einstein equations fix the geometry at null infinity, i.e. how the transverse expansion of the metric is constrained when the conformal Einstein equations are imposed at $\scri$ to infinite order. To analyse this problem it is convenient to introduce the following tensors
\begin{align}
	\mc{Q}_{\alpha\beta}&\d (\mf n-1)\big(\nabla_{\alpha}\nabla_{\beta}\Omega + \Omega L_{\alpha\beta}\big) = (\mf n-1)\nabla_{\alpha}\nabla_{\beta}\Omega + \Omega\left(R_{\alpha\beta}-\dfrac{R}{2\mf n}g_{\alpha\beta}\right),	\label{defC}\\
	\mc{L}_{\alpha} &\d  (\mf n-1) L_{\alpha\beta}\nabla^{\beta}\Omega = R_{\alpha\beta}\nabla^{\beta}\Omega - \dfrac{R}{2\mf n} \nabla_{\alpha}\Omega,	\label{mcL}
\end{align}
and the function $f\d|\nabla\Omega|^2$. These tensors are obviously related by 
\begin{equation}
	\label{QyL}
{\mc{Q}}_{\alpha\beta}\nabla^{\beta}\Omega = \dfrac{\mf n-1}{2}\nabla_{\alpha}f +\Omega{\mc{L}}_{\alpha}.
\end{equation}
Less immediate is the identity
\begin{equation}
	\label{Bianchi}
\nabla_{\mu}\big(\tr_g\Q\big) + \mf{n}{\mc{L}}_{\mu} =  \nabla_{\rho}\Q^{\rho}{}_{\mu},
\end{equation}
which follows from the Bianchi identity \eqref{nablas} after replacing $q=\frac{\tr_g\Q}{\mf n^2-1}$ and $(\mf n-1)\mc T = \Q-\frac{\tr\Q}{\mf n+1}g$. 
In accordance of the general notation of this paper, we introduce the tensors $$\mc{Q}^{(m)}_{\alpha\beta}\d \lie_{\xi}^{(m-1)}\mc{Q}_{\alpha\beta}, \qquad \mc{L}^{(m)}_{\alpha} \d \lie_{\xi}^{(m-1)}\mc{L}_{\alpha}, \qquad f^{(m)}\d \lie_{\xi}^{(m-1)}f $$ 

and
\begin{align}
\mc{Q}_{ab}^{(m)} &\d \big(\Phi^{\star}\mc{Q}^{(m)}\big)_{ab},& \dot{\mc Q}^{(m)}_a &\d \big(\Phi^{\star}\mc{Q}^{(m)}(\xi,\cdot)\big)_a,& \ddot{\mc Q}^{(m)}&\d \Phi^{\star}\big(\mc{Q}^{(m)}(\xi,\xi)\big),\label{Cm}\\
\mc{L}^{(m)}_a &\d \big(\Phi^{\star}\mc{L}^{(m)}\big)_a,& \dot{\mc L}^{(m)} &\d \Phi^{\star}\big(\mc{L}^{(m)}(\xi)\big).&&	\label{Lm}
\end{align}
We also denote the transverse derivatives of $\Omega$ at $\scri$ by $\sigma^{(k)}\d \lie_{\xi}^{(k)}\Omega|_{\scri}$. Note that $\sigma^{(1)}$ agrees with the function $\sigma$ introduced before, so we shall use both symbols indistinctly. As already mentioned this function cannot vanish anywhere on $\scri$.\\

The idea now is to impose the conformal equations to infinite order at $\scri$ to see how $\{\bY^{(k)}\}_{k\ge 1}$ and $\{\sigma^{(k)}\}_{k\ge 1}$ are constrained. Clearly, this depends on how the conformal factor $\Omega$ has been fixed. One sensible choice commonly used in the literature is to require the transformed Ricci scalar to vanish, which amounts to solving a wave equation of the form $\square_g\Omega = \Omega F$ for some function $F$, see \cite{friedrich2002conformal}. Another interesting possibility is to fix $\Omega$ by solving $|\nabla\Omega|^2=0$. We next quote a result from \cite{Mio6} where we showed that this conformal gauge always exists locally near $\scri$ and depends on a free function on a hypersurface transverse to $\scri$.
\begin{lema}
	\label{lema_conformalgauge}
Let $(\mc M,g,\Omega)$ be a conformal manifold with $\lambda=0$ and $\mc H$ a hypersurface transverse to $\scri$. Let $\omega_0$ be a non-vanishing function on $\mc H$. Then, there exists a unique $(\wt g=\omega^2g,\wt \Omega=\omega\Omega)$ in a neighbourhood of $\scri\cap\mc H$ such that $|\wt{\nabla}\wt\Omega|^2_{\wt g}=0$ and $\wt\Omega\st{\mc H}{=}\omega_0\Omega$. 
\end{lema}
In either of the two choices mentioned above, one can verify that $ \scri $ is totally geodesic, which, in terms of hypersurface data means $ \bU = 0 $. However, in a generic conformal gauge, $ \scri $ is only totally umbilical. In this paper we shall be mostly concerned with the choice $|\nabla\Omega|^2=0$, so in Proposition \ref{propQLf} below we write down the tensors $ \Q^{(m)}_{ab} $, $ \dot{\Q}{}^{(m)}_a $, $\ddot{\Q}{}^{(m)} $, $ \mc{L}^{(m)}_a $, $ \mc{\dot L}^{(m)} $, and $ f^{(m)} $ under the assumption $ \bU = 0 $. Each computation is performed up to the order that it will be needed. Given that other conformal choices are possible, in Appendix \ref{appendixC} we compute the expressions in full generality, i.e. without the assumption $\bU=0$. The formulae in Prop. \ref{propQLf} are simply their particularization to $\bU=0$. We also extend the meaning of $\st{(m)}{=}$ in Notation \ref{nota1} as follows.
\begin{nota}
Let $(\mc M,g,\Omega)$ be a conformal manifold with null infinity $\scri$ and $T$, $S$ two tensors on $\scri$. We use $T \st{(m)}{=} S$ to denote that $T-S$ does not depend on transverse derivatives of $g$ and $\Omega$ at $\scri$ of order $m$ or higher.
\end{nota}
\begin{prop}
	\label{propQLf}
Let $\scri=\{\Omega=0\}$ be $(\Phi,\xi)$-embedded in $(\mc M,g)$ and extend $\xi$ off $\Phi(\scri)$ geodesically. Assume $\bU=0$. Then, for every $m\ge 2$,
\begin{align}
\Q_{ab}^{(m+1)} &\st{(m)}{=} (\mf n-1-2m)\sigma^{(1)}\lie_n\Y^{(m)}_{ab} + m\left((\mf n-1)(\lie_n\sigma+\sigma\kappa_n) + (\mf n-1-2m)\sigma^{(1)}\kappa_n  \right)\Y_{ab}^{(m)} \nonumber\\
&\qquad  +\dfrac{m\sigma^{(1)}}{\mf n}\Big(\kappa^{(m+1)}+2\lie_n\big(\tr_P\bY^{(m)}\big) + 2m\kappa_n\tr_P\bY^{(m)}  \Big) \gamma_{ab}\nonumber\\
&\qquad +\dfrac{m(m-1)}{2\mf n}\sigma^{(2)}\kappa^{(m)}\gamma_{ab}+\wt{\mathscr R}^{(m)}_{ab},\label{Qabm}\\
\dot\Q_{a}^{(m+1)} &\st{(m+1)}{=} (\mf n-1)\big(\nablacero_a \sigma^{(m+1)} - \sigma^{(m+1)}(\r-\s)_a  \big)- (\mf n-1-m)\sigma^{(1)}\r_a^{(m+1)} + \dfrac{m\sigma^{(1)}}{\mf n}\kappa^{(m+1)}\ell_a ,\label{dotQam}\\
\ddot\Q^{(m+1)} & \st{(m+1)}{=} (\mf n-1)\sigma^{(m+2)} - m\sigma^{(1)}\tr_P\bY^{(m+1)} + \dfrac{m\sigma^{(1)}\elltwo}{\mf n}\kappa^{(m+1)},\label{ddotQm}\\
\mc{L}^{(m)}_a  &\st{(m)}{=} - \sigma^{(1)} \lie_n \r_a^{(m)} + (m-1) (\lie_n\sigma^{(1)}) \r_a^{(m)} - \sigma^{(1)}\nablacero_a \kappa^{(m)}  + \dfrac{(m-1)\kappa^{(m)}}{\mf n}\nablacero_a\sigma^{(1)} ,\label{Lma}\\
\mf n \dot{\mc{L}}^{(m)} &\st{(m)}{=} -(\mf n-1)\sigma^{(1)}\kappa^{(m+1)} -(\mf n-2)\sigma^{(1)} \lie_n\big(\tr_P\bY^{(m)}\big) \nonumber \\
&\qquad  +\left(\big(2m(1-\mf n)+\mf n\big)\sigma^{(1)}\kappa_n -(m-1)\mf{n} \lie_n\sigma^{(1)} \right)\tr_P\bY^{(m)} + \dot{\mathscr R}^{(m)}_{\mc L},\label{dotLm0}\\
f^{(m+1)}&\st{(m)}{=} 2\sigma^{(1)} \big(\lie_n\sigma^{(m)} + \sigma^{(1)} \kappa^{(m)}\big)+2m\big(\lie_n\sigma^{(1)}+2\sigma^{(1)}\kappa_n\big) \sigma^{(m)},\label{fm0}
\end{align}
where $\wt{\mathscr R}^{(m)}_{ab}$, $\dot{\mathscr R}^{(m)}_{\mc L}$ are tensors that depend on $\br^{(m)}$, $\sigma^{(m)}$ and lower order terms and we do not write for simplicity (they can be easily read out by performing all the calculations in their respective proofs).
\begin{proof}
The first three expressions are the particularization to the case $\bU=0$ of Proposition \ref{propQ}, the next two of Proposition \ref{propL} and the last one of Proposition \ref{propf}.
\end{proof}
\end{prop}
The tensors $\L^{(1)}_{\mu}$, $\Q^{(1)}_{\alpha\beta}$, $\Q^{(2)}_{\alpha\beta}$ and $f^{(2)}$ are not covered by this proposition. Again under the assumption $\bU=0$ they are given by (see \eqref{Q1general}-\eqref{ddQ2general}, \eqref{dotLgeneral}-\eqref{La1general} and \eqref{f2general} in Appendix \ref{appendixC})
\begin{equation}
	\label{Q1}
\Q_{ab}^{(1)} =0,\qquad	\dot{\Q}_a^{(1)} = \nablacero_a\sigma^{(1)} -\sigma^{(1)}(\r-\s)_a,\qquad  \ddot{\Q}^{(1)} = \sigma^{(2)} ,
\end{equation}
\begin{align}
	\Q^{(2)}_{ab} & = (\mf n-3)\sigma^{(1)} \lie_n\Y_{ab} +\big((\mf n-1)(\lie_n\sigma+\sigma\kappa_n)+(\mf n-3)\sigma^{(1)}\kappa_n \big)\Y_{ab} + (\mf n-1)\nablacero_a\nablacero_b\sigma^{(1)}\nonumber\\
	&\quad\, -2\sigma^{(1)}(\mf n-2)\nablacero_{(a}\r_{b)} +\sigma^{(1)}\big(\nablacero_{(a}\s_{b)}-2\r_a\r_b+4\r_{(a}\s_{b)}-\s_a\s_b+ \accentset{\circ}{R}_{(ab)} \big)  - \dfrac{\sigma R}{2\mf n} \gamma_{ab},\label{Q2ab}\\
	\dot{\Q}^{(2)}_a &=  -(\mf n-2)\sigma^{(1)}\r_a^{(2)} - (\mf n-1) V_a{}^b\nablacero_b\sigma^{(1)} - \dfrac{1}{2}\sigma^{(1)}(\mf n-2) n(\elltwo) (\r-\s)_a+2\sigma^{(1)}(\mf n-1) V^c{}_a \r_c\nonumber\\
	&\quad\,-\sigma^{(1)}P^{bc}A_{bca} -\sigma^{(1)}P^{cb}(\r+\s)_c(\Y_{ab}+\F_{ab}) + \dfrac{1}{2}\sigma^{(1)}\kappa_n \nablacero_a\elltwo-\dfrac{\sigma^{(1)}R}{2\mf n}\ell_a,\label{dotQ2a}\\
	\ddot{\Q}^{(2)} &= (\mf n-1)\sigma^{(2)} - \sigma^{(1)}\tr_P\bY^{(2)} + \sigma^{(1)}P^{ab}P^{cd}(\Y+\F)_{ac}(\Y+\F)_{bd} + \sigma^{(1)}P(\br-\bs,d\elltwo) \nonumber\\
	&\quad\, -\dfrac{R}{2\mf n}\sigma^{(1)}\elltwo,\label{ddotQ2}
\end{align}
\begin{align}
	\dfrac{2\mf n}{\sigma^{(1)}} \dot{\mc{L}}^{(1)} &= -2(\mf n-1)\kappa^{(2)} -2(\mf n-2)\lie_n\big(\tr_P\bY\big)- 2(\mf n-2) \kappa_n \tr_P\bY  + 2(\mf n-2)\div\br  \nonumber\\
	&\quad\,+(2\mf n-3)\div\bs +2(2\mf n-3)\kappa_nn(\elltwo)  - 2(2\mf n-3)P(\br,\br) -(4\mf n-5)P(\bs,\bs) \nonumber\\
	&\quad\,- 4(2\mf n-3)P(\br,\bs)- \tr_P\Rcero,\label{dotL}\\
	\dfrac{1}{\sigma^{(1)}} \mc{L}_a^{(1)} &= -\lie_n(\r_b-\s_b) - \nablacero_b \kappa_n,\label{La1}
\end{align}
where $R$ is explicitly given by \eqref{scal}. In addition (cf. \eqref{f2general})
\begin{equation}
	\label{f2}
	f^{(2)} =2\sigma(\lie_n\sigma+\sigma\kappa_n).
\end{equation}
Although $f^{(3)}$ is covered in Proposition \ref{propQLf}, we will need shortly its explicit expression (cf. \eqref{f3general})
\begin{equation}
	\label{f3}
	\begin{aligned}
f^{(3)} &=2\sigma^{(1)}\big(\sigma^{(1)}\kappa^{(2)} + \lie_n\sigma^{(2)}\big) + 2\big(2\sigma^{(2)}-\sigma^{(1)}n(\elltwo)\big)\big(\lie_n\sigma^{(1)}+2\sigma^{(1)}\kappa_n\big)\\
&\quad\,  +8\sigma^{(1)}P\big(\br,\sigma^{(1)}\br-d\sigma^{(1)}\big) + 2P^{ab}\nablacero_a\sigma^{(1)}\nablacero_b\sigma^{(1)}.
	\end{aligned}
\end{equation}

One immediate consequence of \eqref{Q1general} is that, for every embedded $\scri$-structure data satisfying $\dot\Q_a = \mf q \ell_a$, the $\nabla_X$ derivative of $\nabla\Omega$ at $\scri$ is given by (see \eqref{nablaXnablaOmega}) $$\nabla_X(\nabla\Omega) \st{\scri}{=} \mf q X, $$ 

and therefore every $\scri$-structure data written in a conformal gauge in which $\mf q=0$ is in particular a weakly isolated horizon (see \cite{ashtekar2024null}).\\ 

In the conformal gauge in which $|\nabla\Omega|^2=0$, the ``higher order conformal equations'' that we must solve are $\Q^{(m)}_{ab}=0$, $\dot{\Q}{}^{(m)}_a=0$, $\ddot{\Q}{}^{(m)}=0$, $\mc{L}^{(m)}_a=0$, $\mc{\dot L}^{(m)}=0$ and $f^{(m)}=0$ for every $m\ge 1$. Observe that, for each $m$, there are $2(\mf n+1)$ equations more than components of the metric to be fixed ($\bY^{(m)}$ and $\sigma^{(m)}$). This overdeterminacy is related to the identities \eqref{QyL} and \eqref{Bianchi}. Taking $m$ transverse derivatives in \eqref{QyL} and applying \eqref{propMarc} we arrive at
\begin{equation}
	\label{relationQLfm}
	\sum_{k=0}^m \binom{m}{k} \Q^{(m+1-k)}_{\alpha\beta}\sum_{j=0}^{k} \binom{k}{j}( \lie_{\xi}^{(j)}g^{\beta\mu}) \nabla_{\mu} \lie_{\xi}^{(k-j)}\Omega = \dfrac{\mf n-1}{2}\nabla_{\alpha} \lie_{\xi}^{(m)}f + \sum_{k=0}^m \binom{m}{k} (\lie_{\xi}^{(k)}\Omega) \mc{L}^{(m-k+1)}_{\alpha},
\end{equation}
and applying $\lie_{\xi}^{(m-1)}$ to \eqref{Bianchi} gives
\begin{equation}
	\label{bianchim}
	\begin{aligned}
 \left(\nabla_{\rho}\Q^{(m)}{}^{\rho}{}_{\mu} + \sum_{k=0}^{m-2}\binom{m-1}{k+1}\left(\Q^{(m-1-k)}{}^{\sigma}{}_{\mu}\Sigma^{(k+1)}{}^{\rho}{}_{\rho\sigma} - \Q^{(m-1-k)}{}^{\rho}{}_{\sigma}\Sigma^{(k+1)}{}^{\sigma}{}_{\rho\mu}\right)\right)=\\
		=\nabla_{\mu}\left(\sum_{k=0}^{m-1} \binom{m-1}{k} (\lie^{(k)}_{\xi} g^{\alpha\beta}) \Q^{(m-k)}_{\alpha\beta}\right) + \mf n \L_{\mu}^{(m)}.
	\end{aligned}
\end{equation}

In the next two lemmas we prove some direct consequences of these identities. They will be essential in Section \ref{scri} to show, order by order, that the set of $2(\mf n+1)$ redundant equations is automatically satisfied provided the remaining equations hold. This is analogous (though considerably more involved) to Proposition 4.5 in \cite{Mio4}. Each lemma comes with its own general hypothesis, namely the validity of the quasi-Einstein equations up to a certain order, and is divided into several items with additional conditions. Although some items follow directly from others, each will be used separately in Section \ref{scri}. For the sake of clarity, we therefore present them individually.

\begin{lema}
	\label{lema_bianchi1}
Fix $m\ge 1$ and assume $\Q_{\alpha\beta}^{(k)}=0$ and $\L_{\alpha\beta}^{(k)}=0$ for every $k=1,...,m-1$ whenever $m\geq 2$.
	\begin{enumerate}
		\item If $\Q^{(m)}_{ab}n^b=0$, then
		\begin{equation}
			\label{lema_6}
			\sigma^{(1)}\Q^{(m+1)}_{ab}n^an^b + m\big(2\sigma^{(1)}\kappa_n+\lie_n\sigma^{(1)}\big)\dot\Q^{(m)}_an^a= \dfrac{\mf n-1}{2}\lie_n f^{(m+1)} + m\sigma^{(1)}\L_a^{(m)}n^a.
		\end{equation}
		\item If $\Q^{(m)}_{ab}=0$, then
		\begin{equation}
			\label{lema_1.5}
			\sigma^{(1)}\Q^{(m+1)}_{ab}n^b+m\big(2\sigma^{(1)}\kappa_n+\lie_n\sigma^{(1)}\big)\dot\Q_a^{(m)}=\dfrac{\mf n-1}{2}\nablacero_a f^{(m+1)} + m\sigma^{(1)}\L_{a}^{(m)}
		\end{equation}
		\item If $\Q^{(m)}_{ab}=0$ and $\dot{\Q}_a^{(m)}=0$, then
		\begin{equation}
			\label{lema_1}
			\sigma^{(1)}\dot\Q^{(m+1)}_a n^a + m\big(2 \sigma^{(1)}\kappa_n+\lie_n\sigma^{(1)}\big)\ddot{Q}^{(m)}=\dfrac{\mf n-1}{2}f^{(m+2)}+m\sigma^{(1)}\dot\L^{(m)}
		\end{equation}
		and 
		\begin{equation}
			\label{equationQm+1}
			\sigma^{(1)}\Q^{(m+1)}_{ab}n^b=\dfrac{\mf n-1}{2}\nablacero_a f^{(m+1)}+m\sigma^{(1)} \L_{a}^{(m)}.
		\end{equation}
	\end{enumerate}	
	\begin{proof}
Directly from \eqref{relationQLfm}, $$\mc{Q}_{\alpha\beta}^{(m+1)}\nabla^{\beta}\Omega+m\mc{Q}_{\alpha\beta}^{(m)}\left((\lie_{\xi}g^{\beta\mu})\nabla_{\mu}\Omega + g^{\beta\mu}\nabla_{\mu}\lie_{\xi}\Omega\right) \st{\scri}{=} \dfrac{\mf n-1}{2}\nabla_{\alpha}\lie_{\xi}^{(m)}f + m \sigma^{(1)}\L_{\alpha}^{(m)},$$ 

which after inserting $\nabla^{\beta}\Omega\st{\scri}{=}\sigma^{(1)}\nu^{\beta}$, \eqref{Kupnu} and \eqref{inversemetric} becomes 
\begin{equation}
	\label{identitybianchi1}
	\begin{aligned}
\sigma^{(1)} \mc{Q}_{\alpha\beta}^{(m+1)}\nu^{\beta}+ m\Q_{\alpha\beta}^{(m)}\left(P^{bc}\big(\nablacero_c\sigma^{(1)}-2\sigma^{(1)}\r_c\big)e_b^{\beta}  + \left(\sigma^{(2)}-\dfrac{1}{2}\sigma^{(1)} n(\elltwo)\right)\nu^{\beta} \right) \\
+m\big(2\sigma^{(1)}\kappa_n+\lie_n\sigma^{(1)}\big)\xi^{\beta} \Q_{\alpha\beta}^{(m)}\st{\scri}{=}\dfrac{\mf n-1}{2}\nabla_{\alpha}\lie_{\xi}^{(m)}f + m \sigma^{(1)}\L_{\alpha}^{(m)}.
\end{aligned}
\end{equation}
The contraction of this identity with $\nu^{\alpha}$ gives \eqref{lema_6} at once after using that $\Q^{(m)}_{ab}n^b=0$. To prove \eqref{lema_1.5} one contracts \eqref{identitybianchi1} with $e_a^{\alpha}$ and uses the hypothesis $\Q^{(m)}_{ab}=0$. Identity \eqref{equationQm+1} is an immediate consequence of \eqref{lema_1.5}. Finally, to show \eqref{lema_1} we contract \eqref{identitybianchi1} with $\xi^{\alpha}$ and use $\Q^{(m)}_{ab}=0$ and $\dot{\Q}_a^{(m)}=0$.
	\end{proof}
\end{lema}

\begin{lema}
	\label{lema_bianchi2}
Fix $m\ge 1$ and assume $\Q^{(k)}_{\alpha\beta}=0$ for every $k=1,...,m-1$ whenever $m\ge 2$.
	\begin{enumerate}
\item If $\Q^{(m)}_{ab}n^b=0$, then
\begin{equation}
	\label{lema_5}
	\begin{aligned}
		\Q^{(m+1)}_{ab}n^an^b&=  \lie_n\big(\tr_P\Q^{(m)}\big)+\dfrac{\tr_P\bU}{\mf n-1} \tr_P\Q^{(m)}+P^{bc}P^{da}\U_{bd}\wh{\Q}^{(m)}_{ac}+\lie_n\big(\dot{\Q}^{(m)}_an^a\big)\\
		&\quad\,-\big(2\kappa_n+\tr_P\bU\big)\dot{\Q}^{(m)}_an^a  +	\mf{n}\lie^{(m)}_an^a.
	\end{aligned}
\end{equation}
\item If $\Q^{(m)}_{ab}=0$, then 
\begin{equation}
	\label{lema_4}
	\Q^{(m+1)}_{ab}n^b + \lie_n\dot\Q^{(m)}_a + \big(2\kappa_n+\tr_P\bU\big)\dot\Q_a^{(m)} = 2 \nablacero_a\big(\dot{\Q}^{(m)}_bn^b\big) + \mf n \L^{(m)}_a.
\end{equation}
\item If $\Q^{(m)}_{ab}=0$ and $\dot\Q^{(m)}_a=0$, then 
\begin{equation}
	\label{lema_3}
	\lie_n\ddot\Q^{(m)} + \big((3-2m)\kappa_n+\tr_P\bU\big)\ddot\Q^{(m)} = \dot\Q^{(m+1)}_an^a+	\tr_P\Q^{(m+1)} + 	\mf{n}\dot\L^{(m)}.
\end{equation}		
\item Finally, if $\Q^{(m)}_{\alpha\beta}=0$, then 
		\begin{equation}
			\label{lema_2}
	\dot\Q^{(m+1)}_an^a+\tr_P\Q^{(m+1)}+\mf n \dot\L^{(m)}=0
		\end{equation}
		and
		\begin{equation}
			\label{lema_2_2}
			\Q^{(m+1)}_{ab}n^an^b = \mf n \L^{(m)}_an^a.
		\end{equation}
	\end{enumerate}
	\begin{proof}
Relations \eqref{lema_2} and \eqref{lema_2_2} are particularizations respectively of \eqref{lema_3} and \eqref{lema_4}, so it suffices to prove \eqref{lema_5}-\eqref{lema_3}. Our strategy is to write down \eqref{bianchim} at $\scri$ and then compute its contractions first with $e_c^{\mu}$ and then with $\xi^{\mu}$. The former will show \eqref{lema_5}-\eqref{lema_4}, and the later will establish \eqref{lema_3}. With the assumption $\Q^{(k)}_{\alpha\beta}=0$ for $k=1,...,m-1$, identity \eqref{bianchim} reads
		\begin{align}
			0 &\st{\scri}{=}- \nabla_{\rho}\Q^{(m)}{}^{\rho}{}_{\mu}+ \nabla_{\mu}\left(g^{\alpha\beta}\Q^{(m)}_{\alpha\beta} + (m-1) \big(\lie_{\xi}g^{\alpha\beta}\big)\Q^{(m-1)}_{\alpha\beta}\right) + \mf{n}\L_{\mu}^{(m)}\nonumber\\
			&\st{\scri}{=}- \nabla_{\rho}\Q^{(m)}{}^{\rho}{}_{\mu}+ g^{\alpha\beta}\nabla_{\mu}\Q^{(m)}_{\alpha\beta} - (m-1) \mc{K}^{\alpha\beta}\nabla_{\mu}\Q^{(m-1)}_{\alpha\beta}+ \mf{n}\L_{\mu}^{(m)}\label{aux_Bianchi}
		\end{align} 
where in the second line we used $\Q^{(m-1)}_{\alpha\beta}\st{\scri}{=}0$ and $\lie_{\xi}g^{\alpha\beta} = - \mc{K}^{\alpha\beta}$.\\

We now compute the contraction of \eqref{aux_Bianchi} with $e_c^{\mu}$ under the assumption $\Q^{(m)}_{ab}n^b=0$, i.e. (cf. \eqref{decomp_EM}) $\Q^{(m)}_{ab} = \frac{\tr_P\Q^{(m)}}{\mf n-1}\gamma_{ab} +\wh\Q^{(m)}_{ab}$, from which \eqref{lema_5} and \eqref{lema_4} will then follow at once. This contraction is
\begin{equation}
	\label{aux_Bianchi3}
	0 \st{\scri}{=}-\big(\div\Q^{(m)}\big)_a+ \nablacero_a\left(\tr_P\Q^{(m)}+2\dot{\Q}^{(m)}_b n^b\right)  + \mf{n}\L^{(m)}_a
\end{equation}
because the tangential derivatives of $\Q^{(m-1)}_{\alpha\beta}$ vanish at $\scri$. Only the first term requires further analysis. In Appendix \ref{appendix} we recall several general identities for pullbacks onto a null hypersurface $\mc H$. Applying Proposition \ref{propdivergencia} to $\mc H=\scri$ and using $\Q^{(m)}_{ab}n^b=0$ we get
\begin{align*}
	\big(\div\Q^{(m)}\big)_a & = \Q^{(m+1)}_{ab}n^b + P^{bc}\nablacero_b\Q^{(m)}_{ac} + n^b\nablacero_b \dot{\Q}^{(m)}_a + \big(2\kappa_n +\tr_P\bU\big)\dot{\Q}^{(m)}_a\\
	&\quad\,  -2P^{bc}(\r+\s)_b \Q^{(m)}_{ac} + P^{bc}\U_{ba}\dot{\Q}^{(m)}_c + \s_a \dot{\Q}^{(m)}_bn^b\\
	&= \Q^{(m+1)}_{ab}n^b + P^{bc}\nablacero_b\Q^{(m)}_{ac} + \lie_n\dot{\Q}^{(m)}_a + \big(2\kappa_n +\tr_P\bU\big)\dot{\Q}^{(m)}_a  -2P^{bc}(\r+\s)_b \Q^{(m)}_{ac} ,
\end{align*}
where we used (cf. \eqref{derivadannull}) $\lie_n\dot\Q^{(m)}_a = n^b\nablacero_b \dot\Q^{(m)}_a + P^{bc}\U_{ba}\dot\Q^{(m)}_c + \s_a\dot\Q^{(m)}_bn^b$. Since $\Q^{(m)}_{ab} = \frac{\tr_P\Q^{(m)}}{\mf n-1}\gamma_{ab}+\wh\Q^{(m)}_{ab}$, the second and fifth terms can be elaborated further by means of
\begin{align*}
	P^{bc} \nablacero_b \Q^{(m)}_{ac} &\stb[\eqref{nablagamma}]{}{=} \dfrac{1}{\mf n-1}\left(P^{bc}\gamma_{ac} \nablacero_b(\tr_P\Q^{(m)}) - 2\tr_P\Q^{(m)} P^{bc}\U_{b(a}\ell_{c)} \right)+P^{bc}\nablacero_b\wh\Q^{(m)}_{ac}\\ 
	&\stb[\eqref{Pgamma},\eqref{Pell}]{}{=} \dfrac{1}{\mf n-1}\left( \nablacero_a(\tr_P\Q^{(m)}) - \big(\lie_n\big(\tr_P\Q^{(m)}\big) +\tr_P\bU \tr_P\Q^{(m)}\big) \ell_{a} \right)+P^{bc}\nablacero_b\wh\Q^{(m)}_{ac}
\end{align*}
and
$$P^{bc} \Q^{(m)}_{ac} \st{\eqref{Pgamma}}{=} \dfrac{\tr_P\Q^{(m)}}{\mf n-1}(\delta^b_a-n^b\ell_a)+P^{bc}\wh\Q^{(m)}_{ac}.$$

Consequently,
\begin{align*}
	\big(\div\Q^{(m)}\big)_a &= \Q^{(m+1)}_{ab}n^b + \lie_n \dot{\Q}^{(m)}_a+ \big(2\kappa_n +\tr_P\bU\big)\dot{\Q}^{(m)}_a +P^{bc}\nablacero_b\wh\Q^{(m)}_{ac} -2P^{bc}(\r+\s)_b\wh\Q^{(m)}_{ac}\\
	&\quad\,+	\dfrac{1}{\mf n-1} \left(\nablacero_a\tr_P\Q^{(m)} - \lie_n\big(\tr_P\Q^{(m)}\big)\ell_a  -	\big( (2\kappa_n+\tr_P\bU)\ell_a + 2(\r+\s)_a \big)\tr_P\Q^{(m)}\right).
\end{align*}
 Inserting this into \eqref{aux_Bianchi3} and simplifying gives 
\begin{align*}
	0 &= -\Q^{(m+1)}_{ab}n^b - \lie_n \dot{\Q}^{(m)}_a- \big(2\kappa_n +\tr_P\bU\big)\dot{\Q}^{(m)}_a+\dfrac{1}{\mf n-1}\left((\mf n-2)\nablacero_a\tr_P\Q^{(m)} + \lie_n \big(\tr_P\Q^{(m)}\big) \ell_a\right)\\
	&\quad\, +\dfrac{1}{\mf n-1}\Big( \big(2\kappa_n+\tr_P\bU\big)\ell_a + 2(\r+\s)_a \Big)\tr_P\Q^{(m)}-P^{bc}\nablacero_b\wh\Q^{(m)}_{ac} +2P^{bc}(\r+\s)_b\wh\Q^{(m)}_{ac}\\
	&\quad\,+ 2 \nablacero_a\big(\dot{\Q}^{(m)}_bn^b\big) + \mf n \L^{(m)}_a.
\end{align*}
A contraction with $n^a$ gives \eqref{lema_5} after noting $\wh{Q}_{ab}^{(m)}n^a=0$ and $n^a\nablacero_b\wh\Q^{(m)}_{ac} = -P^{ad}\U_{db}\wh\Q^{(m)}_{ac}$ (cf. \eqref{derivadannull}), while \eqref{lema_4} is just this expression after setting $\tr_P\Q^{(m)}=0$ and $\wh{\Q}^{(m)}_{ab}=0$.\\

To show \eqref{lema_3} we now contract \eqref{aux_Bianchi} with $\xi^{\mu}$ and use $$\xi^{\mu}\nabla_{\rho}\Q^{(m)}{}^{\rho}{}_{\mu} = \nabla_{\rho}\dot\Q^{(m)}{}^{\rho}-\Q^{(m)}{}^{\rho}{}_{\mu}\nabla_{\rho}\xi^{\mu},\qquad \xi^{\mu}\nabla_{\mu}\Q^{(m)}_{\alpha\beta} = \Q^{(m+1)}_{\alpha\beta} - 2\Q^{(m)}_{\mu(\alpha}\nabla_{\beta)}\xi^{\mu},$$

 and $\xi^{\mu}\nabla_{\mu}\Q^{(m-1)}_{\alpha\beta} \st{\scri}{=} \Q^{(m)}_{\alpha\beta}$ (because $\Q^{(m-1)}_{\alpha\beta}=0$) to get
\begin{align*}
0\st{\scri}{=} -\nabla_{\rho}\dot\Q^{(m)}{}^{\rho} - \Q^{(m)}{}^{\rho}{}_{\mu}\nabla_{\rho}\xi^{\mu} +g^{\alpha\beta}\Q_{\alpha\beta}^{(m+1)}-(m-1)\mc{K}^{\alpha\beta}\mc{Q}^{(m)}_{\alpha\beta} + \mf{n}\dot\L^{(m)}.
\end{align*}
The first term is given by Prop. \ref{propdivergencia}, which under the assumption $\dot{\Q}_a^{(m)}=0$ is $\nabla_{\rho}\dot\Q^{(m)}{}^{\rho}=\dot{\Q}^{(m+1)}_a n^a + \lie_n \ddot{\Q}^{(m)} +\big(2\kappa_n +\tr_P\bU\big)\ddot{\Q}^{(m)}$. The second term is computed from \eqref{nablaxiup} after noting that, after the assumptions of item 3., namely $\Q^{(m)}_{ab}=0$ and $\dot\Q^{(m)}_a=0$, only the term with two riggings survives, $\Q^{(m)}{}^{\rho}{}_{\mu}\nabla_{\rho}\xi^{\mu} = -\kappa_n \ddot\Q^{(m)}$. The third term is simply $g^{\alpha\beta}\Q_{\alpha\beta}^{(m+1)}\stb[\eqref{inversemetric}]{}{=} \tr_P\Q^{(m+1)}+2\dot\Q^{(m+1)}_an^a$, and for the fourth term we use again that the only piece that does not vanish is the contraction with $\xi$ twice, obtaining $$\mc{K}^{\alpha\beta}\mc{Q}^{(m)}_{\alpha\beta} = g^{\alpha\mu}\mc{K}_{\mu}{}^{\beta}\mc{Q}^{(m)}_{\alpha\beta} \stb[\eqref{inversemetric}]{}{=} \nu^{\mu} \xi^{\alpha} \mc{K}_{\mu}{}^{\beta}\mc{Q}^{(m)}_{\alpha\beta} \stb[\eqref{Kupnu}]{}{=} -2\kappa_n\xi^{\alpha}\xi^{\beta}\mc{Q}^{(m)}_{\alpha\beta}=-2\kappa_n\ddot\Q^{(m)}. $$ 

Adding up the four terms and $\mf{n}\dot\L^{(m)}$, \eqref{lema_3} is established.
		\end{proof}
	\end{lema}
	
Combining \eqref{lema_1.5} and \eqref{lema_4} one has the following corollary. 
	\begin{cor}
		\label{cor_1}
Let $m\ge 1$ and assume $\Q^{(k)}_{\alpha\beta}=0$ and $\L^{(k)}_{\mu}=0$ for every $k=1,...,m-1$, as well as $f^{(m+1)}=0$, $\dot\Q^{(m)}_an^a=0$, $\Q^{(m)}_{ab}n^a=0$ and $\tr_P\Q^{(m)}=0$. If $m=\mf n$, then $\Q^{(\mf n+1)}(n,n) = \mf n \L^{(\mf n)}_an^a$; and if $m\neq\mf n  $, then $\Q^{(m+1)}(n,n) = \L^{(m)}_an^a=0$.
		\begin{proof}
From \eqref{lema_6} one has $\Q^{(m+1)}(n,n)=m\L^{(m)}_an^a$, and from $\Q_{ab}=0$ we have $\U_{ab}=0$ (see \eqref{Q1general}), so \eqref{lema_5} gives $\Q^{(m+1)}(n,n) = \mf n \L^{(m)}_an^a$. When $\mf n=m$, then $\Q^{(m+1)}(n,n) = \mf n \L^{(m)}_an^a$; and when $\mf n\neq m$, $\Q^{(m+1)}(n,n) = \L^{(m)}_an^a=0$. 
		\end{proof}
	\end{cor}
	
As we will see in detail below, these results show that, for generic values of $m$, determining the transverse expansion at $\scri$ reduces to analyzing, order by order, the equations $\ddot\Q^{(m)}=0$, $f^{(m)}=0$, $\dot\L^{(m)}=0$ as well as $\wh\L^{(m)}_a \d \L^{(m)}_a - (\L^{(m)}_b n^b) \ell_a = 0$ (cf. \eqref{decom2}) and (cf. \eqref{decomp_EM}) $$\wh{\Q}_{ab}^{(m)}\d {\Q}_{ab}^{(m)}-\frac{\tr_P{\Q}^{(m)}+\Q^{(m)}(n,n)\elltwo}{\mf n-1}\gamma_{ab}-2\ell_{(a}\Q^{(m)}_{b)c}n^c -\Q^{(m)}(n,n)\ell_a\ell_b=0.$$

 The remaining equations will turn out to be automatically satisfied as a consequence of Lemmas \ref{lema_bianchi1} and \ref{lema_bianchi2}, together with Corollary \ref{cor_1}. 
	
\begin{prop}
	\label{prop_Bianchi}
Let $(\mc M,g,\Omega)$ be an $(\mf n+1)$-dimensional conformal manifold with null infinity $\scri$ and assume $\scri$ admits a cross-section $\Sigma$. Fix an integer $\ell\ge 1$ and assume that equations $\L_{\mu}^{(k)}= 0$, $\Q^{(k)}_{\mu\nu}=0$ hold for every $k\le \ell$, $f^{(k)}=0$ for every $k\le \ell+2$, and also $\L^{(\mf n)}_an^a=0$ when $\ell=\mf n-1$ and $f^{(\mf n+1)}|_{\Sigma}=0$ when $\ell=\mf n-2$. Then,
\begin{enumerate}
	\item $\dot\Q^{(\ell+1)}_an^a=0$, $\Q^{(\ell+1)}_{ab}n^b=0$, $P^{ab}\Q^{(\ell+1)}_{ab}=0$ and $\Q^{(\ell+2)}_{ab}n^an^b=\L^{(\ell+1)}_an^a=0$.
\end{enumerate}
Assume in addition that $\wh\L_a^{(\ell+1)}=0$, $\dot{\wh\Q}_{a}^{(\ell+1)} |_{\Sigma}= 0$ and $\wh\Q_{ab}^{(\ell+1)}=0$. Then,
\begin{enumerate}
	\item[2.] $\dot{\Q}^{(\ell+1)}_a=0$ and $\Q^{(\ell+2)}_{ab}n^b=0$.
\end{enumerate} 
If, moreover, we suppose $\ddot\Q^{(\ell+1)}=\dot\L^{(\ell+1)}=0$ and, provided $\ell\neq \mf n-2$, that $f^{(\ell+3)}=0$. Then
\begin{enumerate}
	\item[3.] $P^{ab}\Q^{(\ell+2)}_{ab}=0$, $\dot{\Q}^{(\ell+2)}_an^a=0$, and if $\ell=\mf n-2$ also $f^{(\mf n+1)}=0$ and $\Q^{(\mf n+1)}_{ab}n^an^b=0$.
\end{enumerate}
\begin{proof}
In order to prove items 1., 2. and 3. we make use of the identities in Lemmas \ref{lema_bianchi1} and \ref{lema_bianchi2} particularized to $\bU=0$ (because $\Q_{ab}=0$, see \eqref{Q1general}). First we prove that the equations $\dot\Q^{(\ell+1)}_an^a=0$, $\Q^{(\ell+1)}_{ab}n^b=0$ and $P^{ab}\Q^{(\ell+1)}_{ab}=0$ hold under the assumptions $\L_{\mu}^{(k)}= 0$, $\Q^{(k)}_{\mu\nu}=0$ for $k\le \ell$, and $f^{(k)}=0$ for $k\le \ell+2$. Under the present hypothesis all the conditions of item 3. in Lemma \ref{lema_bianchi1} with $m=\ell$ are verified and equation \eqref{equationQm+1} simplifies to $\Q^{(\ell+1)}_{ab}n^b=0$. Similarly, $P^{ab}\Q^{(\ell+1)}_{ab}=0$ is a consequence of \eqref{lema_2} also for $m=\ell$. The remaining two claims in this item are now immediate by Corollary \ref{cor_1} for $m=\ell+1$ (note that when $\ell+1=\mf n$, $\L^{(\ell+1)}_an^a = \L^{(\mf n)}_an^a=0$ holds by hypothesis).\\

Now we prove item 2. We first note that the result $\L^{(\ell+1)}_an^a=0$ combined with the hypothesis $\wh\L_a^{(\ell+1)}=0$ gives $\L_a^{(\ell+1)}=0$, and the results in item 1. combined with the hypothesis $\wh\Q_{ab}^{(\ell+1)}=0$ gives $\Q_{ab}^{(\ell+1)}=0$. Since by item 1. we also have $\dot{\Q}_a^{(\ell+1)} n^a =0$, identities \eqref{lema_1.5} and \eqref{lema_4} for $m=\ell+1$ simplify to
\begin{align*}
	\sigma^{(1)}\Q^{(\ell+2)}_{ab}n^b +(\ell+1)\big(2\sigma^{(1)}\kappa_n+\lie_n\sigma^{(1)}\big)\dot\Q^{(\ell+1)}_a &=0,\\
	\Q^{(\ell+2)}_{ab}n^b  + \lie_n\dot\Q^{(\ell+1)}_a + 2\kappa_n \dot\Q^{(\ell+1)}_a &=0.
\end{align*}
Solving for $\Q^{(\ell+2)}_{ab}n^b$ in the second and inserting it into the first gives $$\sigma^{(1)}\lie_n\dot\Q_a^{(\ell+1)} - \big(2\ell\sigma^{(1)}\kappa_n + (\ell+1)\lie_n\sigma^{(1)}\big)\dot\Q_a^{(\ell+1)}=0.$$ 

This is a homogeneous transport equation for $\dot\Q^{(\ell+1)}_a$, and since $\dot\Q^{(\ell+1)}_a\st{\Sigma}{=}0$ (because $\dot{\wh\Q}_{a}^{(\ell+1)}\st{\Sigma}{=}0$ and $\dot\Q^{(\ell+1)}_an^a=0$ everywhere, and in particular at $\Sigma$), one concludes that $\dot\Q^{(\ell+1)}_a=0$, and consequently $\Q^{(\ell+2)}_{ab}n^b=0$. \\

It only remains to prove item 3. which has the additional hypotheses $\ddot\Q^{(\ell+1)}=\dot\L^{(\ell+1)}=0$ (which imply $\Q_{\alpha\beta}^{(\ell+1)}=0$ and $\L_{\mu}^{(\ell+1)}=0$) and, provided $\ell\neq \mf n-2$, that $f^{(\ell+3)}=0$. Then, by \eqref{lema_1} for $m=\ell+1$ one has $\dot\Q^{(\ell+2)}_an^a = 0$, and since $\Q^{(\ell+1)}_{\alpha\beta}=0$, from item 4. in Lemma \ref{lema_bianchi2} for $m=\ell+1$ we get $\tr_P\Q^{(\ell+2)}=0$, which proves the item for $\ell\neq \mf n-2$. For the case $\ell=\mf n-2$ we cannot use \eqref{lema_1} to prove $\dot\Q^{(\mf n)}_an^a=0$ because we do not know yet that $f^{(\mf n+1)}=0$. Instead, the strategy is to consider the identities \eqref{lema_1} and \eqref{lema_2} for $m=\mf n-1$ and the identities \eqref{lema_6} and \eqref{lema_5} for $m=\mf n$ (recall that $\Q^{(\mf n-1)}_{\alpha\beta}=0$, $\L^{(\mf n-1)}_{\mu}=0$, $\L^{(\mf n)}_an^a=0$ and, by item 2., $\Q^{(\mf n)}_{ab}n^a=0$),
\begin{align*}
\sigma^{(1)}\dot{\Q}^{(\mf n)}_an^a -\dfrac{\mf n-1}{2}f^{(\mf n+1)}&=0,\\
\dot{\Q}^{(\mf n)}_an^a + \tr_P\Q^{(\mf n)} &=0,\\
\sigma^{(1)}\Q^{(\mf n+1)}_{ab}n^an^b + \mf n \big(2\sigma^{(1)}\kappa_n +\lie_n\sigma^{(1)}\big)\dot{\Q}^{(\mf n)}_an^a -\dfrac{\mf n-1}{2}\lie_nf^{(\mf n+1)}&= 0,\\
\Q^{(\mf n+1)}_{ab}n^an^b-\lie_n\big(\tr_P\Q^{(\mf n)}\big) - \lie_n\big(\dot{\Q}^{(\mf n)}_an^a\big)+2\kappa_n \dot{\Q}^{(\mf n)}_an^a &= 0.
\end{align*}
Combining the second and the fourth yields $\Q^{(\mf n+1)}_{ab}n^an^b+2\kappa_n \dot{\Q}^{(\mf n)}_an^a =0$, which after inserted in the third and taking into account the first gives $$\sigma^{(1)}\lie_n f^{(\mf n+1)}-\big(2(\mf n-1)\sigma^{(1)}\kappa_n+\mf n\lie_n\sigma^{(1)}\big)f^{(\mf n+1)} = 0.$$ 

Since $f^{(\mf n+1)}|_{\Sigma}=0$ we conclude $f^{(\mf n+1)}=0$, and as a consequence, $\dot\Q^{(\mf n)}_a n^a=0$, $\tr_P\Q^{(\mf n)}=0$ and $\Q^{(\mf n+1)}_{ab}n^an^b=0$ also hold.
\end{proof}
\end{prop}

This proposition shows that it suffices to analyse equations $\wh\L_a^{(k)}=0$, $\ddot\Q^{(k)}=0$, $\dot\L^{(k)}=0$, $\dot{\wh\Q}^{(k)}_a|_{\Sigma}=0$ and $\wh\Q^{(k)}_{ab}=0$ at all orders, as well as $f^{(k)}=0$ (except for $k=\mf n+1$), $f^{(\mf n+1)}|_{\Sigma}=0$ and $\L_a^{(\mf n)}n^a=0$, since the rest of the conformal equations follow from them. Note that the number of equations agrees with the degrees of freedom to be fixed at each order (i.e. a symmetric tensor $\bY^{(k)}$ and a scalar function $\sigma^{(k)}$). Indeed, the reduced set of equations $\wh\L_a^{(k)}=0$, $\ddot\Q^{(k)}=0$, $\dot\L^{(k)}=0$, $f^{(k)}=0$ and $\wh\Q^{(k)}_{ab}=0$ constitute exactly $(\mf n-1)+1+1+1+\frac{(\mf n+1)(\mf n-2)}{2} = \frac{\mf n(\mf n+1)}{2}+1$ equations, the same as the number of degrees of freedom of $\bY^{(k)}$ and $\sigma^{(k)}$. An immediate corollary of this proposition that will be useful in Theorem \ref{teorema} is the following.
\begin{cor}
\label{corolariopff}
Let $(\mc M,g,\Omega)$ be an $(\mf n+1)$-dimensional conformal manifold with null infinity $\scri$ and assume $\scri$ admits a cross-section $\Sigma$. Fix an integer $\ell\ge 1$ and assume that equations $\L_{\mu}^{(k)}= 0$, $\Q^{(k)}_{\mu\nu}=0$ hold for every $k\le \ell$, $f^{(k)}=0$ for every $k\le \ell+2$, and also $\L^{(\mf n)}_an^a=0$ when $\ell=\mf n-1$ and $f^{(\mf n+1)}|_{\Sigma}=0$ when $\ell=\mf n-2$. Then,
\begin{enumerate}
	\item $\L^{(\ell+1)}_an^a=0$.
\end{enumerate}
Assume in addition that $\L_a^{(\ell+1)}=0$, $\dot{\wh\Q}_{a}^{(\ell+1)} |_{\Sigma}= 0$ and $\Q_{ab}^{(\ell+1)}=0$. Then,
\begin{enumerate}
	\item[2.] $\Q^{(\ell+2)}_{ab}n^b=0$ and $\dot{\Q}^{(\ell+1)}_a=0$.
\end{enumerate} 
If, moreover, we suppose $\ddot\Q^{(\ell+1)}=\dot\L^{(\ell+1)}=0$ and, provided $\ell\neq \mf n-2$, that $f^{(\ell+3)}=0$. Then
\begin{enumerate}
\item[3.] $P^{ab}\Q^{(\ell+2)}_{ab}=0$.
\end{enumerate}
\end{cor}
Actually, one could have proven Corollary \ref{corolariopff} directly with much less effort. However, presented in this way it would not have allowed us to count the number of degrees of freedom to be fixed, and ultimately which equations need to be imposed and which ones are automatically fulfilled form the others.\\

In the following remark we motivate how the reduced set of equations mentioned below Proposition \ref{prop_Bianchi} constrain the asymptotic expansion at null infinity. Later, in Section \ref{scri} we will use this motivated free data to characterize the full asymptotic expansion at null infinity.

\begin{rmk}
	\label{elremarkmaslargodelmundo}
Consider a conformal manifold $(\mc M,g,\Omega)$ with null infinity $\scri$ (assumed to admit a cross-section $\iota:\Sigma\hookrightarrow\scri$) and denote as usual the embedded metric hypersurface data by $\{\bg,\bm\ell,\elltwo\}$, the asymptotic expansion by $\{\bY^{(k)}\}$, and the transverse derivatives of $\Omega$ at $\scri$ by $\{\sigma^{(k)}\}$. The idea is to see how the reduced set of conformal equations described above Corollary \ref{corolariopff} constrain $\{\bg,\bm\ell,\elltwo\}$, $\{\bY^{(k)}\}_{k\ge 1}$ and $\{\sigma^{(k)}\}_{k\ge 2}$ order by order (the quantity $\sigma^{(1)}$ is part of the data, see Definitions \ref{universal} and \ref{emb_uni}). Of course, the first restriction is $\bU=\frac{1}{2}\lie_n\bg=0$. Equations $\Q^{(1)}_{\alpha\beta}=0$ (see \eqref{Q1}) impose the following two extra restrictions on the data: (i) $\br = \bs + d(\log|\sigma^{(1)}|)$ and (ii) $\sigma^{(2)}=0$, so the full one-form $\br$ is determined in terms of metric data and $\sigma$. In particular, this means that the surface gravity of $\nabla\Omega$ vanishes, $\kappa\d \sigma^{(1)}\kappa_n+\lie_n\sigma^{(1)}=0$, so $f^{(2)}=0$ and $\L_a^{(1)}=0$ follow automatically (by \eqref{f2} and \eqref{La1}). Moreover, equations $\dot\L^{(1)}=0$ \eqref{dotL} and $f^{(3)}=0$ \eqref{f3} constrain the pair $\{\kappa^{(2)},\lie_n(\tr_P\bY)\}$ such that, given a function $\chi$ at $\Sigma$, there is a unique pair $\{\kappa^{(2)},\tr_P\bY\}$ satisfying $\dot\L^{(1)}=f^{(3)}=0$ and $\tr_P\bY|_{\Sigma}=\chi$. For $\mf n\neq 3$, the remaining part of the tensor $\bY$ that is not yet fixed, namely $\wh\bY$, becomes completely constrained by equation $\wh\Q_{ab}^{(2)}=0$ (which is equivalent to \eqref{Q2ab}, because by item 1. in Prop. \ref{prop_Bianchi} for $\ell=1$, $\Q^{(2)}_{ab}n^b=0$ and $P^{ab}\Q^{(2)}_{ab}=0$) in terms of $\wh\bY|_{\Sigma}$. Showing this statement rigorously requires analyzing the compatibility between the equations, something we will accomplish in Theorem \ref{teorema}.\\

When $\mf n=3$, the equation $\wh\Q_{ab}^{(2)}=0$ does not constrain $\wh\bY$ because the coefficients in front of $\lie_n\bY$ and $\bY$ in \eqref{Q2ab} both vanish. In principle, this could mean that $\wh\Q_{ab}^{(2)}=0$ is a new restriction on the data $\{\bg,\bm\ell,\elltwo,\chi\}$. However, the tensor $\mc{Q}_{ab}^{(2)}$ in dimension $\mf n=3$ identically vanishes, so it imposes no new restrictions on the data. It suffices to prove it in a gauge in which $\sigma^{(1)}=1$ and $\elltwo=0$ (by the transformations laws \eqref{transsigma1} and \eqref{transell2} it is easy to see that this gauge can always be achieved), because $\mc{Q}^{(2)}_{ab}$ being (or not) zero is a gauge invariant property provided $\mc{Q}^{(1)}_{\alpha\beta}=0$ (this is a particular case of Lemma \ref{lemaobvio} below). So, in this gauge we also have $\br=\bs$ and $\kappa_n=0$. Under these assumptions, the values of $\kappa^{(2)}$ and $\lie_n(\tr_P\bY)$ that are obtained from $\dot\L^{(1)}=f^{(3)}=0$ are $\kappa^{(2)}=-4P(\bs,\bs)$ and $\lie_n(\tr_P\bY)=-\frac{1}{2}\big(\tr_P\Rcero-5\div\bs +9P(\bs,\bs)\big)$, and therefore the scalar curvature at $\scri$ is given by $R=3(\tr_P\Rcero+P(\bs,\bs)-\div\bs)$ (see \eqref{scal}). Inserting this into \eqref{Q2ab} for $\mf n=3$ gives $$\Q^{(2)}_{ab} = \Rcero_{(ab)}+\s_a\s_b-\nablacero_{(a}\s_{b)}-\dfrac{1}{2}\left(\tr_P\Rcero+P(\bs,\bs)-\div\bs\right)\gamma_{ab}.$$ 

This tensor is manifestly traceless and its contraction with $n^a$ vanishes because $\gamma_{ab}n^a=0$ and $\big(\Rcero_{(ab)}+\s_a\s_b-\nablacero_{(a}\s_{b)}\big)n^a=0$ (see \cite{miguel3}). Moreover, its pullback to a cross section is \cite{miguel3} $\Q^{(2)}_{AB}=R^h_{AB} -\frac{R^h}{2}h_{AB}$, which is identically zero in dimension two. So we conclude that the equation $\Q^{(2)}_{ab}=0$ is fulfilled when $\scri$ is three-dimensional.\\

 For the higher order derivatives the idea is similar but with some important differences. Let us assume that we already know how $\{\bY^{(k)},\sigma^{(k+1)},\kappa^{(k+1)}\}$ are fixed for every $k=1,...,m-1$. Then, equations $\L_a^{(m)}=0$ \eqref{Lma} and $\dot\Q_a^{(m)}|_{\Sigma}=0$ \eqref{dotQam} constrain the remaining part of the one-form $\br^{(m)}$ provided $m\neq \mf n$, while equations $\ddot\Q^{(m)}=0$, $\dot\L^{(m)}=0$ and $f^{(m+2)}=0$ (cf. \eqref{ddotQm}, \eqref{dotLm0} and \eqref{fm0}) read
\begin{equation}
	\label{system2}
	\begin{aligned}
		(\mf n-1) \sigma^{(m+1)}-(m-1)\sigma^{(1)}\tr_P\bY^{(m)} & = L_{\ddot\Q}^{m},\\
		(\mf n-1)\sigma^{(1)}\kappa^{(m+1)}+(\mf n-2)\sigma^{(1)} \lie_n\big(\tr_P\bY^{(m)}\big) + F_m\tr_P\bY^{(m)} & = L_{\dot\L}^{m},\\
		\sigma^{(1)}\lie_n\sigma^{(m+1)} + (\sigma^{(1)})^2\kappa^{(m+1)} + (m+1)\big(\lie_n\sigma^{(1)}+2\sigma^{(1)}\kappa_n\big) \sigma^{(m+1)}& = L_f^{m+2},
	\end{aligned}
\end{equation}
where $L_{\ddot\Q}^{m}$, $L_{\dot\L}^{m}$ and $L_f^{m+2}$ gather the lower order terms that we already know how are constrained, and $F_m\d \big((2m(\mf n-1)-\mf n )\sigma^{(1)}\kappa_n + (m-1)\mf n \lie_n\sigma^{(1)}\big)$. Given $\sigma^{(m+1)}|_{\Sigma}$ (this is a completely free function associated to the remaining conformal freedom present in Lemma \ref{lema_conformalgauge}), this system admits a unique solution for $\{\sigma^{(m+1)},\kappa^{(m+1)},\tr_P\bY^{(m)}\}$ provided $m\neq \mf n-1$. Indeed, taking the Lie derivative of the first one along $n$ and combining it with the three equations in \eqref{system2} gives (the precise combination can be read off from the RHS)
\begin{equation}
	\label{combinacionmagica}
	\begin{aligned}
	(\mf n-m-1)(\sigma^{(1)})^2\lie_n\big(\tr_P\bY^{(m)}\big)+\sigma^{(1)}G_m\tr_P\bY^{(m)}\\
=	\sigma^{(1)}L_{\dot\L}^{m}-(\mf n-1)L_f^{m+2}+\sigma^{(1)}\lie_n\big(L_{\ddot\Q}^{m}\big)+(m+1)\big(\lie_n\sigma^{(1)}+2\sigma^{(1)}\kappa_n\big)L_{\ddot\Q}^{m},
\end{aligned}
\end{equation}
where $G_m\d m(\mf n-m-1)\sigma\kappa_n+(m-1)(\mf n-m-2)\kappa= m(\mf n-m-1)\sigma\kappa_n$ (recall that the surface gravity $\kappa = \sigma\kappa_n+\lie_n\sigma$ has already been shown to be zero). This admits a unique solution for $\tr_P\bY^{(m)}$ with ``initial'' data determined by the first equation in \eqref{system2}, namely $\tr_P\bY^{(m)}|_{\Sigma} = \frac{1}{(m-1)\sigma^{(1)}}\big((\mf n-1)\sigma^{(m+1)}|_{\Sigma}-L_{\ddot\Q}^{m}|_{\Sigma}\big)$ except when $m= \mf n-1$, because then the coefficient in front of $\lie_n\big(\tr_P\bY^{(m)}\big)$ vanishes and $G_{\mf n-1}=0$. Finally, equation ${\wh\Q}_{ab}^{(m+1)}=0$, which is equivalent to \eqref{Qabm} because $\Q_{ab}^{(m+1)}n^b=0$ and $P^{ab}\Q_{ab}^{(m+1)}=0$ by item 1. in Prop. \ref{prop_Bianchi} for $\ell=m$, constrains the remaining part of $\bY^{(m)}$, i.e. $\wh{\bY}^{(m)}$, in terms of free data $\wh\bY^{(m)}|_{\Sigma}$ provided $2m\neq \mf n-1$. As before, this statement requires checking the compatibility with the rest of the equations, see Theorem \ref{teorema} below. Let us now analyse the ``special'' cases $m = \frac{\mf n-1}{2}$, $m=\mf n-1$ and $m = \mf n$. \\\\\\

\underline{Case $m= \frac{\mf n-1}{2}$ ($\mf n$ even)}\\

In this case, the coefficients multiplying $\lie_n\bY^{(m)}$ and $\bY^{(m)}$ in $\Q^{(m+1)}_{ab}$ (cf. \eqref{Qabm}) vanish (recall that $\kappa=0$), and therefore the equation $\wh\Q^{(m+1)}_{ab}=0$ does not constrain $\wh\bY^{(m)}$ in terms of $\wh\bY^{(m)}|_{\Sigma}$. Instead, it becomes a potential constraint on the remaining free data, i.e. the metric hypersurface data $\{\bg,\bm\ell,\elltwo\}$, the functions $\tr_P\bY|_{\Sigma}$, $\{\sigma^{(k)}|_{\Sigma}\}_{k\le m}$ and the tensors $\{\wh\bY^{(k)}|_{\Sigma}\}_{k< m}$. We denote the resulting tensor by $\mc{O}^{\scri}_{ab}$ and refer to it as the ``radiative obstruction tensor''. It is symmetric (since $\Q$ is), and by item 1. in Prop. \ref{prop_Bianchi} it is also transverse, i.e. $P^{ab}\mc{O}^{\scri}_{ab}=0$ and $\mc{O}^{\scri}_{ab}n^a=0$. A more detailed analysis of $\mc{O}^{\scri}_{ab}$ is presented in Section \ref{subsec_obs}.\\

\underline{Case $m= \mf n-1$}\\

As already discussed, when $m=\mf n-1$ the system \eqref{system2} no longer determines the set of tensors $\{\sigma^{(\mf n)},\kappa^{(\mf n)},\tr_P\bY^{(\mf n-1)}\}$ in terms of $\sigma^{(\mf n)}|_{\Sigma}$, as it does for the remaining values of $m$. Instead, these quantities become fully determined in terms of $\sigma^{(\mf n)}|_{\Sigma}$ and $\kappa^{(\mf n)}|_{\Sigma}$ through the equations $\ddot{Q}^{(\mf n-1)}=0$, $\dot{\L}^{(\mf n-1)}=0$, $\L^{(\mf n)}_a n^a=0$ and $f^{(\mf n+1)}|_{\Sigma}=0$ (that is, we replace the third equation $f^{(\mf n+1)} =0$ in the system \eqref{system2} by $f^{(\mf n+1)}|_{\Sigma}=0$ and $\L^{(\mf n)}_a n^a =0$), as we discuss next. First, the equations $\ddot{Q}^{(\mf n-1)}=0$ and $\dot{\L}^{(\mf n-1)}=0$ (see \eqref{ddotQm} and \eqref{dotLm0}) take the form
\begin{align}
	(\mf n-1) \sigma^{(\mf n)}-(\mf n-2)\sigma^{(1)}\tr_P\bY^{(\mf n-1)} & = L^{\mf n-1}_{\ddot\Q},\label{1st}\\
	(\mf n-1)\sigma^{(1)}\kappa^{(\mf n)}+(\mf n-2)\sigma^{(1)} \lie_n\big(\tr_P\bY^{(\mf n-1)}\big) + F\tr_P\bY^{(\mf n-1)} & = L_{\dot\L}^{\mf n-1}.\label{2nd}
\end{align}
Substituting the expression for $\kappa^{(\mf n)}$ obtained from \eqref{2nd} into the equation $\L^{(\mf n)}_a n^a=0$ (see \eqref{Lnunordenmenos}) yields\footnote{Equation \eqref{Lnunordenmenos} is written in a particular gauge, and therefore it will change under another choice of rigging $\xi$. Nevertheless, its second-order character is gauge invariant.} a second-order transport equation for $\tr_P\bY^{(\mf n-1)}$ of the form 
\begin{equation}
	\label{2ndorder}
\sigma \lie_n^{(2)}\big(\tr_P\bY^{(\mf n-1)}\big) + \wt{H}_1 \lie_n\big(\tr_P\bY^{(\mf n-1)}\big) + \wt{H}_2 \tr_P\bY^{(\mf n-1)} + \mbox{known terms}=0
\end{equation}
for some functions $\wt{H}_i$, $i=1,2$. Thus, given $\sigma^{(\mf n)}|_{\Sigma}$ and $\kappa^{(\mf n)}|_{\Sigma}$, one can then uniquely determine $\tr_P\bY^{(\mf n-1)}|_{\Sigma}$ and $\lie_n\sigma^{(\mf n)}|_{\Sigma}$ from $\ddot{Q}^{(\mf n-1)}|_{\Sigma}=0$ and $f^{(\mf n+1)}|_{\Sigma}=0$ (see the third equation in \eqref{combinacionmagica} which recall is just a rewriting of \eqref{fm0}), namely
\begin{align*}
(\mf n-1) \sigma^{(\mf n)}|_{\Sigma}-(\mf n-2)\sigma^{(1)}\tr_P\bY^{(\mf n-1)}|_{\Sigma} & = L^{\mf n-1}_{\ddot\Q}|_{\Sigma},\\
\sigma^{(1)}\lie_n\sigma^{(\mf n)}|_{\Sigma} + (\sigma^{(1)})^2\kappa^{(\mf n)}|_{\Sigma} + \mf n\big(\lie_n\sigma^{(1)}+2\sigma^{(1)}\kappa_n\big) \sigma^{(\mf n)}|_{\Sigma}& = L_f^{\mf n+1}|_{\Sigma},
\end{align*}
and from $\lie_n\sigma^{(\mf n)}|_{\Sigma}$ one gets $\lie_n\big(\bY^{(\mf n-1)}\big)|_{\Sigma}$ using \eqref{1st}. Integrating \eqref{2ndorder} with this initial data $\tr_P\bY^{(\mf n-1)}|_{\Sigma}$ and $\lie_n\big(\bY^{(\mf n-1)}\big)|_{\Sigma}$ gives $\tr_P\bY^{(\mf n-1)}$ everywhere. Once $\tr_P\bY^{(\mf n-1)}$ is known, equations \eqref{1st} and \eqref{2nd} determine $\sigma^{(\mf n)}$ and $\kappa^{(\mf n)}$ uniquely. Moreover, since the initial condition $\text{\pounds}_n\big(\bY^{(\mf n-1)}\big)|_{\Sigma}$ has been fitted so that $f^{(\mf n+1)}|_{\Sigma}=0$, Proposition \ref{prop_Bianchi} ensures $f^{(\mf n+1)}=0$ everywhere. In summary, we have shown that the equations $\ddot{\Q}^{(\mf n-1)}=0$, $\dot{\L}^{(\mf n-1)}=0$, $\L^{(\mf n)}_a n^a =0$ and $f^{(\mf n+1)}=0$ form a compatible system which determines uniquely $\{\sigma^{(\mf n)},\kappa^{(\mf n)},\tr_P\bY^{(\mf n-1)}\}$ in terms of $\{\sigma^{(\mf n)}|_{\Sigma},\kappa^{(\mf n)}|_{\Sigma}\}$. In particular the system \eqref{system2} for $m = \mf n-2$ is compatible and hence equation \eqref{combinacionmagica} is satisfied. Since the LHS of this equations is identically zero when $m = \mf n-1$, it follows that the RHS is also identically zero, which means that there is no obstruction in this case. \\


\underline{Case $m= \mf n$}\\

For this value of $m$, the coefficient multiplying $\br^{(m+1)}$ in \eqref{dotQam} vanishes, and as a consequence the equation $\dot{\Q}^{(\mf n)}|_{\Sigma}=0$ does not impose any new constraint on the one-form $\br^{(\mf n)}$ on $\Sigma$. Instead, the tensor $\dot{\Q}^{(\mf n)}|_{\Sigma}$ only depends on $\sigma^{(\mf n)}$, $\kappa^{(\mf n)}$ and lower order terms, that have already been determined, so $\dot{\Q}^{(\mf n)}|_{\Sigma}=0$ becomes a potential constraint on the free data. We denote the resulting tensor by $\mc{O}^{\Sigma}_a$ and refer to it as the ``Coulombic obstruction''. Note that by item 1. of Proposition \ref{prop_Bianchi} it satisfies $\mc{O}_a^{\scri}n^a=0$. A more detailed analysis of this obstruction tensor is undertaken in Section \ref{subsec_obs}.
\end{rmk}

In summary, when $\mf n$ is even, the data that can be freely prescribed at $\scri$ consist of:
\begin{enumerate}
	\item Metric hypersurface data $\{\bg,\bm\ell,\elltwo\}$ satisfying $\bU=0$,
	\item The collection of functions $\{\sigma,\sigma^{(2)}=0,\sigma^{(3)},...\}$ on $\Sigma$ ($\sigma$ is pure metric hypersurface gauge, cf. \eqref{transsigma1}),
	\item The scalar $\tr_P\bY$ on $\Sigma$,
	\item The family of tensors $\{\wh{\Y}^{(1)}_{AB},\wh{\Y}^{(2)}_{AB},...\}$ on $\Sigma$, and
	\item The full one-form $\br^{(\mf n)}$ on $\Sigma$.
\end{enumerate}
When $\mf n$ is odd, one must additionally prescribe the tensor $\wh{\bY}{}^{(\frac{\mf n -1}{2})}$ on $\scri$. This extra data must of course be compatible with the $\wh{\bY}{}^{(\frac{\mf n -1}{2})}|_{\Sigma}$ already prescribed in point 4. One way of ensuring this is to just prescribe $\lie_n \wh{\bY}{}^{(\frac{\mf n -1}{2})}$ and solve for $\wh{\bY}{}^{(\frac{\mf n -1}{2})}$ with the already given data at $\Sigma$.


\section{Existence and uniqueness results}
\label{scri}

In this section we prove that the free data described in Remark \ref{elremarkmaslargodelmundo} fully characterizes the geometric structure at null infinity (Theorem \ref{uniqueness}), and reciprocally, that given such free data there exists a conformal spacetime realizing it (Theorem \ref{teorema}). It is convenient first to set up the following notions of asymptotic flatness of finite and infinite order. 

\begin{defi}
	\label{defi_AF}
Let $(\mc M,g,\Omega)$ be a conformal manifold with null infinity $\scri$. We say $(\mc M,g,\Omega)$ is $k$-asymptotically flat provided it satisfies the conformal Einstein equations \eqref{quasiEinsteinT} with $\mc{T}^{(\ell)}_{\alpha\beta}\st{\scri}{=}0$ for every $\ell\le k$. When $k=\infty$ we simply say that $(\mc M,g,\Omega)$ is asymptotically flat.
\end{defi}

This definition is less restrictive that others usually made in the literature (see e.g. the reviews \cite{frauendiener2004conformal,ashtekar2015geometry} or the classic references \cite{penrose1965zero,geroch1977asymptotic}), but it is well-adapted to the asymptotic expansion analysis we perform in this paper. We are ready to establish our uniqueness theorem.

\begin{teo}
	\label{uniqueness}
Let $(\mc M,g,\Omega)$, $({\mc M}', g', \Omega')$ be two asymptotically flat manifolds with respective null infinities $\Phi:\scri\hookrightarrow\mc M$, $\Phi':\scri'\hookrightarrow\mc M'$. Choose any conformal gauge on $(\mc M,g,\Omega)$ and $({\mc M}', g', \Omega')$ such that $|\nabla\Omega|^2_g=0$ and $|\nabla'\Omega'|^2_{g'}=0$, and let $\xi$, $\xi'$ be respective riggings of $\Phi(\scri)$, $\Phi'(\scri')$ extended geodesically. Suppose that $\scri$ admits a cross section $\iota:\Sigma\hookrightarrow\scri$ and that there exists a diffeomorphism $\chi:\scri\to \scri'$. Let $\iota'\d \chi\circ \iota$ so that $\Sigma'\d \chi(\Sigma)\st{\iota'}{\hookrightarrow}\scri'$ and let $\Psi$ be the diffeomorphism of Prop. \ref{prop_diffeo} constructed between neighbourhoods $\mc U\subset\mc M$, $\mc U\subset \mc M'$ of $\Phi(\Sigma)$, $\Phi'(\Sigma')$. Define the function $\omega\d \Omega^{-1}\Psi^{\star}\Omega'$, which is smooth and nowhere zero on $\mc U$, and let $\ol{\bY}^{(k)}\d \Phi^{\star}\lie_{\xi}^{(k)}(\omega^2g)$ and $\varpi\d\Phi^{\star}\omega$. Finally, assume
\begin{enumerate}
	\item $\chi^{\star}\{\bg',\bm\ell',\elltwo{}'\} = \{\varpi^2\bg,\varpi^2\bm\ell,\varpi^2\elltwo\}$,
	\item $\chi^{\star}\tr_{P'}\bY'|_{\Sigma} = \varpi^{-2}\tr_P\ol\bY|_{\Sigma}$ and $\chi^{\star}\wh{\bY}'{}^{(k)} |_{\Sigma}= \wh{\ol\bY}{}^{(k)}|_{\Sigma}$ for every $k\ge 1$,
	\item $\chi^{\star}\br'{}^{(\mf n)}|_{\Sigma} = \ol\br^{(\mf n)}|_{\Sigma}$,
	\item and $\chi^{\star} \wh{\bY}'{}^{(\frac{\mf n -1}{2})}=\wh{\ol\bY}{}^{(\frac{\mf n -1}{2})}$ for $\mf n$ odd.
\end{enumerate}
Then, for all $k\ge 0$, $$\Psi^{\star}\lie_{\xi'}^{(k)}g' \st{\scri}{=} \lie_{\xi}^{(k)}(\omega^2g).$$
\begin{proof}
Let $\ol g\d \omega^2g$ and $\ol\Omega\d \omega\Omega$. Since $\Phi_{\star}\xi = \xi'$ we have $\Psi^{\star}\sigma^{(k)}{}' = \ol{\sigma}^{(k)}$ for every $k\ge 0$, and given that $\chi^{\star}\{\bg',\bm\ell',\elltwo{}'\} = \{\varpi^2\bg,\varpi^2\bm\ell,\varpi^2\elltwo\}$, by Proposition \ref{teo_iso0} it suffices to show that $\chi^{\star}\bY'{}^{(k)} = \ol{\bY}{}^{(k)}$ for all $k\ge 1$. This follows after noting that, excluding the exceptional orders, the equations that constrain $\bY'{}^{(k)}$ and $\ol{\bY}{}^{(k)}$ at each order (see Remark \ref{elremarkmaslargodelmundo}) are identical and share the same initial conditions (by item 2.), which necessarily forces $\bY'{}^{(k)}=\ol{\bY}{}^{(k)}$ order by order. The first exception is related to the one-form $\br^{(\mf n)}$, which is constrained by the conformal equations only up to its value at one cross-section, but since by item 3. one has that $\chi^{\star}\br'{}^{(\mf n)} = \ol\br^{(\mf n)}$ at $\Sigma$, then $\chi^{\star}\br'{}^{(\mf n)} = \ol\br^{(\mf n)}$ everywhere. Additionally, when $\mf n$ is even there is another exceptional case regarding $\wh{\bY}{}^{(\frac{\mf n -1}{2})}$ because the conformal equations do not constrain it at all. This is taken care by item 4., which imposes $\chi^{\star} \wh{\bY}'{}^{(\frac{\mf n -1}{2})}=\wh{\ol\bY}{}^{(\frac{\mf n -1}{2})}$. Hence the conformal equations along with the agreement on the ``free data'' in items 1.-4. completely constrain the transverse expansion, and then a direct application of Proposition \ref{teo_iso0} yields the result.
\end{proof}
\end{teo}
\begin{rmk}
Note that the statement of the theorem already rules out all possible obstructions discussed in Remark \ref{elremarkmaslargodelmundo}, because $(\mc M,g,\Omega)$ and $({\mc M}', g', \Omega')$ are assumed to be asymptotically flat (as in Def. \ref{defi_AF}). With comparable effort one can prove that $\Psi$ is an asymptotic isometry up to order $\frac{\mf n-3}{2}$ for asymptotically flat metrics with differentiability $\mc{C}^{\frac{\mf n-3}{2}}$ and non-vanishing obstruction tensors; or an isometry to order $\frac{\mf n-1}{2}$ for smooth metrics that are $\frac{\mf n-1}{2}$-asymptotically flat.
\end{rmk}

\begin{rmk}
It is worth connecting the uniqueness result in Theorem \ref{uniqueness} with our previous work \cite{Mio6}, where we considered conformal manifolds $(\mc M,g,\Omega)$ arising from Fefferman-Graham ambient metrics, which as we proved there are in one-to-one correspondence with manifolds $(\mc M,g,\Omega)$ admitting an integrable conformal Killing vector $\lie_{\eta}g =2\psi g$ with $\scri$ as one of the horizons and satisfying $\lie_{\eta}(\Omega) = (\psi-1)\Omega$. As proven in \cite{Mio6}, for such manifolds it is always possible to choose the conformal gauge such that $\psi=0$, and thus $\eta$ is Killing. We now connect this with the fact that, under the assumption that $\eta$ is a Killing vector and $\eta \st{\scri}{=} \alpha\nu$ for some function $\alpha$, each element of the transverse expansion $\{\bY^{(k)}\}$ satisfies a constraint given by \cite{Mio3,Mio4}
\begin{equation}
	\label{alfalienYmexact}
	\alpha \lie_{n}\bY^{(m)} = m \lie_n(\alpha) \bY^{(m)} - 2 d\alpha\otimes_s \br^{(m)} +\mc{P}^{(m)} \qquad \forall m\ge 1,
\end{equation}
where $\mc{P}^{(1)} \d \bm\ell \otimes_s d \big(\lie_n\alpha\big) + \dfrac{1}{2} \lie_{X_{\eta}}\bg$,
\begin{equation}
	\label{Pm}
	\mc{P}^{(m)}\d  m\lie_{X_{\eta}}\bY^{(m-1)}-\dfrac{1}{2}\sum_{i=2}^m \binom{m}{i}\Phi^{\star}\Big(\lie_{\lie^{(i)}_{\xi}\eta} \lie_{\xi}^{(m-i)}g\Big)\qquad m\ge 1,
\end{equation}
and $X_{\eta}^a = \frac{1}{2}\alpha n(\elltwo)n^a + P^{ab}\big(2\alpha\s_b+\nablacero_b\alpha\big)$. As proven in \cite{Mio3}, when $\eta$ is non-degenerate, i.e. its surface gravity is nowhere zero, one can always choose $\xi|_{\scri}$ and extend it geodesically so that $\mc{P}^{(1)}=\bm\ell \otimes_s d \big(\lie_n\alpha\big)$ and $\mc{P}^{(m)}=0$ for every $m\ge 2$. Contracting \eqref{alfalienYmexact} for $m\ge 2$ with $n$ once and twice in this gauge gives $$\alpha \lie_n\br^{(m)} = (m-1) \lie_n(\alpha)\br^{(m)} + \kappa^{(m)} d\alpha ,\qquad \alpha \lie_n\kappa^{(m)} = (m-2) \lie_n(\alpha) \kappa^{(m)}.$$ 

Taking $\Sigma$ to be the bifurcation surface, where $\alpha=0$ and $\lie_n\alpha\neq 0$, one concludes that the free data of item 2. in Theorem \ref{uniqueness} must be trivial. Note that this statement is gauge invariant because the leading order of $\bY'{}^{(k)}$ under a gauge transformation is $z^k\bY^{(k)}$. Moreover, condition $\eta(\Omega)=-\Omega$ at all orders implies that the free functions $\{\sigma^{(k)}\}_{k\ge 1}$ at $\Sigma$ are also zero. In conclusion, the only remaining freedom is a traceless tensor $\Psi_{AB}=\wh{\Y}^{(\frac{\mf n-1}{2})}_{AB}|_{\Sigma}$ for $\mf n$ odd, from which the full ${\bY}^{(\frac{\mf n-1}{2})}$ is constrained by equation \eqref{alfalienYmexact}. This freedom agrees with the one we obtained by a different argument in \cite{Mio6}, which related this tensor to the free data that appears in the Fefferman-Graham construction of the ambient metric.
\end{rmk}

We now prove one of the main results of this paper, namely that given $\scri$-structure data (as in Definition \ref{universal}) together with the free data described in Theorem \ref{uniqueness}, it is possible to construct a smooth asymptotically flat\footnote{In the sense of Def. \ref{defi_AF}, which as already noted is weaker than other notions in the literature. Stronger conclusions such as e.g. that $\Omega^{-2}g$ is vacuum in a full neighbourhood of $\scri$ clearly cannot be proven with data given only at $\scri$.} spacetime $(\mc M,g,\Omega)$ provided the obstruction tensors vanish.

\begin{teo}
	\label{teorema}
Let $\{\scri,\bg,\bm\ell,\elltwo,\sigma,\mf q=0\}$ be an $\mf n$-dimensional $\scri$-structure data admitting a cross-section $\iota:\Sigma \hookrightarrow\scri$ with induced metric $h\d\iota^{\star}\bg$. Let
\begin{enumerate}
	\item[a.] $\chi$ be a scalar function on $\Sigma$,
	\item[b.] $\{\sigma_{\Sigma}^{(k)}\}_{k\ge 3}$ a collection of scalar functions on $\Sigma$,
	\item[c.] $\{\mc{Y}^{(k)}_{AB}\}_{k\ge 1}$ a set of symmetric tensors on $\Sigma$ traceless w.r.t. $h$,
	\item[d.] $\bm{\beta}$ and $\mf{m}$ a one-form and a function on $\Sigma$, and
	\item[e.] for $\mf n$ odd, let $\wh\bcY^{(\frac{\mf n-1}{2})}_{ab}$ be a symmetric transverse tensor on $\scri$ satisfying $\iota^{\star}\wh\bcY^{(\frac{\mf n-1}{2})}=\mc{Y}^{(\frac{\mf n-1}{2})}$.
\end{enumerate}
Assume the obstruction tensors $\mc{O}_a^{\Sigma}$ and $\mc{O}_{ab}^{\scri}$ constructed from this data as described in Remark \ref{elremarkmaslargodelmundo} both vanish. Then, there exists a smooth conformal manifold $(\mc M,g,\Omega)$, an embedding $\Phi:\scri\hookrightarrow \mc M$ and a geodesic vector $\xi$ such that (i) $(\mc M,g,\Omega)$ satisfies $|\nabla\Omega|^2_g=0$ and the conformal Einstein equations \eqref{quasiEinstein} to infinite order at $\scri=\{\Omega=0\}$, (ii) $\{\scri,\bg,\bm\ell,\elltwo,\sigma,\mf{q}\}$ is $(\Phi,\xi)$-embedded in $(\mc M,g)$ in the sense of Def. \ref{emb_uni}, and (iii) $\tr_P\bY|_{\Sigma} = \chi$, $\lie_{\xi}^{(2)}\Omega|_{\Sigma}=0$ and $\lie_{\xi}^{(k)}\Omega|_{\Sigma} = \sigma_{\Sigma}^{(k)}$ for every $k\ge 3$, $({\iota^{\star}{\bY}}{}^{(k)})^{tf} = \mc{Y}^{(k)}$ for every $k\ge 1$, $\iota^{\star}\br^{(\mf n)} = \bm{\beta}$ and $\kappa^{(\mf n)}|_{\Sigma} = \mf{m}$. For $\mf n$ odd one has, in addition, that $\wh{\Y}^{(\frac{\mf n-1}{2})}_{ab} = \wh\bcY^{(\frac{\mf n-1}{2})}_{ab}$.\\

\textbf{Remark:} We may assume $q=0$ from the beginning without loss of generality by applying a suitable conformal gauge (see the comment right after Definition \ref{emb_uni}).
\begin{proof}
The strategy of the proof is similar to the one we followed in Theorem 4.7 of \cite{Mio4} with some important differences. The idea is to construct a spacetime $(\mc M,g)$ and a function $\Omega$ using Theorem \ref{borel} and Lemma \ref{Borel_lemma_for_functions} from a collection of ``abstract'' tensors $\{\bcY^{(k)}\}_{k\ge 1}$ and functions $\{\sigma^{(k)}\}_{k\ge 1}$ on $\scri$ so that $(\mc M,g,\Omega)$ satisfies the quasi-Einstein equations to infinite order at $\scri$. To do that, we will construct each $\bcY^{(k)}$ and $\sigma^{(k)}$ from the equations derived in Section \ref{section_transverse} obtained after setting $\Q_{\alpha\beta}^{(k)}=0$, $\L_{\mu}^{(k)}=0$ and $f^{(k)}=0$ for all $k$ and replacing each $\bY^{(k)}$ by $\bcY^{(k)}$. To follow a consistent notation we introduce $\bcr^{(k)}\d \bcY^{(k)}(n,\cdot)$ and $\bck^{(k)}\d - \bcY^{(k)}(n,n)$ to denote the would-be tensors $\br^{(k)}$ and $\kappa^{(k)}$, respectively. \\

\underline{1. The strategy}\\

To construct the abstract expansion $\{\bcY^{(k)}\}$ there arises the following difficulty. Let's say that we want to use the equation $\Q^{(k+1)}_{ab}=0$ to construct $\bcY^{(k)}$ from some initial condition on $\Sigma$. By looking at the right hand side of \eqref{Qabm} (recall that the tensor $\wt{O}_{ab}^{(k)}$ collects additional terms depending on $\bcr^{(k)}$ and $\sigma^{(k)}$) one immediately realizes that this is not a transport equation, but a partial differential system due to the presence of $\nablacero$-derivatives of $\bcr^{(k)}$ in $\wt{O}_{ab}^{(k)}$ and of $\tr_P\bcY^{(k)}$. This is a complicated system and to the best of our knowledge there is no existence theorem available. In addition, there appears the scalar $\kappa^{(k+1)}$, which is one order higher. Thus, our strategy is to first construct a suitable one-form $\bm{c}^{(k)}$ and three scalar functions $t^{(k)}$, $\mf{c}^{(k)}$ and $\mf{c}^{(k+1)}$, and then rewrite the equation $\Q^{(k+1)}_{ab}=0$ but replacing $\bck^{(k)}$ by $\mf{c}^{(k)}$, $\bcr^{(k)}$ by $\bm{c}^{(k)}$, $\tr_P\bcY^{(k)}$ by $t^{(k)}$, and $\kappa^{(k+1)}$ by $\mf{c}^{(k+1)}$. The resulting equation is now an actual transport equation for $\bcY^{(k)}$, from which we can build $\bcY^{(k)}$ given an initial condition at $\Sigma$. Later, we will need to close the argument by showing that $\bm{c}^{(k)}(n)=-\mf{c}^{(k)}$, $\bcY^{(k)}(n,\cdot)\d \bcr^{(k)}=\bm{c}^{(k)}$, $\tr_P\bcY^{(k)}= t^{(k)}$ and $\bcY^{(k+1)}(n,n)=-\mf{c}^{(k+1)}$ for every $k\ge 1$. \\

Now we provide the summary of the argument step by step: Suppose that for a given value\footnote{Assume for the sake of the argument that $m+1$ is not one of the exceptional values $\mf n$, $\mf n-1$ or $\frac{\mf n-1}{2}$. These will be dealt with separately.} of $m\ge 1$ we have already constructed the collection $\{\bcY^{(k)},\sigma^{(k)}\}_{k\le m}$ and also $\sigma^{(m+1)}$ and $\mf{c}^{(m+1)}$ such that the equations $\Q_{\alpha\beta}^{(k)}=0$, $\L_{\mu}^{(k)}=0$ and $f^{(k+2)}=0$ for all $k\le m$, and $\Q_{ab}^{(m+1)}=0$ are all satisfied with the replacements $\bY^{(k)}\sto \bcY^{(k)}$ for every $k\le m$, and $\kappa^{(m+1)}\sto \mf{c}^{(m+1)}$. Then we construct $\bcY^{(m+1)}$, $\sigma^{(m+2)}$ and $\mf{c}^{(m+2)}$ by following these steps:
\begin{enumerate}
\item We define $\bm{c}^{(m+1)}|_{\Sigma}$ as the unique one-form that satisfies (i) $\bm{c}^{(m+1)}(n)|_{\Sigma}=-\mf{c}^{(m+1)}|_{\Sigma}$ and (ii) the equation $\iota^{\star}\dot{\Q}_a^{(m+1)}=0$ (see \eqref{dotQam}) with the replacements $\bY^{(k)}\sto \bcY^{(k)}$ for every $k\le m$, $\br^{(m+1)}\sto \bm{c}^{(m+1)}$ and $\kappa^{(m+1)}\sto \mf{c}^{(m+1)}$.
\item Given $\bm{c}^{(m+1)}|_{\Sigma}$, we integrate the one-form $\bm{c}^{(m+1)}$ by solving the equation $\L_a^{(m+1)}=0$ \eqref{Lma} with the replacements $\bY^{(k)}\sto \bcY^{(k)}$ for every $k\le m$, $\br^{(m+1)}\sto \bm{c}^{(m+1)}$ and $\kappa^{(m+1)}\sto \mf{c}^{(m+1)}$. We emphasize that our working hypothesis is that the function $\mf{c}^{(m+1)}$ has already been fixed at this stage.
\item Having determined $\bm{c}^{(m+1)}$, we build three functions $\sigma^{(m+2)}$, $t^{(m+1)}$ and $\mf{c}^{(m+2)}$ by solving the system of equations $f^{(m+3)}=\dot\L^{(m+1)}=\ddot\Q^{(m+1)}=0$ with the replacements $\bY^{(k)}\sto \bcY^{(k)}$ for every $k\le m$, $\br^{(m+1)}\sto \bm{c}^{(m+1)}$, $\kappa^{(m+1)}\sto \mf{c}^{(m+1)}$, $\tr_P\bY^{(m+1)}\sto t^{(m+1)}$ and $\kappa^{(m+2)}\sto\mf{c}^{(m+2)}$, from an initial condition $\sigma^{(m+2)}|_{\Sigma}=\sigma^{(m+2)}_{\Sigma}$ (see \eqref{system2} and the subsequent discussion on existence of solutions). The free function $\sigma^{(m+2)}_{\Sigma}$ simply encodes the residual conformal freedom, as described in Lemma \ref{lema_conformalgauge}.
\item Finally, with the tensors $\bm{c}^{(m+1)}$, $\sigma^{(m+2)}$, $t^{(m+1)}$ and $\mf{c}^{(m+2)}$ just determined, we construct $\bcY^{(m+1)}$ by integrating the equation $\Q_{ab}^{(m+2)}=0$ with the replacements $\bY^{(k)}\sto \bcY^{(k)}$ for every $k\le m$, $\br^{(m+1)}\sto \bm{c}^{(m+1)}$, $\kappa^{(m+1)}\sto \mf{c}^{(m+1)}$, $\tr_P\bY^{(m+1)}\sto t^{(m+1)}$, $\kappa^{(m+2)}\sto\mf{c}^{(m+2)}$ and $\bY^{(m+1)}\sto \bcY^{(m+1)}$, with the initial condition $\tr_P\bcY^{(m+1)}|_{\Sigma}=t^{(m+1)}|_{\Sigma}$, $\bcY^{(m+1)}(n,\cdot)|_{\Sigma}=\bm{c}^{(m+1)}|_{\Sigma}$ and $(\bcY^{(m+1)}_{AB})^{tf}|_{\Sigma}=\mc{Y}_{AB}^{(m+1)}$.
\end{enumerate}
To close the argument we need to show that the collection $\{\bcY^{(k)},\sigma^{(k)}\}_{k\le m+1}$ and also the functions $\sigma^{(m+2)}$ and $\mf{c}^{(m+2)}$ satisfy the equations $\Q_{\alpha\beta}^{(k)}=0$, $\L_{\mu}^{(k)}=0$ and $f^{(k+2)}=0$ for all $k\le m+1$, and also $\Q_{ab}^{(m+2)}=0$ (each one with the replacements $\bY^{(k)}\sto \bcY^{(k)}$ for every $k\le m+1$, and $\kappa^{(m+2)}\sto \mf{c}^{(m+2)}$). This is equivalent to proving three things: (1) That the scalars $\bm{c}^{(m+1)}(n)$ and $-\mf{c}^{(m+1)}$ are the same; (2) That $\bcr^{(m+1)}\d\bcY^{(m+1)}(n,\cdot)$ agrees with $\bm{c}^{(m+1)}$; And (3) that $\tr_P\bcY^{(m+1)}$ is the same as $t^{(m+1)}$. Finally, we will construct the conformal manifold $(\mc M,g,\Omega)$ from the collection $\{\bcY^{(k)},\sigma^{(k)}\}_{k\ge 1}$ using Theorem \ref{borel} and Lemma \ref{Borel_lemma_for_functions}, and given that these tensors satisfy $\Q_{\alpha\beta}^{(k)}=0$, $\L_{\mu}^{(k)}=0$ and $f^{(k)}=0$ (with the replacements $\bY^{(k)}\sto \bcY^{(k)}$) for every $k\ge 1$, it is clear that $(\mc M,g,\Omega)$ will also satisfy $\Q_{\alpha\beta}^{(k)}=0$, $\L_{\mu}^{(k)}=0$ and $f^{(k)}=0$ to all orders because $\bY^{(k)}=\bcY^{(k)}$ and $\lie_{\xi}^{(k)}\Omega|_{\scri} = \sigma^{(k)}$ for every $k\ge 1$.\\

\underline{2. The first order}\\

The first order is special because, as already noted several times, the conformal equations have a different structural form at their lowest orders. This fact reflects itself in the way how the one-form $\bm{c}$ and the scalars $\mf{c}$, $t$ and $\sigma^{(2)}$ will be built. Let us start by defining $\bm{c} \d \bs + d\log|\sigma|$ (cf. \eqref{Q1}), and the scalars\footnote{Observe that by construction $\mf c= -\bm c(n)$ so we do not need to worry about proving this at this order, in contrast with the rest of orders.} $\mf c\d- \lie_n(\log|\sigma|)$, $\sigma^{(2)}=0$ (see \eqref{Q1}) and $\mf{c}^{(2)}\d \mf{c} n(\elltwo) -4P(\bm c,\bs)-2P(\nablacero\log|\sigma|,\nablacero\log|\sigma|)$ (which is obtained after equating \eqref{f3} to zero with the replacements $\kappa^{(2)}\sto \mf{c}^{(2)}$, $\sigma^{(2)}\sto 0$, $\br\sto\bm c$ and $\kappa_n\sto\mf c$). We can now construct the function $t$ by integrating the ODE $\dot{\L}^{(1)}=0$ (see \eqref{dotL}) with the replacements $\br$ by $\bm c$, $\kappa_n$ by $\mf c$, $\kappa^{(2)}$ by $\mf{c}^{(2)}$ and $\tr_P\bY$ by $t$, and with the initial condition $t|_{\Sigma} = \chi$. Next, if $\mf n> 3$ we construct the tensor $\bcY_{ab}$ by integrating the transport equation $\mc{Q}^{(2)}_{ab}=0$ (cf. \eqref{Q2ab}) obtained after replacing $\br$ by $\bm c$, $\kappa_n$ by $\mf{c}$, $\kappa^{(2)}$ by $\mf{c}^{(2)}$ and $\tr_P\bY$ by $t$, with the initial conditions $\bcY_{AB}^{tf}|_{\Sigma}=\mc{Y}_{AB}$, $\tr_P\bcY|_{\Sigma}=t|_{\Sigma}$ and $\bcr|_{\Sigma}=\bm c|_{\Sigma}$, namely
\begin{equation}
\label{ecuacionparaY}
\begin{aligned}
0&=(\mf n-3)\sigma \lie_n\bcY_{ab} +(\mf n-3)\sigma\mf{c} \bcY_{ab} + (\mf n-1)\nablacero_a\nablacero_b\sigma-2\sigma(\mf n-2)\nablacero_{(a}c_{b)}\\
&\quad\,  +\sigma\big(\nablacero_{(a}\s_{b)}-2c_ac_b+4c_{(a}\s_{b)}-\s_a\s_b+ \accentset{\circ}{R}_{(ab)} \big)  - \dfrac{\sigma\mf{R}}{2\mf n} \gamma_{ab},
\end{aligned}
\end{equation}
where $\mf{R}$ is given by \eqref{scal} with the corresponding replacements ($\tr_P\bY\sto t$, $\kappa^{(2)}\sto \mf{c}^{(2)}$, $\br\sto\bm c$, $\kappa_n\sto\mf c$). If $\mf n=3$ we simply define the symmetric tensor $\bcY_{ab}$ by means of the decomposition \eqref{decomp_EM} with the values $\bcY_{ab}n^b = c_a$, $\tr_P\bcY = t$ and its transverse part given by the data $\wh\bcY_{ab}$ prescribed in item (e) of the theorem. Recall that for $\mf n=3$ the equation $\Q^{(2)}_{ab}=0$ holds automatically.\\

We now prove that in the case $\mf n>3$, $\bcr=\bm c$ and $\tr_P\bcY=t$ everywhere (for $\mf n=3$ it holds by construction). Firstly, we contract equation \eqref{ecuacionparaY} with $n$ and use identity \eqref{nnablaomega} twice (with $\omega_a = c_a$ and $\omega_a=\s_a$) and $\Rcero_{(ab)}n^b = \frac{1}{2}\lie_n\s_a$ \cite{miguel3} to get
\begin{align}
0&=(\mf n-3)\sigma\lie_n\bcr_a  +(\mf n-3)\sigma\mf c \bcr_{a} + (\mf n-1)n^b \nablacero_a\nablacero_b\sigma -\sigma(\mf n-2) (\lie_n c_a - \nablacero_a\mf c +2\mf c \s_a)\nonumber\\
&\quad\,  +\sigma\lie_n\s_a+2\sigma\mf c (c_a- \s_{a}).\label{eqbmc}
\end{align}
Contracting it again with $n$ and recalling the notation $\bck_n\d-\bcr(n)=-\bcY(n,n)$,
\begin{equation}
	\label{eqmfc}
	0=-(\mf n-3)\sigma\lie_n\bck_n  -(\mf n-3)\sigma\mf c \, \bck_n + (\mf n-1)n^an^b \nablacero_a\nablacero_b\sigma +2\sigma(\mf n-2) \lie_n\mf c -2\sigma\mf{c}^2 .
\end{equation}
Additionally, the contraction of \eqref{ecuacionparaY} with $P^{ab}$ gives, after using \eqref{lietrPY},
\begin{align}
0&=(\mf n-3)\sigma\lie_n(\tr_P\bcY) + 4\sigma(\mf n-3) P(\bcr,\bs) - \sigma(\mf n-3)n(\elltwo)\bck_n +(\mf n-3)\sigma\mf{c} \tr_P\bcY+ (\mf n-1)\square_P\sigma\nonumber\\
&\quad\,-2\sigma(\mf n-2)\div_P\bm c + \sigma\big(\div_P\bs - 2P(\bm c,\bm c) + 4P(\bm c,\bm s) - P(\bm s,\bm s) + \tr_P\Rcero\big) - \dfrac{\mf n-1}{2\mf n} \sigma\mf{R}.\label{ecuaciontrazadeY}
\end{align}
 Next, we construct an auxiliary spacetime $(\mc M_1,g_1)$ using Theorem \ref{borel} from the sequence $\{\T^{(k)}_{ab}\}_{k\ge 1}$ defined by (cf. \eqref{decomp_EM}) $$\T^{(1)}_{ab} \d \dfrac{t-\mf{c} \elltwo}{\mf n-1}\gamma_{ab} + 2\wh{c}_{(a}\ell_{b)} - \mf{c}\ell_a\ell_b + \wh{\bcY}_{ab},\qquad \T^{(2)}_{ab} \d -\mf{c}^{(2)}\ell_a\ell_b,$$ 
 
 and $\T^{(k)}_{ab}=0$ for every $k\ge 3$. It is immediate to check that $P^{ab}\T^{(1)}_{ab} = t$, $\T^{(1)}_{ab}n^a =\wh{c}_b-\mf{c}\ell_b =c_b$ and $\T^{(2)}_{ab}n^a n^b = -\mf{c}^{(2)}$. We also construct a function $\Omega_1$ on $\mc M_1$ using Borel's Lemma \ref{Borel_lemma_for_functions} from the sequence $\{0,\sigma,0,0,...\}$. Since $(\mc M_1,g_1,\Omega_1)$ satisfies $\L^{(1)}_{\mu}=0$, $\Q^{(1)}_{\alpha\beta}=0$ and $f^{(1)}=f^{(2)}=f^{(3)}=0$ ($\dot\L=0$, $\dot{\Q}_a=0$, $\ddot{\Q}=0$ and $f^{(3)}=0$ hold by construction, and $\Q_{ab}=0$, $\L_a=0$ and $f^{(1)}=f^{(2)}=0$ are automatically true, see \eqref{Q1}, \eqref{La1}, \eqref{f2}), the first item in Proposition \ref{prop_Bianchi} (with $\ell=1$) ensures that it also satisfies $\Q^{(2)}_{ab}n^a=0$ and $P^{ab}\Q^{(2)}_{ab}=0$, which imply the following:
\begin{enumerate}
	\item The scalar $\mf c = -\bT^{(1)}(n,n)$ satisfies the same equation as $\bck_n$ (cf. \eqref{eqmfc}) with the replacement $\bck_n\sto \mf c$, namely
	\begin{align*}
0&=-(\mf n-3)\sigma\lie_n\mf c  -(\mf n-1)\sigma\mf c^2 + (\mf n-1)n^an^b \nablacero_a\nablacero_b\sigma +2\sigma(\mf n-2)\lie_n\mf c  .
	\end{align*}
Subtracting both equations one then arrives at $$0=(\mf n-3)\sigma\lie_n(\bck_n-\mf c)  +(\mf n-1)\sigma\mf c (\bck_n-\mf c).$$ This is a linear homogeneous ODE for $\bck_n-\mf c$, and since $\bck_n$ and $\mf c$ agree on $\Sigma$ (because of the initial data we have imposed when solving \eqref{ecuacionparaY}) we conclude $\bck_n=\mf c$ everywhere.
\item Once $\bck_n=\mf c$ is known to be true, the one-form $\bm c=\bT^{(1)}(n,\cdot)$ satisfies the same equation than $\bcr$ \eqref{eqbmc} with the replacement $\bcr\sto\bm c$, namely
\begin{align*}
	0&=(\mf n-3)\sigma\lie_n c_a  +(\mf n-3)\sigma\mf c c_a + (\mf n-1)n^b \nablacero_a\nablacero_b\sigma -\sigma(\mf n-2) (\lie_n c_a - \nablacero_a\mf c +2\mf c \s_a)\\
	&\quad\,  +\sigma\lie_n\s_a+2\sigma\mf c (c_a- \s_{a}),
\end{align*}	
so subtracting it from \eqref{eqbmc} and using that $\bcr|_{\Sigma}=\bm c|_{\Sigma}$, it follows that $\bm c = \bcr$ everywhere.
\item Finally, once we have shown that $\bm c = \bcr$, the scalar $t=P^{ab}\T^{(1)}_{ab}$ satisfies the same equation than $\tr_P\bcY$ \eqref{ecuaciontrazadeY} with the replacement $\tr_P\bcY\sto t$, namely
	\begin{align*}
		0&=(\mf n-3)\sigma\lie_nt + 4\sigma(\mf n-3) P(\bcr,\bs) - \sigma(\mf n-3)n(\elltwo)\bck_n +(\mf n-3)\sigma\mf{c} t + (\mf n-1)\square_P\sigma\\
		&\quad\, -2\sigma(\mf n-2)\div_P\bm c + \sigma\big(\div_P\bs - 2P(\bm c,\bm c) + 4P(\bm c,\bm s) - P(\bm s,\bm s) + \tr_P\Rcero\big) - \dfrac{\mf n-1}{2\mf n} \sigma\mf R.
	\end{align*}
Subtracting them and using that they coincide at $\Sigma$ we also conclude that $t=\tr_P\bcY$ everywhere.
\end{enumerate}
Note that at this point we have already constructed $\sigma^{(2)}$, $\mf{c}^{(2)}$ and the full tensor $\bcY_{ab}$. The construction guarantees that the equations $\ddot{Q}^{(1)}=0$, $f^{(3)}=0$, $\dot\Q_a^{(1)}=0$, $\dot{\L}^{(1)}=0$ and $\Q^{(2)}_{ab}=0$ are fulfilled. Recall also that $\Q^{(1)}_{ab}=0$ and $\L^{(1)}_a=0$ hold automatically. \\

\underline{3. Higher order terms}\\

For the higher order terms we apply a similar strategy. Fix $m\ge 1$ (suppose that $m+1$ is not one of the exceptional values, i.e. that $m+1\neq \frac{\mf n-1}{2}$, $m+1\neq \mf n$ and $m+1\neq \mf n-1$) and that we have already constructed $\{\bcY^{(k)},\sigma^{(k)}\}$ for all $k\le m$ and also $\sigma^{(m+1)}$ and $\mf{c}^{(m+1)}$, and that the equations $\Q_{\alpha\beta}^{(k)}=0$, $\L_{\mu}^{(k)}=0$, $f^{(k+2)}=0$ (with the replacements $\{\bY^{(k)}\sto\bcY^{(k)}\}_{k\leq m}$ and $\kappa^{(m+1)}\sto\mf{c}^{(m+1)}$) hold for every $k\le m$, and also $\Q_{ab}^{(m+1)}=0$. Note that this is what we achieved in the first order.
We now follow the steps 1.-4. that we explained in the strategy of the proof.
\begin{enumerate}
	\item We construct uniquely $\bm{c}^{(m+1)}|_{\Sigma}$ from conditions (i) $\bm{c}^{(m+1)}(n)|_{\Sigma}=-\mf{c}^{(m+1)}|_{\Sigma}$ and such that (ii) the equation $\iota^{\star}\dot{\Q}_a^{(m+1)}=0$ with the replacements $\bY^{(k)}\sto \bcY^{(k)}$ for every $k\le m$ and $\kappa^{(m+1)}\sto \mf{c}^{(m+1)}$ is fulfilled.
	\item Given $\bm{c}^{(m+1)}|_{\Sigma}$, we integrate the one-form $\bm{c}^{(m+1)}$ by solving the equation $\L_a^{(m+1)}=0$ \eqref{Lma} with the replacements $\bY^{(k)}\sto \bcY^{(k)}$ for every $k\le m$, $\br^{(m+1)}\sto \bm{c}^{(m+1)}$ and $\kappa^{(m+1)}\sto \mf{c}^{(m+1)}$, namely
\begin{equation}
	\label{propagationck}
	-\sigma \lie_nc_a^{(m+1)} + m\lie_n\sigma c_a^{(m+1)} -\sigma \nablacero_a\mf{c}^{(m+1)}-\dfrac{m}{\mf n}\mf{c}^{(m+1)} \nablacero_a\sigma = \mbox{lower order terms}.
\end{equation}
To prove that $\bm{c}^{(m+1)}(n)=-\mf{c}^{(m+1)}$ later we will need its contraction with $n$, which is 
\begin{equation}
	\label{propagationckn}
	-\sigma \lie_n\big(c_a^{(m+1)}n^a\big) + m(\lie_n\sigma)\, c_a^{(m+1)}n^a -\sigma \lie_n\mf{c}^{(m+1)}-\dfrac{m}{\mf n}\mf{c}^{(m+1)} \lie_n\sigma = \mbox{lower order terms}.
\end{equation}
\item With the $\bm{c}^{(m+1)}$ constructed in 2., we build the three functions $\sigma^{(m+2)}$, $t^{(m+1)}$ and $\mf{c}^{(m+2)}$ by solving the system of equations $f^{(m+3)}=\dot\L^{(m+1)}=\ddot\Q^{(m+1)}=0$ (see \eqref{system2}) with the replacements $\bY^{(k)}\sto \bcY^{(k)}$ for every $k\le m$, $\kappa^{(m+1)}\sto\mf{c}^{(m+1)}$, $\br^{(m+1)}\sto \bm{c}^{(m+1)}$, $\tr_P\bY^{(m+1)}\sto t^{(m+1)}$ and $\kappa^{(m+2)}\sto\mf{c}^{(m+2)}$) from the initial condition $\sigma^{(m+2)}|_{\Sigma}=\sigma^{(m+2)}_{\Sigma}$.
\item Finally, we construct $\bcY^{(m+1)}_{ab}$ by integrating the equation $\Q_{ab}^{(m+2)}=0$ with the replacements $\bY^{(k)}\sto \bcY^{(k)}$ for every $k\le m$, $\br^{(m+1)}\sto \bm{c}^{(m+1)}$, $\kappa^{(m+1)}\sto \mf{c}^{(m+1)}$, $\tr_P\bY^{(m+1)}\sto t^{(m+1)}$, $\kappa^{(m+2)}\sto\mf{c}^{(m+2)}$ and $\bY^{(m+1)}\sto \bcY^{(m+1)}$ with the initial conditions $\tr_P\bcY^{(m+1)}|_{\Sigma}=t^{(m+1)}|_{\Sigma}$, $\bcY^{(m+1)}(n,\cdot)|_{\Sigma}=\bm{c}^{(m+1)}|_{\Sigma}$ and $(\bcY^{(m+1)}_{AB})^{tf}|_{\Sigma}=\mc{Y}_{AB}^{(m+1)}$.
\end{enumerate}
Once the tensors $\{\bcY^{(k)},\sigma^{(k)}\}_{k\le m+1}$ and the functions $\sigma^{(m+2)}$ and $\mf{c}^{(m+2)}$ are built, we need to check that the equations $\Q_{\alpha\beta}^{(k)}=0$, $\L_{\mu}^{(k)}=0$ and $f^{(k+2)}=0$ hold for all $k\le m+1$, and also $\Q_{ab}^{(m+2)}=0$ (each one with the replacements $\bY^{(k)}\sto \bcY^{(k)}$ for every $k\le m+1$, and $\kappa^{(m+2)}\sto \mf{c}^{(m+2)}$). To do that it suffices to prove that (1) $\bm{c}^{(m+1)}(n)=-\mf{c}^{(m+1)}$, (2) $\bcY^{(m+1)}(n,\cdot)=\bm{c}^{(m+1)}$, and (3) $\tr_P\bcY^{(m+1)}=t^{(m+1)}$.\\

To prove the three claims we shall construct an auxiliary spacetime $(\mc M_m,g_m)$ using Theorem \ref{borel} from the sequence $\{\bT^{(k)}\}_{k\ge 1}$, with $\bT^{(k)} = \bcY^{(k)}$ for all $k\le m$, $$\T^{(m+1)}_{ab} \d \dfrac{t^{(m+1)}-\mf{c}^{(m+1)} \elltwo}{\mf n-1}\gamma_{ab} + 2\wh{c}^{(m+1)}_{(a}\ell_{b)} - \mf{c}^{(m+1)}\ell_a\ell_b + \wh{\bcY}^{(m+1)}_{ab},\qquad
\T^{(m+2)}_{ab} \d -\mf{c}^{(m+2)}\ell_a\ell_b,$$

and $\T^{(k)}_{ab}=0$ for every $k\ge m+3$. We also construct a function $\Omega_m$ on $\mc M_m$ from the sequence $\{0,\sigma,0,...,\sigma^{(m+2)},0,...\}$ using Borel's Lemma \ref{Borel_lemma_for_functions}. Note that $P^{ab}\T^{(m+1)}_{ab} = t^{(m+1)}$, $\T_{ab}^{(m+1)}n^b = \wh{c}^{(m+1)}-\mf{c}^{(m+1)}\ell_a$, $\T_{ab}^{(m+1)}n^an^b = -\mf{c}^{(m+1)}$ and $\T_{ab}^{(m+2)}n^an^b = -\mf{c}^{(m+2)}$, but we do not yet know that $\bm{c}^{(m+1)}(n)=-\mf{c}^{(m+1)}$. To establish this we note that by construction this spacetime satisfies the equations $\Q_{\alpha\beta}^{(k)}=0$, $\L_{\mu}^{(k)}=0$ for $k\leq m$ and $f^{(k)}=0$ for $k\le m+2$. So, item 1. of Corollary \ref{corolariopff} for $\ell=m$ yields firstly $\L^{(m+1)}_an^a=0$, which by \eqref{Lma} takes the form $$\sigma\lie_n \mf c^{(m+1)} -m(\lie_n\sigma)\, \mf c^{(m+1)}  -\sigma \lie_n \mf c^{(m+1)} + \dfrac{m}{\mf n}\mf c^{(m+1)}\lie_n\sigma = \mbox{lower order terms}.$$ 

The key point is that the lower order terms in this equation are, by construction, the same ones as in \eqref{propagationckn}. Thus, subtracting both equations we get a homogeneous first order ODE for $c_a^{(m+1)}n^a+\mf{c}^{(m+1)}$, and since $\bm{c}^{(m+1)}(n)|_{\Sigma} = -\mf{c}^{(m+1)}|_{\Sigma}$ we conclude $\bm{c}^{(m+1)}(n) = -\mf{c}^{(m+1)}$ everywhere.  \\

Once we have shown $\bm{c}^{(m+1)}(n) = -\mf{c}^{(m+1)}$ we also have $\T_{ab}^{(m+1)}n^b = \wh{c}^{(m+1)}-\mf{c}^{(m+1)}\ell_a=c_a^{(m+1)}$ and we are ready to prove the other two claims, namely that $\bcY^{(m+1)}(n,\cdot)=\bm{c}^{(m+1)}$ and $\tr_P\bcY^{(m+1)}=t^{(m+1)}$. First of all note that $(\mc M_m,g_m,\Omega_m)$ satisfies $\L_a^{(m+1)}=0$, $\iota^{\star}\dot\Q^{(m+1)}_a=0$, $\ddot\Q^{(m+1)}=\dot\L^{(m+1)}=f^{(m+3)}=0$ (because these were the equations we used to obtain $\bm{c}^{(m+1)}$, $\sigma^{(m+2)}$, $t^{(m+1)}$ and $\mf{c}^{(m+2)}$) and $\Q^{(m+1)}_{ab}=0$ (by assumption). Then, items 2. and 3. in Corollary \ref{corolariopff} (with $\ell=m$) imply that $(\mc M_m,g_m,\Omega_m)$ satisfies also $\Q^{(m+2)}_{ab}n^an^b=0$, $\Q^{(m+2)}_{ab}n^a=0$ and $P^{ab}\Q^{(m+2)}_{ab}=0$. From \eqref{Qabm} with $\bY^{(m+1)}\sto \bT^{(m+1)}$ we get (for the third one we use \eqref{lietrPY}, and define for shortness $N^{\mf n}_m\d \mf n-3-2m\neq 0$)
\begin{align*}
0&=-N^{\mf n}_m\sigma\lie_n\mf{c}^{(m+1)} - (m+1)N^{\mf n}_m\sigma\bck_n \mf{c}^{(m+1)} + \wt{O}^{(m+1)}_{ab}[\bm{c}^{(m+1)},\mf{c}^{(m+1)},\sigma^{(m+1)}]n^an^b + \mbox{l.o.t},\\
0&=N^{\mf n}_m\sigma\lie_n\bm{c}^{(m+1)}_b + (m+1)N^{\mf n}_m\sigma\bck_n \bm{c}^{(m+1)}_b + \wt{O}^{(m+1)}_{ab}[\bm{c}^{(m+1)},\mf{c}^{(m+1)},\sigma^{(m+1)}]n^a + \mbox{l.o.t},\\
0&=N^{\mf n}_m\sigma\big(\lie_nt^{(m+1)}+4P(\bm{c}^{(m+1)},\bs)-n(\elltwo)\mf{c}^{(m+1)}\big) + (m+1)N^{\mf n}_m\bck_n t^{(m+1)} \\
&\quad\, + \dfrac{(m+1)(\mf n-1)\sigma}{\mf n}\left(\mf{c}^{(m+2)}+2\lie_nt^{(m+1)}+2(m+1)\bck_n  t^{(m+1)}\right)\\
&\quad\, + P^{ab}\wt{O}^{(m+1)}_{ab}[\bm{c}^{(m+1)},\mf{c}^{(m+1)},\sigma^{(m+1)}]+ \mbox{l.o.t}.
\end{align*}
But the tensor $\bcY^{(m+1)}$ is the solution of the equation $\Q_{ab}^{(m+2)}=0$ with the replacements $\bY^{(k)}\sto \bcY^{(k)}$ for every $k\le m$, $\br^{(m+1)}\sto \bm{c}^{(m+1)}$, $\kappa^{(m+1)}\sto \mf{c}^{(m+1)}$, $\tr_P\bY^{(m+1)}\sto t^{(m+1)}$, $\kappa^{(m+2)}\sto\mf{c}^{(m+2)}$ and $\bY^{(m+1)}\sto \bcY^{(m+1)}$. This means that the equations $\Q^{(m+2)}_{ab}n^an^b=0$, $\Q^{(m+2)}_{ab}n^a=0$ and $P^{ab}\Q^{(m+2)}_{ab}=0$ (with the corresponding replacements) are also satisfied. Explicitly,
\begin{align*}
0&=-N^{\mf n}_m\sigma\lie_n\bck^{(m+1)} - (m+1)N^{\mf n}_m\sigma\bck_n \bck^{(m+1)} + \wt{O}^{(m+1)}_{ab}[\bm{c}^{(m+1)},\mf{c}^{(m+1)},\sigma^{(m+1)}]n^an^b + \mbox{l.o.t},\\
0&=N^{\mf n}_m\sigma\lie_n\bcr^{(m+1)}_b + (m+1)N^{\mf n}_m\sigma\bck_n \bcr^{(m+1)}_b + \wt{O}^{(m+1)}_{ab}[\bm{c}^{(m+1)},\mf{c}^{(m+1)},\sigma^{(m+1)}]n^a + \mbox{l.o.t},\\
0&=N^{\mf n}_m\sigma\big(\lie_n\big(\tr_P\bcY^{(m+1)}\big)+4P(\bcr^{(m+1)},\bs)-n(\elltwo)\bck^{(m+1)}\big) + (m+1)N^{\mf n}_m\sigma\bck_n \tr_P\bcY^{(m+1)} \\
&\quad\, + \dfrac{(m+1)(\mf n-1)\sigma}{\mf n}\left(\mf{c}^{(m+2)}+2\lie_nt^{(m+1)}+2(m+1)\bck_n  t^{(m+1)}\right)\\
&\quad\, + P^{ab}\wt{O}^{(m+1)}_{ab}[\bm{c}^{(m+1)},\mf{c}^{(m+1)},\sigma^{(m+1)}] + \mbox{l.o.t}.
\end{align*}
By subtracting both systems and recalling that $N^{\mf n}_m\neq 0$ (because we have assumed $m+1\neq\frac{\mf n-1}{2}$) and that the lower order terms agree (by construction), one arrives at a homogeneous hierarchical system of ODEs for $\mf{c}^{(m+1)}-\bck^{(m+1)}$, $\bm{c}^{(m+1)}-\bcr^{(m+1)}$ and $t^{(m+1)}-\tr_P\bcY^{(m+1)}$. Since these quantities vanish at $\Sigma$ (because of the initial conditions employed to build $\bcY^{(m+1)}_{ab}$) one concludes that they vanish everywhere.\\

Summarizing, the tensors $\bcY^{(m+1)}$, $\sigma^{(m+2)}$ and $\mf{c}^{(m+2)}$ we have just constructed satisfy the equations $\Q_{ab}^{(m+2)}=0$, $\L_a^{(m+1)}=0$, $\dot\L^{(m+1)}=\ddot\Q^{(m+1)}=f^{(m+3)}=0$ and $\dot{\wh\Q}_a^{(m+1)}|_{\Sigma}=0$. An application of item 2. in Corollary \ref{corolariopff} (with $\ell=m$) for $(\mc M_m,g_m,\Omega_m)$ shows that $\bcY^{(m+1)}$ also satisfies $\dot{\Q}_a^{(m+1)}=0$ because $\bcY^{(m+1)}_{ab}=\T^{(m+1)}_{ab}$. Then, the equations $\Q_{\alpha\beta}^{(k)}=0$, $\L_{\mu}^{(k)}=0$ and $f^{(k+2)}=0$ hold for all $k\le m+1$, and also $\Q_{ab}^{(m+2)}=0$. This closes the induction argument. To conclude the proof we only need to analyze the three exceptional cases.
 
\begin{itemize}
\item In the case $m+1=\frac{\mf n-1}{2}$ (when $\mf n$ is odd) the only thing that changes is that the equation $\mc Q_{ab}^{(\frac{\mf n+1}{2})}=0$ cannot be employed to build $\bcY^{(\frac{\mf n-1}{2})}$. Instead, we construct it using decomposition \eqref{decomp_EM} with $\bcr^{(\frac{\mf n-1}{2})}\d\bm{c}^{(\frac{\mf n-1}{2})}$, $\tr_P\bcY^{(\frac{\mf n-1}{2})}\d t^{(\frac{\mf n-1}{2})}$ and the transverse part given by the free data $\wh\bcY^{(\frac{\mf n-1}{2})}$. Note that equation $\mc Q_{ab}^{(\frac{\mf n+1}{2})}=0$ still holds by hypothesis because we have assumed $\mc{O}^{\scri}_{ab}=0$. Here obviously we do not need to prove that $\bcY^{(m+1)}(n,\cdot)=\bm{c}^{(m+1)}$ and $\tr_P\bcY^{(m+1)}=t^{(m+1)}$. The rest of the argument remains unchanged.
\item When $m+1=\mf n-1$, the problem is that the system \eqref{system2} does not determine the quantities $\{\sigma^{(\mf n)},\mf{c}^{(\mf n)},t^{(\mf n-1)}\}$, so instead we integrate the second-order transport equation \eqref{2ndorder} with the replacement $\tr_P\bY^{(\mf n-1)}\sto t^{(\mf n-1)}$ with the initial conditions $t^{(\mf n-1)}|_{\Sigma}$ and $\lie_n t^{(\mf n-1)}|_{\Sigma}$ determined from $\sigma_{\Sigma}^{(\mf n)}$ and $\mf{c}^{(\mf n)}|_{\Sigma}\d \mf{m}$ using the equations $\ddot\Q^{(\mf n-1)}|_{\Sigma}=0$ and $f^{(\mf n+1)}|_{\Sigma}=0$, exactly as explained in Remark \ref{elremarkmaslargodelmundo}. The remainder of the argument proceeds identically.
\item Finally, for $m+1=\mf n$ the issue is that the equation $\dot{\wh\Q}^{(\mf n)}_a|_{\Sigma}=0$ cannot be imposed to obtain the value of the one-form $\bm c^{(\mf n)}$ at $\Sigma$, and instead we establish $\iota^{\star}\bm c^{(\mf n)}=\bm\beta^{(\mf n)}$ and $\bm{c}^{(\mf n)}(n)|_{\Sigma}=-\mf c^{(\mf n)}|_{\Sigma}$ as initial condition. Note again that the equation $\dot{\wh{\Q}}^{(\mf n)}_a|_{\Sigma}=0$ holds by hypothesis because we have assumed $\mc{O}_a^{\Sigma}=0$. The argument continues in the same manner.
\end{itemize}

Once the full collections $\{\sigma^{(k)}\}_{k\ge 0}$ and $\{\bcY^{(k)}\}_{k\ge 1}$ have been constructed, we use Borel's Lemma \ref{Borel_lemma_for_functions} and Theorem \ref{borel} to build a function $\Omega$ and a spacetime $(\mc M,g)$ that satisfies $f^{(k)}=0$, $\Q^{(k)}_{\alpha\beta}=0$ and $\L^{(k)}_{\alpha}=0$ for all $k\ge 1$. Therefore, $(\mc M,g,\Omega)$ solves the conformal Einstein equations to infinite order at $\scri$ and realizes the initial data.
\end{proof}
\end{teo}

\begin{rmk}
	\label{remark_freedata}
The geometric objects in items a.-e. of Theorem \ref{teorema} will be referred to as the ``free data'' for the conformal Einstein equations at null infinity $\scri$.
\end{rmk}

\section{Obstruction tensors at $\scri$. Four and six dimensional cases}
\label{subsec_obs}

The purpose of this section is to study the obstruction tensors at their lowest non-trivial orders, namely the Coulombian obstruction tensor $\mc{O}_a^{\Sigma}$ in spacetime dimension four ($\mf n=3$) and the radiative obstruction tensor $\mc{O}^{\scri}_{ab}$ in spacetime dimension six ($\mf n=5$), because recall that $\Q^{(2)}_{ab}$ is identically zero in spacetime dimension four. First of all we put forward the precise definition of the obstruction tensors.
\begin{defi}
Let $(\mc M,g,\Omega)$ be an $(\mf n+1)$-dimensional conformal manifold with null infinity $\Phi:\scri\hookrightarrow\mc M$ written in a conformal gauge satisfying $|\nabla\Omega|^2=0$. Let $\xi$ be a rigging extended off $\Phi(\scri)$ geodesically, $\iota:\Sigma\hookrightarrow\scri$ a cross-section and $\mc{Q}\d (\mf n-1)\big(\hess\Omega+\Omega\sch_g\big)$. We define the radiative obstruction tensor $\mc{O}^{\scri}$ (for $\mf n$ odd) and the Coulombian obstruction tensor $\mc{O}^{\Sigma}$ (for any $\mf n$) by means of $$\mc{O}^{\scri} \d \Phi^{\star}\big(\lie_{\xi}^{(\frac{\mf n-1}{2})}\mc Q\big),\qquad \mc{O}^{\Sigma} \d \iota^{\star}\Phi^{\star}\big(\lie_{\xi}^{(\mf n-1)}\mc{Q}(\xi,\cdot)\big).$$
\end{defi}

A definition of the obstruction tensors can also be given in an arbitrary conformal gauge, but this is beyond the scope of this paper.

\subsection{Coulombian obstruction in four dimensions}

As already indicated in the previous section, the factor $(\mf n-m)$ multiplying $\br^{(m)}$ in equation $\dot\Q_a^{(m)}|_{\Sigma}=0$ leads to two important consequences when $m=\mf n$. The first one is that the equation does not constrain the value of the one-form $\wh{\br}{}^{(\mf n)}|_{\Sigma}$ (which therefore becomes free data), and the second one is that if the reminder of the equation does not vanish, either the conformal spacetime is not smooth or does not satisfy the Einstein equations beyond order $m-1$. This reminder defines the Coulombian obstruction tensor $\mc{O}_a^{\Sigma}$. In spacetime dimension four, $\mc{O}_a^{\scri}= \dot\Q^{(3)}_a$, and hence it depends on the $\scri$-structure data, $\bY$, $\bY^{(2)}$, $\sigma^{(3)}$ and $\kappa^{(3)}$. These tensors are uniquely given in terms of the free data $\{\chi,\sigma_{\Sigma}^{(3)},\mf{m},\mc{Y}_{AB}^{(1)}\}$ on $\Sigma$ and the radiation field $\wh{\bcY}_{ab}$ on $\scri$ after solving the equations $\Q^{(1)}_{\alpha\beta}=\Q^{(2)}_{\alpha\beta}=0$, $\L_{\mu}=\L^{(2)}_{\mu}=0$, $f^{(2)}=f^{(3)}=0$, $\L^{(3)}_an^a=0$ and $f^{(4)}|_{\Sigma}=0$. Note also that by Prop. \ref{prop_Bianchi} the equations $\dot\Q^{(3)}_an^a=0$, $\Q^{(3)}_{ab}n^a=0$, $P^{ab}\Q^{(3)}_{ab}=0$ and $f^{(4)}=0$ follow automatically, so in particular the obstruction tensor satisfies $\mc{O}_a^{\Sigma}n^a=0$.\\

In order to find a necessary and sufficient condition for $\mc{O}_a^{\Sigma}$ to vanish, let us note that since $\Q_{\alpha\beta}$ vanishes up to an including order 2, it suffices to study the tensor $\dot\Q_a^{(3)}$ in any gauge, since its vanishing is a gauge-invariant statement. This is a consequence of the following simple observation.
\begin{lema}
	\label{lemaobvio}
	Assume $\Q^{(k)}_{\alpha\beta}=0$ for all $k=1,...,m$. Let $\xi' = z(\xi+V)$, with $z$ and $V$ extended arbitrarily off $\scri$. Then, $\Q_{ab}^{(m+1)}{}' \d(\lie_{\xi'}^{(m)}\Q)_{ab}= z^m \Q_{ab}^{(m+1)}$.
	\begin{proof}
The result is obtained at once by inserting $\xi' = z(\xi+V)$ into $\lie_{\xi'}^{(m)}\Q_{\alpha\beta}$ and using $\Q^{(k)}_{\alpha\beta}=0$ for all $k=1,...,m$.
	\end{proof}
\end{lema}
Hence, it is sufficient to analyze $\mc{O}_a^{\Sigma}$ in a gauge in which $\sigma^{(1)}=1$, $\elltwo=0$ and the pullback of $\bm\ell$ to the cross-sections of $\scri$ vanishes, $\bm\ell_{\para}=0$. This immediately implies $\bs=0$ \cite{Mio1,Mio2} and hence $\br=\bs+d\sigma^{(1)}=0$. Moreover, the tensor $P$ at $\Sigma$ decomposes as $P^{ab} = h^{AB}e_A^ae_B^b$, where $h^{AB}$ is the inverse metric of $h_{AB}$ and $\{e_A\}$ is a basis in $\Sigma$ with dual $\{\bm\theta^A\}$, and therefore $\delta^{\alpha}_{\rho} = e_B^{\alpha}\theta_{\rho}^B + \xi^{\alpha}\nu_{\rho}+\nu^{\alpha}\xi_{\rho}$. Since the tensor $\Q$ involves up to third derivatives of the metric and the quasi-Einstein equations to second order are imposed, it is to be expected that $\Q$ may have some relation to derivatives of the Weyl tensor. We pursue this idea by applying a transverse derivative to the identity \eqref{CnablaOmega} and evaluating the result at $\scri$. Since $\Q$ vanishes up to order two, we may perform the substitution $\mc{T}_{\alpha\beta}=\frac{1}{\mf n-1}\Q_{\alpha\beta}=\frac{\Omega^2}{2(\mf n-1)}\mc{P}_{\alpha\beta}$ for some tensor $\mc{P}_{\alpha\beta}$ satisfying $\Q^{(3)}=\mc{P}^{(1)}$ at $\scri$. Rewriting identity \eqref{CnablaOmega} (recall $d=\mf n+1$) in terms of $\mc{P}_{\alpha\beta}$ gives 
\begin{align*}
	C^{\alpha}{}_{\beta\mu\nu}\nabla_{\alpha}\Omega - \dfrac{\Omega}{\mf n-2}\nabla_{\alpha} C^{\alpha}{}_{\beta\mu\nu} &= -\dfrac{\Omega}{\mf n-1}\left(\Omega\nabla_{[\mu}\mc P_{\nu]\beta} + 2 \mc P_{\beta[\nu}\nabla_{\mu]}\Omega\right)\\
	&\quad\, +\dfrac{\Omega}{\mf n(\mf n-1)}g_{\beta[\mu}\left(\Omega \nabla^{\rho}\mc P_{\nu]\rho} + 2 \mc{P}_{\rho[\nu}\nabla^{\rho}\Omega\right).
\end{align*}

Applying $\nabla_{\xi}$ (here it turns out to be more useful to take a covariant derivative along $\xi$ rather than a Lie derivative) and evaluating the result at $\Omega=0$ gives (we use that $\xi(\Omega)\st{\scri}{=}\sigma^{(1)}\st{\scri}{=}1$ and hence $\nu_{\mu}\st{\scri}{=}\nabla_{\mu}\Omega$) $$\nu_{\alpha}\xi^{\rho}\nabla_{\rho} C^{\alpha}{}_{\beta\mu\nu} + C^{\alpha}{}_{\beta\mu\nu} \xi^{\rho}\nabla_{\rho}\nabla_{\alpha}\Omega - \dfrac{1}{\mf n-2}\nabla_{\alpha} C^{\alpha}{}_{\beta\mu\nu} \st{\scri}{=} -\dfrac{2}{\mf n-1} \mc P_{\beta[\nu}\nu_{\mu]}+\dfrac{2}{\mf n(\mf n-1)} g_{\beta[\mu}\mc{P}_{\nu]\rho}\nu^{\rho}.$$ 

The second term in the left-hand side vanishes because $\xi^{\rho}\nabla_{\rho}\nabla_{\alpha}\Omega \st{\scri}{=} \frac{1}{\mf n-1}\Q^{(1)}_{\rho\alpha}\xi^{\rho}- \Omega\xi^{\rho}L_{\rho\alpha} \st{\scri}{=} 0$, so we arrive at 
\begin{equation}
	\label{bianchiatscri1}
	\nu_{\alpha}\xi^{\rho}\nabla_{\rho} C^{\alpha}{}_{\beta\mu\nu}  - \dfrac{1}{\mf n-2}\nabla_{\alpha} C^{\alpha}{}_{\beta\mu\nu} \st{\scri}{=} -\dfrac{2}{\mf n-1} \mc P_{\beta[\nu}\nu_{\mu]}+\dfrac{2}{\mf n(\mf n-1)} g_{\beta[\mu}\mc{P}_{\nu]\rho}\nu^{\rho}.
\end{equation}
We now contract this equation with $\xi^{\beta} e_A^{\mu} \xi^{\nu}$ to make the tensor $\mf{E}_{\alpha\beta} \d  \xi^{\mu} \xi^{\nu} C_{\alpha\mu\beta\nu}$ appear. The contraction of \eqref{bianchiatscri1} with $\xi^{\beta} e_A^{\mu} \xi^{\nu}$ then gives, after using $e^{\mu}_A \nu_{\mu} = e^{\mu}_A \xi_{\mu} = \xi^{\mu}\xi_{\mu}=0$, $\nu_{\alpha}\xi^{\alpha}=1$ and $\nabla_{\xi}\xi=0$, $$\nu_{\alpha} e_A^{\mu} \xi^{\rho}\nabla_{\rho} \mf{E}^{\alpha}{}_{\mu}  - \dfrac{1}{\mf n-2}\xi^{\beta} e_A^{\mu} \xi^{\nu}\nabla_{\alpha} C^{\alpha}{}_{\beta\mu\nu} \st{\scri}{=} \dfrac{1}{\mf n-1} \xi^{\beta} e_A^{\mu} \mc P_{\beta\mu}\st{\scri}{=} \dfrac{1}{\mf n-1}\dot{\Q}^{(3)}_A.$$ 

We still need to elaborate the second term in the left-hand side. Note that due to $\br=\bs=0$ and $\elltwo=0$ we have from \eqref{nablaxiup}, \eqref{V} and \eqref{Vn} that $\nabla^{\alpha}\xi^{\beta} = P^{ac}V^b{}_c e_a^{\alpha} e_b^{\beta}=P^{ac}P^{bd}(\Y_{cd}+\F_{cd})e_a^{\alpha} e_b^{\beta} = h^{AC}h^{BD}(\Y_{CD}+\F_{CD})e_A^{\alpha}e_B^{\beta}$, and hence 
\begin{align*}
\xi^{\beta} e_A^{\mu} \xi^{\nu}\nabla_{\alpha} C^{\alpha}{}_{\beta\mu\nu}& \st{\scri}{=} e_A^{\mu}\nabla_{\alpha}\mf{E}^{\alpha}{}_{\mu} - e_A^{\mu}C_{\alpha\beta\mu\nu} \big(\xi^{\beta}\nabla^{\alpha}\xi^{\nu}+\xi^{\nu}\nabla^{\alpha}\xi^{\beta}\big) \\
&\st{\scri}{=} e_A^{\mu}\nabla_{\alpha}\mf{E}^{\alpha}{}_{\mu} +\big({}^{(2)}C^B{}_A{}^C+{}^{(2)}C_A{}^{BC}\big)(\Y_{BC}+\F_{BC}),
\end{align*}
where we recall the notation in Appendix \ref{appendix} for ${}^{(2)}C_{\alpha\beta\mu}\d \xi^{\nu}C_{\alpha\nu\beta\mu}$. We finally use $\delta^{\alpha}_{\rho} = e_B^{\alpha}\theta_{\rho}^B + \xi^{\alpha}\nu_{\rho}+\nu^{\alpha}\xi_{\rho}$ in the first term of the right-hand side to get 
\begin{align*}
	e_A^{\mu}\nabla_{\alpha}\mf{E}^{\alpha}{}_{\mu} & \st{\scri}{=}  e_A^{\mu}e_B^{\alpha}\theta_{\rho}^B \nabla_{\alpha}\mf{E}^{\rho}{}_{\mu}+\nu_{\rho}e_A^{\mu}\xi^{\alpha}\nabla_{\alpha}\mf{E}^{\rho}{}_{\mu} + e_A^{\mu}\nu^{\alpha}\xi_{\rho}\nabla_{\alpha}\mf{E}^{\rho}{}_{\mu} \\
&\st{\scri}{=} \nabla^h_B\mf{E}^B{}_A + \nu_{\alpha}e_A^{\mu}\xi^{\rho}\nabla_{\rho}\mf{E}^{\alpha}{}_{\mu},
\end{align*}
where in the first term we used $e_A^{\mu}e_B^{\alpha}\theta_{\rho}^B \nabla_{\alpha}\mf{E}^{\rho}{}_{\mu} = \nabla^h_B\mf{E}^B{}_A$ (because $\mf{E}$ is orthogonal both to $\xi$ and $\nu$), and the third term vanishes because $\nu^{\alpha}\nabla_{\alpha}\xi^{\rho} \st{\scri}{=} 0 $ (by \eqref{nablaxi} and \eqref{Vn}) and $\mf{E}(\xi,\cdot)=0$. Therefore, we conclude 
\begin{equation}
	\label{eqobscoul}
	(\mf n-3) \nu_{\alpha}e_A^{\mu}\xi^{\rho}\nabla_{\rho}\mf{E}^{\alpha}{}_{\mu} - \nabla^h_B\mf{E}^B{}_A - \big({}^{(2)}C^B{}_A{}^C+{}^{(2)}C_A{}^{BC}\big)(\Y_{BC}+\F_{BC})\st{\scri}{=} \dfrac{\mf n-2}{\mf n-1}\dot{\Q}^{(3)}_A.
\end{equation}
In four spacetime dimensions ($\mf n=3$) the first term vanishes and the tensor ${}^{(2)}C_{ABC}$ is zero, because the identity \eqref{CnablaOmega} entails $C_{\alpha\beta\mu\nu}\nu^{\alpha}\st{\scri}{=}0$, so $C_{\alpha\beta\mu\nu}$ has Petrov type $N$ at $\scri$, and all its non-vanishing components at $\scri$ are encoded in $\mf{E}_{\alpha\beta}$ (see \cite{ortaggio2009bel,fernandez2022asymptotic}). Since $\dot{\Q}^{(3)}_A=\mc{O}^{\Sigma}_A$, then a necessary and sufficient condition for the Coulombian obstruction tensor to vanish is $\mf{E}_{AB}$ being divergence-free. When the cuts of $\scri$ are 2-spheres (e.g. for asymptotically simple spacetimes \cite{newman1989global}), the fact that there are no TT tensors on $\sph^2$ \cite{kroon2017conformal} imply that $\mc{O}_a^{\Sigma}=0$ if and only if $\mf{E}_{AB}=0$, and hence the full Weyl tensor vanishes at $\scri$. For other topologies of $\scri$, such as $\real\times T^2$ (see \cite{schmidt1996vacuum}), $\mf{E}_{AB}$ being divergence-free does not imply $\mf{E}_{AB}=0$. One can then construct four dimensional, smooth, asymptotically flat spacetimes (in the sense of Def. \ref{defi_AF}) as a particular case of Theorem \ref{teorema} with $\scri\not\simeq\real\times\sph^2$ whose Weyl tensor does not vanish at $\scri$. \\

Establishing existence in the stronger sense that the spacetime $\Omega^{-2} g$ is Ricci flat in a neighborhood of $\scri$ is an interesting and open problem. To the best of our knowledge, the asymptotic characteristic problem has only been solved under the assumption that the Weyl tensor vanishes at $\scri$ \cite{kannar1996existence,hilditch2020improved} (and in dimension 4). This problem is related to the fact that the conformal Einstein equations can be cast into a well-posed hyperbolic system only in spacetime dimension four \cite{friedrich2015geometric}. One interesting open problem is whether there exists a larger system of equations that is well-posed in higher dimensions.\\

Equation \eqref{eqobscoul} is interesting also in higher dimensions, because if one is interested in using the Weyl as a variable to be determined iteratively from an expansion, \eqref{eqobscoul} can be used to compute the second order term $\nu_{\alpha}e_A^{\mu}\xi^{\rho}\nabla_{\rho}\mf{E}^{\alpha}{}_{\mu}$ in terms of $\mf{E}$, ${}^{(2)}C_{ABC}$ and $\dot{\Q}^{(3)}$ provided $\mf n\neq 3$ (recall that in higher dimension the tensor ${}^{(2)}C_{ABC}$ need not to vanish, see \cite{ortaggio2009bel}). When $\mf n=3$ the Coulombian obstruction would manifest itself also in this approach.

\subsection{Radiative obstruction in six dimensions}

As already mentioned, for $\mf n$ odd the equation $\Q_{ab}^{(\frac{\mf n+1}{2})}=0$ does not fix the full tensor $\bY^{(\frac{\mf n-1}{2})}$. Furthermore, after having assumed that the previous orders are satisfied, the tensor $\Q_{ab}^{(\frac{\mf n+1}{2})}$ turns out to only depend on null metric hypersurface data, $\chi$, $\{\sigma_{\Sigma}^{(k)}\}_{k\le \frac{\mf n+1}{2}}$ and $\{\mc{Y}^{(k)}_{AB}\}_{k\le \frac{\mf n-3}{2}}$. This defines the radiative obstruction tensor $\mc{O}^{\scri}_{ab}$ whose vanishing determines whether $\Q_{ab}^{(\frac{\mf n+1}{2})}=0$ can be satisfied. This behaviour is reminiscent of the Fefferman-Graham obstruction tensor $\mc{O}^{FG}$ in the context of ambient metrics \cite{fefferman1985conformal,fefferman2012ambient}. Recall that for $\mf n=3$ the tensor $\mc{O}^{\scri}_{ab}$ automatically vanishes, just like $\mc{O}^{FG}$, and that they appear at the same order. This suggests a strong connection between $\mc{O}^{\scri}_{ab}$ and the FG obstruction tensor at the cross-sections of $\scri$. Establishing this connection would require understanding in detail all the lower order terms arising in \eqref{Qabm}. This task is challenging and well beyond the scope of this paper, so in this section we just analyze the first non-trivial case, namely $\mf n=5$ (i.e. spacetime dimension six), where $\mc{O}^{\scri}_{ab}=\Q^{(3)}_{ab}$.\\

Assume $(\mc M,g,\Omega)$ is a six-dimensional conformal manifold with $\lambda=0$ satisfying $\Q^{(1)}_{\alpha\beta}=\Q^{(2)}_{\alpha\beta}=0$, $\L^{(1)}_{\mu}=\L_{\mu}^{(2)}=0$ and $f^{(1)}=f^{(2)}=f^{(3)}=f^{(4)}=0$. By item 1. in Prop. \ref{prop_Bianchi} we know that the tensor $\Q^{(3)}_{ab}$ satisfies $P^{ab}\Q^{(3)}_{ab}=0$ and $\Q^{(3)}_{ab}n^b=0$. In Lemma \ref{lemaobvio} we have established that under any change of rigging $\xi' = z(\xi+V)$ (with $z$ and $V$ extended arbitrarily off $\scri$) one has $\Q_{ab}^{(3)}{}' = z^2 \Q_{ab}^{(3)}$. So, in order to analyze the obstruction tensor at $\scri$ it suffices to compute $\Q_{ab}^{(3)}$ in a simple gauge. We choose, as in the previous subsection, the gauge in which $\sigma^{(1)}=1$, $\elltwo=0$ and the pullback of $\bm\ell$ to the cross-sections of $\scri$ vanishes, $\bm\ell_{\para}=0$. More specifically, we shall work in Gaussian null coordinates $\{t,u,x^A\}$ in which the metric in a neighbourhood of $\scri = \{t=0\}$ takes the form $$g = 2du dt + \phi du^2 + 2\beta_A dx^A du + \mu_{AB} dx^A dx^B,$$ 

where $\phi$ and $\bm\beta$ vanish at $t=0$ and the rigging is $\xi=\partial_t$. Following the same notation as in Appendix \ref{appendixD}, we have $\gamma_{ab} = \delta_a^A\delta_b^B h_{AB}$, where $h_{AB}\d \mu_{AB}|_{t=0}$, $\kappa^{(m)}=-\frac{1}{2}\dot\phi^{(m)}$, $\r^{(m)}_A=\frac{1}{2}\dot{\bm\beta}^{(m)}$, $\Y_{AB}^{(m)}=\frac{1}{2}\dot\mu_{AB}^{(m)}$, $P^{ab}=\mu^{AB}\delta_A^a\delta_B^b$ and $n^a=\delta^a_u$. In particular, we use a prime to denote derivative w.r.t. $u$ and a dot for derivative w.r.t. $t$.\\

Computing the quasi-Einstein equation (including all terms) by hand becomes intractable very quickly. Therefore, and since we need the full expression of the quasi-Einstein equation up to order $\Q^{(3)}_{ab}$, we have performed the computation with the aid of the \texttt{xAct} package \cite{xAct2026} in \texttt{Mathematica}. The outcome of the computation has the following consequences. Firstly, equations $\Q^{(1)}_{AB}=0$, $\dot{\Q}^{(1)}_A=0$, $\dot{\Q}^{(1)}_u=0$ and $\ddot{\Q}^{(1)}=0$ at $\scri$ imply $h_{AB}'=0$, $\r_A=0$, $\kappa_n=0$ and $\sigma^{(2)}=0$, while $\Q^{(1)}_{Au}=\Q^{(1)}_{uu}=\L^{(1)}_A=\L^{(1)}_u=0$ hold automatically. Next, equation $f^{(3)}=0$ fixes $\kappa^{(2)}=0$, which inserted into $\dot{\L}^{(1)}=0$ gives $\lie_n\big(P^{AB}\Y_{AB}\big) = -\frac{R^h}{4}$. One then checks that $\mu^{AB}\Q^{(2)}_{AB}=\Q^{(2)}_{uA}=\Q^{(2)}_{uu}=0$ hold automatically, and from $\Q^{(2)}_{AB}=0$ one obtains 
 \begin{equation}
 	\label{derY1}
\lie_n \Y_{AB} = -L^h_{AB},
 \end{equation}
  where $L^h_{AB}$ is the Schouten tensor of $h_{AB}$. One can also check that $\dot{\Q}^{(2)}_u={\L}^{(2)}_u=0$ hold automatically. Equation $\dot{\Q}_A=0$ then gives $\r_A^{(2)} = \frac{1}{3}\left(D_B\Y{}^B{}_A - D_A \Y{}^B{}_B\right)$, where $D$ is the Levi-Civita derivative of $h$, and $\L^{(2)}_A=\Q^{(3)}_{Au}=\Q^{(3)}_{uu}=0$ are automatically satisfied. Next, equations $\ddot{\Q}^{(2)}=0$, $f^{(4)}=0$ and $\dot{\L}^{(2)}=0$ read, respectively, 
 \begin{align*}
 4\sigma^{(3)} + \Y^{AB}\Y_{AB} - \tr_P\bY^{(2)}&=0,\\
 \lie_n\sigma^{(3)} +\kappa^{(3)} &=0,\\
 6 \lie_n\big(\tr_P\bY^{(2)}\big) + 8\kappa^{(3)} + 6 \Y^{AB} R^h_{AB} - R^h \tr_P\bY &= 0.
 \end{align*}
Taking a derivative of the first equation w.r.t. $u$ and solving the system one obtains $\lie_n\sigma^{(3)}=\phi^{(3)}=0$ and $\lie_n(\tr_P\bY^{(2)}) = -2h^{AB}h^{CD}L^h_{AC}\Y_{BD}$. One can now check that $P^{ab}\Q^{(3)}_{ab}=0$ is automatically verified. Substituting all these expressions into the tensor $\Q^{(3)}_{AB}$ yields
\begin{align}
\Q^{(3)}_{AB}& =-\dfrac{1}{3}\Big( 6\square_h\Y_{AB} -6D^CD_{(A}\Y_{B)C} +6W^h_{ACBD}\Y^{CD}-2D_{(A}D^C\Y_{B)C}+2D_AD_B\Y^C{}_C\nonumber\\
&\qquad\quad -2 \left(\square_h \Y^C{}_C - D^CD^D\Y_{CD}\right)h_{AB}\Big),\label{Q3}
\end{align}
where $W^h_{ACBD}$ denotes the Weyl tensor of $h_{AB}$. Observe that the right-hand side is manifestly traceless. In accordance with the general results in the previous sections, equation $\Q^{(3)}_{AB} = 0$ does not determine the tensor $\Y^{(2)}_{AB}$ via a transport equation. Instead, the right-hand side of \eqref{Q3} defines a symmetric, traceless tensor 
\begin{align}
\mc{O}_{AB}^{\scri}& \d-\dfrac{1}{3}\Big( 6\square_h\Y_{AB} -6D^CD_{(A}\Y_{B)C} +6W^h_{ACBD}\Y^{CD}-2D_{(A}D^C\Y_{B)C}+2D_AD_B\Y^C{}_C\nonumber\\
		&\qquad\quad -2 \left(\square_h \Y^C{}_C - D^CD^D\Y_{CD}\right)h_{AB}\Big),\label{obstruction3}
\end{align}
constructed solely from $h_{AB}$ and $\Y_{AB}$. Taking the derivative of $\mc{O}_{AB}^{\scri}$ along $\partial_u$ and using $\partial_u h_{AB}=0$, \eqref{derY1}, and the identity $D_{(A}D^CL^h_{B)C}=h^{CD}D_AD_B L^h_{CD}$ (which follows at once form the contracted Bianchi identity, see \eqref{bianchiL}), one finds $$\partial_u \Q^{(3)}_{AB} = 2\left(\square_h L^h_{AB} - D^CD_{(A} L^h_{B)C} + W^h_{ACBD}\Y^{CD}\right) .$$ 

The term between round brackets in the right-hand side is precisely the Bach tensor of $h_{AB}$ (see e.g. \cite{fefferman2012ambient}), which we denote by $B^h_{AB}$. Summarizing, we have obtained $\partial_u \Q^{(3)}_{AB} = 2 B^h_{AB}$. The Bach tensor is precisely the Fefferman and Graham obstruction tensor in the case of conformal metrics of dimension four. This provides strong support for our expectation that $\mc{O}^{\scri}$ is closely related to the FG obstruction tensor $\mc{O}^{FG}$ of the corresponding dimension.\\ 

We now give a different argument to show that $\partial_u \Q^{(3)}_{AB}$ must be proportional to the FG obstruction tensor. In the recent work \cite{Mio6} we proved that the Fefferman–Graham ambient metric associated to the conformal class $[h]$ admits a null infinity whose $\scri$-structure is given precisely by $h$. Furthermore, we showed that the derivative of $\Q^{(3)}_{AB}$ along $\partial_u$ at $\scri$ is proportional to the Fefferman–Graham obstruction tensor of $[h]$. Since the derivative $\partial_u \Q^{(3)}_{AB}$ of \eqref{Q3} only depends on $h_{AB}$, it must be the same for all spacetimes sharing the same $h_{AB}$ at $\scri$. Thus, the only possibility is $\partial_u \Q^{(3)}_{AB}$ being proportional to the FG obstruction tensor in dimension four, i.e. the Bach tensor.\\

In summary, the vanishing of the Bach tensor of $h_{AB}$, together with the condition $\mc{O}_{AB}^{\scri}=0$ on a cross-section $\Sigma$ (which by \eqref{Q3} may be interpreted as a restriction on the free data $\Y_{AB}|_{\Sigma}$), guarantees that the full obstruction tensor $\mc{O}^{\scri}$ vanishes everywhere on $\scri$. Thus, the hypothesis $\mc{O}_{ab}^{\scri}=0$ in Theorem \ref{teorema} for $\mf n=5$ can be rewritten as $B^h_{AB}=0$ and $\mc{O}_{AB}^{\scri}|_{\Sigma}=0$. Note that the FG obstruction tensor is conformally covariant, and hence the condition $B^h_{AB}=0$ does not depend on the conformal representative of the $\scri$-structure.\\

Our recent results in \cite{Mio6} concerning a geometric characterization of conformal infinity for the Fefferman--Graham ambient metric, together with preliminary analysis of the general case, suggest that a similar picture emerges in higher dimensions, namely that the FG obstruction tensor arises after taking a sufficient number of derivatives along $n$ on the radiative obstruction tensor. We therefore expect the following conjecture to be true.

\begin{con}
	\label{conjecture}
Let $\{\scri,\bg,\bm\ell,\elltwo,\sigma,\mf q\}$ be $\scri$-structure data of odd dimension $\mf n\ge 7$ admitting a cross-section $\iota:\Sigma \hookrightarrow\scri$ with induced metric $h\d\iota^{\star}\bg$ and let $\mc{O}_{ab}^{\scri}$ be the radiative obstruction tensor. Denote by $\mc{O}^{FG}$ the Fefferman-Graham obstruction tensor of $[h]$. Then, $$\iota^{\star}\big(\text{\pounds}_n^{(\frac{\mf n-3}{2})} \mc{O}^{\scri} \big) = c_{\mf n} \, \mc{O}^{FG},$$ where $c_{\mf n}$ is a constant depending only on $\mf n$.
\end{con}

Note that the Fefferman-Graham obstruction tensor vanishes e.g. when $h$ is Einstein or conformally flat, among others \cite{fefferman2012ambient}. Establishing this conjecture would require a detailed analysis of the lower order terms appearing in \eqref{Qabm}. We intend to analyze this problem in future work.


\section{Conclusions and future work}

In this paper we have analyzed how the conformal Einstein equations constrain the geometry at null infinity without imposing any restriction on the spacetime dimension, fall-off conditions for the Weyl tensor or the topology of $\scri$ beyond admitting a cross-section. Our analysis leads to the identification of a collection of free tensors on $\scri$ that completely characterize asymptotically flat spacetimes within this framework. Moreover, we proved that, provided the obstruction tensors vanish, any such choice of free data gives rise to a smooth asymptotically flat spacetime.\\

There are several natural directions for future research. First, it would be important to clarify the physical interpretation of the free tensors $\mf{m}$ and $\bm\beta$ appearing in Theorem~\ref{teorema}, and to relate them to the notions of Bondi mass and angular momentum in higher dimensions as developed in \cite{hollands2005asymptotic, ishibashi2008higher, tanabe2011asymptotic, godazgar2012peeling, hollands2015bondi}. Addressing this question within our framework will require an appropriate conformally covariant definition of these quantities, as well as suitable higher-dimensional generalizations of the news tensor and Geroch’s $\rho$-tensor. Second, we plan to investigate further the radiative obstruction tensor in arbitrary dimension and conformal gauge and its relation to the Fefferman--Graham obstruction tensor, as proposed in Conjecture~\ref{conjecture}. In this context, it would be particularly interesting to connect our results with those of \cite{riello2024renormalization} and similar references, where logarithmic terms in the asymptotic expansions at infinity are allowed.\\

Further open problems include extending the notion of double null data introduced in \cite{Mio1, Mio2} to the asymptotic characteristic problem, incorporating the new free data identified in this work. It would be interesting to determine whether this detached object suffices to construct a conformal spacetime satisfying Einstein’s equations in a neighbourhood of $\scri$, once suitable existence theorems for the conformal equations are available, either in higher dimensions or in four dimensions with non-spherical topology. Finally, another promising direction is the inclusion of Killing and homothetic initial data, making use of the general identities developed in \cite{Mio5}, and the study of the corresponding asymptotic KID problem in this setting (see \cite{paetz2014kids}).
%

\section*{Acknowledgements}
This work has been supported by Grant PID2024-158938NB-I00 funded by MICIU/AEI/10.13039/ 501100011033 and by ``ERDF A way of making Europe''. M. Mars acknowledges financial support under projects SA097P24 (JCyL) and RED2022-134301-T funded by MCIN/AEI/10.13039/ 501100011033. G. Sánchez-Pérez also acknowledges support of the PhD. grant FPU20/03751 from Spanish Ministerio de Universidades.
\begin{appendices}
	\section{Some pullbacks into a null hypersurface}
	\label{appendix}
	
In this appendix we particularize several results from \cite{Mio3,Mio4,Mio5} to the null case and we prove additional pullback identities concerning the Hessian of a function (Prop. \ref{proppullback}). Given a $(0,p)$ tensor field $T_{\alpha_1\cdots\alpha_p}$ on $\mc M$, we use the notation $T_{a_1\cdots a_p}$ to denote its pullback to $\mc H$, and ${}^{(i,j)}T_{\alpha_1\cdots \alpha_{p-1}}$ for its contraction first in the $j$-th slot and then in the $i$-th slot, i.e. ${}^{(i,j)}T={}^{(i)}\big({}^{(j)}T\big)$ (note that the order $(i,j)$ is relevant). Similarly, ${}^{(i)}T_{a_1\cdots a_{p-1}}$ denotes the pullback of ${}^{(i)}T_{\alpha_1\cdots \alpha_{p-1}}$ to $\mc H$. The rest of the quantities are defined in Section \ref{sec_hypersurfacedata}.
	
\begin{prop}
	\label{proppullback0}
	Let $(\mc M,g)$ be a semi-Riemannian manifold and $\Phi:\mc H\hookrightarrow\mc M$ a smooth embedded null hypersurface with rigging $\xi$. Let $T$ be a $(0,p)$-tensor on $\mc M$ and $f\in\mc{F}(\mc M)$. Then, for any $j\ge 1$
	\begin{align}
		\hspace{-0.4cm}	\left(\nabla T\right)_{b a_1\cdots a_p} &= \nablacero_b T_{a_1\cdots a_p} + \sum_{i=1}^p \Y_{ba_i}T_{a_1\cdots a_{i-1} c a_{i+1}\cdots a_p} n^c  + \sum_{i=1}^p \U_{ba_i}\big({}^{(i)} T\big)_{a_1\cdots a_{i-1} a_{i+1}\cdots  a_p},\hfill\label{identity10}\\
		\hspace{-0.4cm}			\left({}^{(1)}\nabla T\right)_{a_1\cdots a_p} &= (\lie_{\xi}T)_{a_1\cdots a_p} - \sum_{i=1}^p \left(\r-\s\right)_{a_i}\big({}^{(i)} T\big)_{a_1\cdots a_{i-1} a_{i+1}\cdots a_p}- \sum_{i=1}^p V^c{}_{a_i} T_{a_1\cdots a_{i-1} c a_{i+1}\cdots a_p},\hfill\label{identity20}\\
		\left({}^{(j+1)}\nabla T\right)_{ba_1\cdots a_{p-1}} &= \nablacero_{b} \big({}^{(j)} T\big)_{a_1\cdots a_{p-1}} + \sum_{i=1}^{p-1} \Y_{b a_i} \big({}^{(j)} T\big)_{a_1\cdots a_{i-1} c a_{i+1}\cdots a_{p-1}}n^c\hfill \nonumber\\
		&  \quad\, + \sum_{i=1}^{p-1} \U_{b a_i} \big({}^{(i,j)}T\big)_{a_1\cdots a_{i-1} a_{i+1}\cdots a_{p-1}}- \left(\r-\s\right)_{b} \big({}^{(j)} T\big)_{a_1\cdots a_{p-1}} \nonumber\\
		& \quad\, - V^c{}_{b}T_{a_1\cdots a_{j-1}c a_{j}\cdots a_{p-1}}.\hfill\label{identity30}
	\end{align}
\end{prop}

\begin{prop}
	\label{propdivergencia}
	Let $(\mc M,g)$ be a semi-Riemannian manifold and $\Phi:\mc H\hookrightarrow\mc M$ a smooth null hypersurface with rigging $\xi$. Let $T$ be a $(0,p+1)$-tensor on $\mc M$ and denote by $\div T$ the $p$-covariant tensor defined by $\left(\div T\right)_{\alpha_1\cdots\alpha_p}\d g^{\mu\nu}\nabla_{\mu}T_{\nu\alpha_1\cdots\alpha_p}$. Then,
	\begin{equation*}
		\begin{aligned}
			\hspace{-0.3cm}\Phi^{\star}\left(\div T\right)_{a_1\cdots a_p} &=  n^b(\lie_{\xi} T)_{ba_1\cdots a_p}+P^{bc}\nablacero_b T_{c a_1\cdots a_p}  + n^b\nablacero_b \big({}^{(1)} T\big)_{a_1\cdots a_p} + \big(2\kappa_n + \tr_P\bU\big) \big({}^{(1)} T\big)_{a_1\cdots a_p}\\
			&\quad\, + (\tr_P\bY - n(\elltwo))n^cT_{c a_1\cdots a_p} - 2P^{ac}(\r+\s)_a  T_{c a_1\cdots a_p} \\
			&\quad\,  + \sum_{i=1}^p P^{bc}\Y_{ba_i}T_{c a_1\cdots a_{i-1} d a_{i+1}\cdots a_p}n^d + \sum_{i=1}^p P^{bc}\U_{ba_i} \big({}^{(i+1)} T\big)_{c a_1\cdots a_{i-1} a_{i+1}\cdots a_p} \\
			&\quad\,  - \sum_{i=1}^p (\r-\s)_{a_i}n^b\, \big({}^{(i+1)} T\big)_{b a_1\cdots a_{i-1} a_{i+1}\cdots a_p}- \sum_{i=1}^p V^c{}_{a_i}n^b T_{b a_1\cdots a_{i-1}c a_{i+1}\cdots a_p}\\
			&\quad\,  + \sum_{i=1}^p \r_{a_i} \big({}^{(1)} T\big)_{a_1\cdots a_{i-1}b a_{i+1}\cdots a_p}n^b .
		\end{aligned}
	\end{equation*}
\end{prop}

\begin{prop}
	\label{proppullback3}
	Let $(\mc M,g)$ be a semi-Riemannian manifold, $\Phi:\mc H\hookrightarrow\mc M$ a smooth embedded null hypersurface with geodesic rigging $\xi$ and let $T$ be a $(0,p)$-tensor on $\mc M$. Then,
	\begin{align}
	e_b^{\mu}\xi^{\alpha_1}\cdots \xi^{\alpha_p}\nabla_{\mu}T_{\alpha_1\cdots \alpha_p} &\st{\mc H}{=}\nablacero_b\big(T(\xi,\cdots,\xi)\big) - p T(\xi,\cdots,\xi) \left(\r-\s\right)_b \nonumber\\
	&\quad\, -\sum_{i=1}^p V^{c}{}_b \big({}^{(1,...,i-1,i+1,...p)}T\big)_{c}, \label{a2id1}\\
	\xi^{\mu}\xi^{\alpha_1}\cdots \xi^{\alpha_p}\nabla_{\mu}T_{\alpha_1\cdots \alpha_p} &\st{\mc H}{=} \big(\lie_{\xi}T\big)(\xi,\cdots,\xi) ,\label{a2id2}\\
	\xi^{\mu}\xi^{\nu} e^{\alpha_1}_{a_1}\cdots e^{\alpha_{p-1}}_{a_{p-1}}\nabla_{\mu}T_{\nu\alpha_1\cdots \alpha_{p-1}} &\st{\mc H}{=} \big(\lie_{\xi}\big({}^{(1)}T\big)\big)_{a_1\cdots a_{p-1}}- \sum_{i=1}^{p-1}V^b{}_{a_i} \big({}^{(1)}T\big)_{a_1\cdots a_{i-1} b a_{i+1}\cdots a_{p-1}}\nonumber\\
	&\quad\,  -\sum_{i=1}^{p-1}\left(\r-\s\right)_{a_i}\big({}^{(i,1)}T\big)_{a_1\cdots a_{i-1} a_{i+1}\cdots  a_{p-1}}.\label{a2id3}
	\end{align}
\end{prop}	

	\begin{prop}
		\label{proppullback}
		Let $(\mc M,g)$ be a semi-Riemannian manifold and $\Phi:\mc H\hookrightarrow\mc M$ a smooth null hypersurface with geodesic rigging $\xi$. Let $f$ be function on $\mc M$. Then,
		\begin{align}
			\Phi^{\star}\left(\hess f\right)_{ab} &= \nablacero_a\nablacero_b f + \lie_n(f) \Y_{ab} + \lie_{\xi}(f) \U_{ab},\hfill\label{identity1}\\
			\Phi^{\star}\big((\hess f)(\xi,\cdot)\big)_a &= \nablacero_a \big(\lie_{\xi}f\big) - \lie_{\xi}f (\r-\s)_a - V^b{}_a\nablacero_b f,\hfill\label{identity2}\\
			\Phi^{\star}\big(\big(\hess f\big)(\xi,\xi)\big) &= \lie^{(2)}_{\xi}f .\hfill\label{identity3}
		\end{align}
 As a consequence,
		\begin{equation}
			\label{box}
\square_g f \st{\H}{=} \square_P f + \big(\tr_P\bY-n(\elltwo)\big) \lie_n f +\left(\tr_P\bU+2\kappa_n\right)\lie_{\xi}f +2\lie_n\big(\lie_{\xi}f\big) - 2P^{ab}(\r+\s)_a \nablacero_b f,
		\end{equation}
		where $\square_P \d P^{ab}\nablacero_a\nablacero_b$.
		\begin{proof}
Identity \eqref{identity1} was already proven in \cite{Mio5}. To prove the second one we contract $\nabla_{\alpha}\nabla_{\beta}f $ with $e_a^{\alpha} \xi^{\beta}$ and use \eqref{nablaxi}, $$e_a^{\alpha} \xi^{\beta}\nabla_{\alpha}\nabla_{\beta}f = \nabla_{e_a} (\lie_{\xi}f) - (\nabla_{e_a}\xi)^{\beta}\nabla_{\beta} f = \nablacero_{e_a}(\lie_{\xi}f) -(\r_a-\s_a) \lie_{\xi}f - e_a^c V^b{}_c\nablacero_b f.$$

 Expression \eqref{identity3} is immediate because $\xi^{\alpha}\xi^{\beta}\nabla_{\alpha}\nabla_{\beta} f = \nabla_{\xi}\nabla_{\xi} f - \nabla_{\nabla_{\xi}\xi} f = \lie^{(2)}_{\xi}f - \lie_{a_{\xi}}f$. Finally, identity \eqref{box} follows from (see \eqref{inversemetric}) $g^{\alpha\beta}\nabla_{\alpha}\nabla_{\beta}f = P^{ab}e_a^{\alpha}e_b^{\beta} \nabla_{\alpha}\nabla_{\beta}f + 2 \xi^{\alpha}\nu^{\beta} \nabla_{\alpha}\nabla_{\beta} f $ after inserting \eqref{identity1}-\eqref{identity2} and using \eqref{Vn}.
		\end{proof}
	\end{prop}

\section{Auxiliary computations and proof of \eqref{dotR}}
\label{appendixB}

In this appendix we compute several contractions of the tensors $\uwh\Sigma^{(m)}$ and $\Sigma^{(m)}$ that will be used both to prove \eqref{dotR} and in Appendix \ref{appendixC} below. Let us begin by recalling the following expressions computed in \cite[Prop. 4.21]{Mio3},
\begin{equation}
	\label{Sigmamcab}
\uwh\Sigma^{(m)}_{cab} \st{(m)}{=} \nablacero_a\Y^{(m)}_{bc} +\nablacero_b\Y^{(m)}_{ac} -\nablacero_c\Y^{(m)}_{ab} + 2\r^{(m)}_c\Y_{ab} + 2\r_{c} \Y^{(m)}_{ab},
\end{equation}
	\begin{minipage}{0.6\textwidth}
	\noindent
	\begin{equation}
		\label{apC5}
\big({}^{(3)}\uwh\Sigma^{(m)}\big)_{ab} \st{(m+1)}{=} -({}^{(1)}\uwh\Sigma^{(m)})_{ab} \st{(m+1)}{=} \Y^{(m+1)}_{ab},
	\end{equation}
\end{minipage}
\begin{minipage}{0.4\textwidth}
	\noindent
\begin{equation}
	\label{apC5bis}
	\xi^{\beta}\Sigma^{(m)}{}^{\mu}{}_{\mu \beta}\st{(m+1)}{=} \tr_P\bY^{(m+1)}.
\end{equation}
\end{minipage}
\,

The next expressions also proved in \cite{Mio3} (formulas (65), (67) and (68)) and valid for $m\ge 1$ will be also needed
\begin{multicols}{2}
	\noindent
\begin{equation}
	\label{apC}
	\Sigma^{(m)}_{\mu\alpha\beta} \st{[m]}{=} \uwh\Sigma^{(m)}_{\mu\alpha\beta} - (m-1) \mc{K}^{\rho}{}_{\mu}\uwh\Sigma^{(m-1)}_{\rho\alpha\beta},
\end{equation} 
\begin{equation}
	\label{apCextra}
	\Sigma_{\rho\alpha\beta}^{(m)} \st{[m+1]}{=} \uwh\Sigma_{\rho\alpha\beta}^{(m)},
\end{equation}
\end{multicols}
\vspace{-0.5cm}
\begin{equation}
	\label{apcextra2}
\uwh{\Sigma}^{(m)}_{\mu\alpha\beta}\st{[m]}{=} \nabla_{(\alpha}\mc{K}^{(m)}_{\beta)\mu}-\dfrac{1}{2}\nabla_{\mu}\mc{K}^{(m)}_{\alpha\beta} - \mc{K}^{\eps}{}_{\mu}\nabla_{(\alpha}\mc{K}^{(m-1)}_{\beta)\eps}+\dfrac{1}{2}\mc{K}^{\eps}{}_{\mu}\nabla_{\eps}\mc{K}^{(m-1)}_{\alpha\beta}.
\end{equation}
Note that these three identities become exact for $m=1$. Later in this appendix we will also need several contractions of $\xi^{\beta}\uwh{\Sigma}^{(m)}_{\mu\alpha\beta}$ that we compute next. First observe that the contraction of \eqref{apcextra2} with $\xi^{\beta}$ can be written as $$\xi^{\beta}\uwh{\Sigma}^{(m)}_{\mu\alpha\beta} \st{[m]}{=} \dfrac{1}{2}\xi^{\beta}\nabla_{\beta}\mc{K}^{(m)}_{\alpha\mu}+\xi^{\beta}\nabla_{[\alpha}\mc{K}_{\mu]\beta}^{(m)}-\dfrac{1}{2}\mc{K}^{\eps}{}_{\mu}\xi^{\beta}\nabla_{\beta}\mc{K}^{(m-1)}_{\alpha\eps}-\mc{K}^{\eps}{}_{\mu}\xi^{\beta}\nabla_{[\alpha}\mc{K}_{\eps]\beta}^{(m-1)}.$$ 

We now use the two identities \eqref{Kup} and get (the notation introduced in Appendix \ref{appendix} also applies here)
\begin{align}
e^{\mu}_c e^{\alpha}_a \xi^{\beta}\uwh{\Sigma}^{(m)}_{\mu\alpha\beta} &\st{(m)}{=} \dfrac{1}{2}\big({}^{(1)}\nabla\mc{K}^{(m)}\big)_{ac} + \big({}^{(2)}\nabla\mc{K}^{(m)}\big)_{[ac]}- \r_c \big({}^{(2,3)}\nabla\mc{K}^{(m-1)}\big)_{a} \nonumber\\
&\quad\,  + \left(2P^{bd}\Y_{bc}+\dfrac{1}{2}n^d\nablacero_c\elltwo\right)\left(\big({}^{(2)}\nabla\mc{K}^{(m-1)}\big)_{[da]}-\dfrac{1}{2}\big({}^{(1)}\nabla\mc{K}^{(m-1)}\big)_{ad}\right),\label{3Sigma} \\
e^{\mu}_c \xi^{\alpha} \xi^{\beta}\uwh{\Sigma}^{(m)}_{\mu\alpha\beta}&\st{(m)}{=} \big({}^{(1,2)}\nabla\mc K^{(m)}\big)_c - \dfrac{1}{2}\big({}^{(2,3)}\nabla\mc K^{(m)}\big)_c -\r_c \big({}^{(1,2,3)}\nabla\mc K^{(m-1)}\big) \nonumber\\
&\quad\, -\left(2P^{bd}\Y_{bc}+\dfrac{1}{2}n^d\nablacero_c\elltwo\right)\left(\big({}^{(1,2)}\nabla\mc K^{(m-1)}\big)_d-\dfrac{1}{2}\big({}^{(2,3)}\nabla\mc K^{(m-1)}\big)_d\right)\label{23Sigma}\\
\xi^{\mu} e^{\alpha}_a \xi^{\beta}\uwh{\Sigma}^{(m)}_{\mu\alpha\beta}&\st{(m)}{=} \dfrac{1}{2}\big({}^{(2,3)}\nabla\mc K^{(m)}\big)_a +\dfrac{1}{2}P^{bc}\nablacero_b\elltwo \big({}^{(2)}\nabla\mc{K}^{(m-1)}\big)_{[ca]}-\dfrac{1}{4}P^{bc}\nablacero_b\elltwo \big({}^{(1)}\nabla\mc{K}^{(m-1)}\big)_{ac}\nonumber\\
&\quad\,  -\dfrac{1}{4} n(\elltwo)\big({}^{(2,3)}\nabla\mc{K}^{(m-1)}\big)_{a}.\label{13Sigma}
\end{align}

Taking trace in \eqref{apC} in the indices $\mu,\alpha$ and using \eqref{apcextra2} one gets
\begin{equation}
	\label{Sigmamup}
\Sigma^{(m)\mu}{}_{\mu\beta} \st{[m]}{=} \dfrac{1}{2}\left(\nabla_{\beta}\mc{K}^{(m)\mu}{}_{\mu}-m\mc{K}^{\mu\rho}\nabla_{\beta}\mc{K}^{(m-1)}_{\mu\rho}\right).
\end{equation} 
Finally, the $m-1$ Lie derivative of the Ricci tensor is \cite[Eq. (63)]{Mio3}, for $m\ge 2$,
\begin{equation}
	\label{liericci}
	\lie_{\xi}^{(m-1)}R_{\alpha\beta} \st{[m]}{=} \nabla_{\mu}\Sigma^{(m-1)\mu}{}_{\alpha\beta} - \nabla_{\beta}\Sigma^{(m-1)\mu}{}_{\mu\alpha}.
\end{equation}

\begin{prop}
	\label{propB3}
Let $\{\H,\bg,\bm\ell,\elltwo\}$ be null metric hypersurface data $(\Phi,\xi)$-embedded in $(\mc M,g)$ and extend $\xi$ off $\Phi(\H)$ by $\nabla_{\xi}\xi=0$. Then, for any $m\ge 2$,
\begin{equation}
	\label{apc0}
\Sigma_{cab}^{(m)} \st{(m)}{=} \nablacero_a\Y^{(m)}_{bc} +\nablacero_b\Y^{(m)}_{ac} -\nablacero_c\Y^{(m)}_{ab} + 2\r^{(m)}_c\Y_{ab} +2m\r_{c} \Y^{(m)}_{ab},
\end{equation}
\begin{multicols}{2}
	\noindent
	\begin{equation}
		\label{apC9}
		e_a^{\beta}\Sigma^{(m)}{}^{\mu}{}_{\mu\beta} \st{(m)}{=}  \nablacero_a \big(\tr_P\bY^{(m)}\big),
	\end{equation}
	\begin{equation}
	\label{apC6}
	\big({}^{(2,3)}\uwh\Sigma^{(m)}\big)_a \st{(m+1)}{=}\big({}^{(1,3)}\uwh\Sigma^{(m)}\big)_a \st{(m+1)}{=} 0,
\end{equation}
\end{multicols}
\vspace{-0.8cm}
\begin{align}
\big({}^{(3)}\uwh\Sigma^{(m)}\big)_{ca}& \st{(m)}{=} \Y_{ac}^{(m+1)} -2V^b{}_a \Y_{cb}^{(m)} - 2P^{bd}\Y_{bc}\Y^{(m)}_{ad}-\dfrac{1}{2}\r^{(m)}_a\nablacero_c\elltwo,\label{3uwhSigmaca}\\
\big({}^{(3)}\Sigma^{(m)}\big)_{ca}& \st{(m)}{=} \Y_{ac}^{(m+1)} -2V^b{}_a \Y_{cb}^{(m)} - 2mP^{bd}\Y_{bc}\Y^{(m)}_{ad}-\left(m-\dfrac{1}{2}\right)\r^{(m)}_a\nablacero_c\elltwo.\label{3Sigmaca}
\end{align}
As a consequence,
\begin{align}
	n^c \Sigma_{cab}^{(m)} &\st{(m)}{=} 2\nablacero_{(a}\r^{(m)}_{b)} - \lie_n\Y^{(m)}_{ab} - 2\kappa^{(m)}\Y_{ab} -2m \kappa_n \Y^{(m)}_{ab},\label{apC1}\\
 ({}^{(3)}\uwh{\Sigma}^{(m)})_{ca}n^c &\st{(m)}{=} \r_a^{(m+1)} -2V^b{}_a \r_b^{(m)} - 2P^{bc}\r_b \Y^{(m)}_{ac} - \dfrac{1}{2}n(\elltwo)\r_a^{(m)} ,\label{apC8}\\
\big({}^{(3)}\Sigma^{(m)}\big)_{ca}n^c &\st{(m)}{=} \r_a^{(m+1)} - 2V^b{}_a \r_b^{(m)} - 2mP^{bc}\r_b \Y^{(m)}_{ac}   - \dfrac{m}{2}n(\elltwo)\r_a^{(m)}.\label{apC10}
\end{align}
\begin{proof}
Identity \eqref{apc0} will be a consequence of \eqref{Sigmamcab} and \eqref{apC}. The pullback of the second is, after taking into account \eqref{Kup}, \eqref{apC5} and $\uwh\Sigma^{(m-1)}_{cab}\st{(m)}{=} 0$ (by \eqref{Sigmamcab}), $$\Sigma^{(m)}_{cab} \st{(m)}{=} \uwh\Sigma^{(m)}_{cab}+ 2(m-1)\r_c \Y^{(m)}_{ab}.$$ 

Replacing here $\uwh\Sigma^{(m)}_{cab}$ from \eqref{Sigmamcab} yields \eqref{apc0}. The contraction of \eqref{apc0} with $n^c$ gives \eqref{apC1} after using $$n^c\big(\nablacero_a\Y^{(m)}_{bc}+\nablacero_b\Y^{(m)}_{ac}-\nablacero_c\Y^{(m)}_{ab}\big) = 2\nablacero_{(a}\r^{(m)}_{b)} -\lie_n\Y^{(m)}_{ab}.$$


To show \eqref{apC9} we simply contract \eqref{Sigmamup} with $e_a^{\beta}$ and use \eqref{trazas} (recall that a tangential derivative of $\mc{K}^{(m-1)}$ is at most of transverse order $m-1$).\\

 The remaining identities will rely on \eqref{3Sigma}-\eqref{13Sigma}. Before applying them we need to determine the terms of the form $({}^{(i)}\nabla\mc K^{(m)})_{ab}$ and $({}^{(i,j)}\nabla\mc K^{(m)})_{a}$ for various values of $i,j$. For the former we use \eqref{identity20}-\eqref{identity30}, and for the latter \eqref{a2id1}-\eqref{a2id3}. Taking also into account \eqref{Kxi} the result is
 \begin{gather}
 	\big({}^{(1)}\nabla\mc{K}^{(m)}\big)_{ac} \st{(m)}{=}\mc{K}^{(m+1)}_{ac} - 2V^b{}_{(a}\mc{K}^{(m)}_{c)b},\qquad \big({}^{(2)}\nabla\mc{K}^{(m)}\big)_{[ac]} \st{(m)}{=}V^b{}_{[c}\mc{K}^{(m)}_{a]b},\label{aux1301}\\
 	\big({}^{(1,2)}\nabla\mc{K}^{(m)}\big)_{c} \st{(m)}{=}0,\qquad \big({}^{(2,3)}\nabla\mc{K}^{(m)}\big)_{c} \st{(m)}{=}0,\qquad \big({}^{(1,2,3)}\nabla\mc{K}^{(m-1)}\big) \st{(m)}{=}0.\nonumber
 \end{gather} 
Note that these expressions immediately imply 
\begin{equation}
	\label{aux1302}
\begin{gathered}
	\big({}^{(1)}\nabla\mc{K}^{(m-1)}\big)_{ad} \st{(m)}{=} \mc{K}^{(m)}_{ad},\qquad \big({}^{(2)}\nabla\mc{K}^{(m-1)}\big)_{[da]} \st{(m)}{=}0 ,\\
	 \big({}^{(2,3)}\nabla\mc{K}^{(m-1)}\big)_{a} \st{(m)}{=} \big({}^{(1,2)}\nabla\mc{K}^{(m-1)}\big)_{c} \st{(m)}{=}0.
	 \end{gathered}
\end{equation} 
With all these expressions at hand, identities \eqref{23Sigma} and \eqref{13Sigma} become, respectively,
\begin{align*}
\big({}^{(2,3)}\uwh\Sigma^{(m)}\big)_c \st{(m+1)}{=} 0,\qquad \big({}^{(1,3)}\uwh\Sigma^{(m)}\big)_c \st{(m+1)}{=} 0,
\end{align*}
which proves \eqref{apC6}. Replacing \eqref{aux1301}-\eqref{aux1302} into \eqref{3Sigma} and substituting $\mc{K}^{(m)}_{ab}=2\Y^{(m)}_{ab}$ gives
\begin{align*}
\big({}^{(3)}\uwh\Sigma^{(m)}\big)_{ca} & \st{(m)}{=} \dfrac{1}{2}\mc{K}^{(m+1)}_{ac} - V^b{}_{(a}\mc{K}^{(m)}_{c)b} +V^b{}_{[c}\mc{K}^{(m)}_{a]b}-\dfrac{1}{2}\left(2P^{bd}\Y_{bc}+\dfrac{1}{2}n^d\nablacero_c\elltwo\right)\mc{K}^{(m)}_{ad}\\
& \st{(m)}{=} \Y_{ac}^{(m+1)} -2V^b{}_a \Y_{cb}^{(m)} - 2P^{bd}\Y_{bc}\Y^{(m)}_{ad}-\dfrac{1}{2}\r^{(m)}_a\nablacero_c\elltwo,
\end{align*}
which is \eqref{3uwhSigmaca}. A contraction with $n^c$ gives \eqref{apC8}. Now, identity \eqref{3Sigmaca} is a consequence of replacing \eqref{apC6}, \eqref{3uwhSigmaca} and \eqref{Kup} into \eqref{apC}, namely
\begin{align*}
	e^{\mu}_c e^{\alpha}_a \xi^{\beta} \Sigma_{\mu\alpha\beta}^{(m)}&\st{(m)}{=} e^{\mu}_c e^{\alpha}_a \xi^{\beta} \uwh\Sigma_{\mu\alpha\beta}^{(m)} - (m-1) \mc{K}_{\mu}{}^{\rho} e^{\mu}_c e^{\alpha}_a \xi^{\beta} \uwh\Sigma_{\rho\alpha\beta}^{(m-1)}\\
	&\st{(m)}{=} \big({}^{(3)}\uwh\Sigma^{(m)}\big)_{ca} - (m-1) \left(2P^{bd}\Y_{cb}+\dfrac{1}{2}n^d\nablacero_c\elltwo\right) \big({}^{(3)}\uwh\Sigma^{(m-1)}\big)_{da}-2\r_c\big({}^{(1,3)}\Sigma^{(m-1)}\big)_a\\
	& \st{(m)}{=} \Y_{ac}^{(m+1)} -2V^b{}_a \Y_{cb}^{(m)} - 2mP^{bd}\Y_{bc}\Y^{(m)}_{ad}-\dfrac{m}{2}\r^{(m)}_a\nablacero_c\elltwo.
\end{align*}
Finally, \eqref{apC10} is its contraction with $n^c$. 
\end{proof}
\end{prop}

We already have all the necessary ingredients to compute $\dot{\mc R}^{(m)}_a$ up to order $m$.

\begin{prop}
	\label{propliericxitang}
	Let $\mc H$ be a null hypersurface $(\Phi,\xi)$-embedded in $(\mc M,g)$ and extend $\xi$ off $\Phi(\mc H)$ by $\nabla_{\xi}\xi=0$. Then for any $m\ge 2$, 
	\begin{equation}
		\label{liericxiX}
		\begin{aligned}
			\dot{\mc R}_a^{(m)} &\st{(m)}{=}  \r^{(m+1)}_a + P^{bc} \nablacero_b\Y^{(m)}_{ac} -2m P^{bc}\r_b \Y^{(m)}_{ac} - \nablacero_a\big(\tr_P\bY^{(m)}\big)  + \big(\tr_P\bY^{(m)}\big)(\r_a-\s_a) \\
			&\quad\,  + \big(P^{bc}\Y_{ab}- 3V^c{}_a\big)\r^{(m)}_c + \left(\tr_P\bY-\dfrac{m}{2}n(\elltwo)\right)\r^{(m)}_a ,
		\end{aligned}
	\end{equation}
	and consequently
	\begin{equation}
		\label{liericxin}
		\begin{aligned}
			\dot{\mc R}_a^{(m)}n^a &\st{(m)}{=} -\kappa^{(m+1)} + P^{bc}\nablacero_b\r^{(m)}_c - P^{bc}P^{ad} \U_{bd}\Y^{(m)}_{ac}  - \lie_n\big(\tr_P\bY^{(m)}\big)- \kappa_n \tr_P\bY^{(m)}\\
			& \qquad\, - 2P^{bc}\big((m+1)\r_b+2 \s_b\big)\r^{(m)}_c +\left(\dfrac{m+3}{2}n(\elltwo)-\tr_P\bY\right)\kappa^{(m)}.
		\end{aligned}
	\end{equation}	
	\begin{proof}
From \eqref{liericci},
		\begin{align}
			\wh e^{\alpha}_a\xi^{\beta}\lie_{\xi}^{(m-1)}R_{\alpha\beta} &\st{[m]}{=} \wh e^{\alpha}_a\xi^{\beta}\nabla_{\mu}\Sigma^{(m-1)}{}^{\mu}{}_{\alpha\beta} - \wh e^{\alpha}_a\xi^{\beta} \nabla_{\beta} \Sigma^{(m-1)}{}^{\mu}{}_{\mu\alpha}\nonumber\\
			&\st{[m]}{=} \wh e^{\alpha}_a \nabla_{\mu} \big(\xi^{\beta}\Sigma^{(m-1)}{}^{\mu}{}_{\alpha\beta}\big) - \wh e^{\alpha}_a\Sigma^{(m-1)}{}^{\mu}{}_{\alpha\beta} \nabla_{\mu}\xi^{\beta} - \wh e^{\alpha}_a\xi^{\beta}\nabla_{\beta}\big(\Sigma^{(m-1)}{}^{\mu}{}_{\mu\alpha}\big).\label{Rmordenmas}
		\end{align}
We evaluate each term separately. For the first we apply Prop. \ref{propdivergencia} to $T={}^{(3)}\Sigma^{(m-1)}$ and take into account $\big({}^{(3)}\Sigma^{(m-1)}\big)_{ab} \st{(m)}{=} \Y^{(m)}_{ab}$ (by \eqref{apC5} and \eqref{apCextra}) and ${}^{(2,3)}\Sigma^{(m-1)}\st{(m)}{=}{}^{(1,3)}\Sigma^{(m-1)}\st{(m)}{=} 0$ (by \eqref{apC6} and \eqref{apCextra}). Thus,
\begin{equation}
	\label{auxiliar0}
	\begin{aligned}
\wh e^{\alpha}_a \nabla_{\mu} \big(\xi^{\beta}\Sigma^{(m-1)}{}^{\mu}{}_{\alpha\beta}\big) &\st{(m)}{=}  n^b \wh e^{\alpha}_a \wh e^{\mu}_b \xi^{\beta}\lie_{\xi}\Sigma^{(m-1)}_{\mu\alpha\beta} +P^{bc}\nablacero_b \Y^{(m)}_{ca} + \big(\tr_P\bY-n(\elltwo)\big)\r^{(m)}_a\\
&\qquad\, -2P^{bc}(\r+\s)_b\Y^{(m)}_{ca} + P^{bc}\Y_{ba}\r^{(m)}_c -V^c{}_a \r_c^{(m)}.		
	\end{aligned}
\end{equation}
Only the first term needs further analysis. Note that $$\lie_{\xi}\Sigma^{(m-1)}_{\mu\alpha\beta} = \lie_{\xi}\big(g_{\mu\nu}\Sigma^{(m-1)\nu}{}_{\alpha\beta}\big) = \mc{K}_{\mu\nu}\Sigma^{(m-1)\nu}{}_{\alpha\beta} + g_{\mu\nu} \Sigma^{(m)\nu}{}_{\alpha\beta} \stackbin[\eqref{apC}]{[m]}{=}\uwh\Sigma_{\mu\alpha\beta}^{(m)}-(m-2)\mc{K}_{\mu}{}^{\rho}\uwh\Sigma^{(m-1)}_{\rho\alpha\beta},$$ so 
\begin{equation}
	\label{auxiliar}
	n^b \wh e_a^{\alpha} \wh e_{b}^{\mu} \xi^{\beta} \lie_{\xi}\Sigma^{(m-1)}_{\mu\alpha\beta} \st{[m]}{=} n^b \wh e_a^{\alpha} \wh e_{b}^{\mu} \xi^{\beta}\uwh{\Sigma}^{(m)}_{\mu\alpha\beta} - (m-2)n^b \wh e_a^{\alpha} \wh e_{b}^{\mu} \xi^{\beta}  \mc{K}_{\mu}{}^{\rho}\uwh{\Sigma}^{(m-1)}_{\rho\alpha\beta}.
\end{equation} 
The first term in the right hand side is directly given by \eqref{apC8}, and for the second one we use
		\begin{align*}
			n^b \wh e_a^{\alpha} \wh e_{b}^{\mu} \xi^{\beta} \mc{K}_{\mu}{}^{\rho}\uwh{\Sigma}^{(m-1)}_{\rho\alpha\beta} & \stb[\eqref{Kup}]{}{=}\left(2 P^{cd} \r_c  + \dfrac{1}{2}n(\elltwo) n^d\right) \big({}^{(3)}\uwh\Sigma^{(m-1)}\big)_{da}  +2\r_b \big({}^{(1,3)}\uwh\Sigma^{(m-1)}\big)_a\\
			&\stb[\eqref{apC5}]{(m)}{=} 2P^{cd}\r_c \Y^{(m)}_{da} + \dfrac{1}{2}n(\elltwo) \r^{(m)}_a ,
		\end{align*}
		where in the second line we also used \eqref{apC6}. Replacing this and \eqref{apC8} into \eqref{auxiliar} gives
		\begin{align*}
			n^b \wh e_a^{\alpha} \wh e_{b}^{\mu} \xi^{\beta} \lie_{\xi}\Sigma^{(m-1)}_{\mu\alpha\beta} \st{(m)}{=} \r_a^{(m+1)}- 2V^c{}_a\r_c^{(m)} - 2(m-1)P^{bc}\r_b\Y^{(m)}_{ca} -\dfrac{1}{2}(m-1) n(\elltwo)\r_a^{(m)}  , 
		\end{align*}
and hence \eqref{auxiliar0} takes the form
		\begin{align}
			\wh e^{\alpha}_a \nabla_{\mu} \big(\xi^{\beta}\Sigma^{(m-1)}{}^{\mu}{}_{\alpha\beta}\big) &\st{(m)}{=}  \r_a^{(m+1)}- 3V^c{}_a\r_c^{(m)} - 2P^{bc} \big(m\r_b+\s_b \big)\Y^{(m)}_{ca}   \nonumber\\
			&\qquad\, + \left(\tr_P\bY-\dfrac{m+1}{2}n(\elltwo)\right)\r^{(m)}_a+ P^{bc}\nablacero_b \Y^{(m)}_{ca} + P^{bc}\Y_{ab}\r^{(m)}_c .\label{term1}
		\end{align}	
For the second term of \eqref{Rmordenmas} we insert \eqref{nablaxiup} and use that $\wh e_a^{\alpha}\xi^{\mu}\xi^{\beta}\Sigma_{\mu\alpha\beta}^{(m-1)}\st{(m)}{=}0$ (by \eqref{apC6}) and $\wh e_c^{\mu}\wh e_a^{\alpha}\wh e_b^{\beta}\Sigma_{\mu\alpha\beta}^{(m-1)}\st{(m)}{=}0$ (by \eqref{Sigmamcab}), obtaining
\begin{align}
- g^{\mu\rho}\wh e^{\alpha}_a\Sigma^{(m-1)}_{\mu\alpha\beta} \nabla_{\rho}\xi^{\beta}&\st{(m)}{=} -P^{cd}(\r-\s)_c \big({}^{(3)}\Sigma^{(m-1)}\big)_{da} - V^b{}_c n^c \big({}^{(1)}\Sigma^{(m-1)}\big)_{ab} \nonumber\\
&\stb[\eqref{apC5}]{(m)}{=} 2P^{bc}\s_b \Y^{(m)}_{ca} + \dfrac{1}{2}n(\elltwo) \r_a^{(m)},\label{term2}
\end{align}
where in the second equality we also inserted \eqref{Vn}. Finally, the last term in \eqref{Rmordenmas} is 
		\begin{align}
			- \wh e^{\alpha}_a\xi^{\beta}\nabla_{\beta}\big(\Sigma^{(m-1)}{}^{\mu}{}_{\mu\alpha}\big) & =-  \wh e^{\alpha}_a \lie_{\xi}\Sigma^{(m-1)\mu}{}_{\mu\alpha} + \wh e_a^{\alpha}\Sigma^{(m-1)}{}^{\mu}{}_{\mu\beta}\nabla_{\alpha}\xi^{\beta}\nonumber\\
			&= -\wh e_a^{\alpha}\Sigma^{(m)}{}^{\mu}{}_{\mu\alpha} + \wh e_a^{\alpha}\Sigma^{(m-1)}{}^{\mu}{}_{\mu\beta}\nabla_{\alpha}\xi^{\beta}\nonumber\\
			&\st{(m)}{=} -\nablacero_a\big(\tr_P\bY^{(m)}\big)+ \tr_P\bY^{(m)}(\r_a-\s_a),\label{term3}
		\end{align}
where in the last step we used \eqref{nablaxi} for $\wh e_a^{\alpha}\nabla_{\alpha}\xi^{\beta}$ and recalled \eqref{apC9} and \eqref{apC5bis}. Equation \eqref{liericxiX} is now obtained by simply inserting \eqref{term1}-\eqref{term3} into \eqref{Rmordenmas}. To prove \eqref{liericxin} it suffices to contract \eqref{liericxiX} with $n^a$ and use \eqref{derivadannull} and \eqref{Vn}.
	\end{proof}
\end{prop}

We conclude the appendix with the explicit expressions for $\Sigma_{cab}$, $({}^{(3)}\Sigma)_{ca}$ and $({}^{(2,3)}\Sigma)_c$. Although they are a particular case of more general identities derived in \cite{Mio5}, we re-derive them for completeness and to avoid the need of introducing additional notation to connect with the results in \cite{Mio3}. We emphasize that this result is not contained in Proposition \ref{propB3} because, as usual, the lowest values of $m$ require a different treatment.

\begin{prop}
	\label{propSigma}
	Let $\{\H,\bg,\bm\ell,\elltwo\}$ be null metric hypersurface data $(\Phi,\xi)$-embedded in $(\mc M,g)$ and extend $\xi$ off $\Phi(\H)$ by $\nabla_{\xi}\xi=0$. Then,
	\begin{align}
	\Sigma_{cab}&= 2\nablacero_{(a}\Y_{b)c} - \nablacero_c\Y_{ab}+2\r_c\Y_{ab}  +\dfrac{1}{2}\U_{ab}\nablacero_c\elltwo  ,\label{Sigmaabc}\\	
	({}^{(3)}\Sigma)_{ca}&=\Y^{(2)}_{ac}-\dfrac{1}{2}(\r-\s)_a\nablacero_c\elltwo -2V^d{}_a\Y_{cd},\label{3Sigma1ca}\\
	({}^{(2,3)}\Sigma)_c&=0.\label{23Sigma1c}
\end{align}
As a consequence,
\begin{align}
\Sigma_{cab}n^c &= -\lie_n\Y_{ab}-2\kappa_n \Y_{ab} + 2\nablacero_{(a}\r_{b)}+\dfrac{1}{2}n(\elltwo)\U_{ab},\label{sigmaprop1}\\
({}^{(3)}\Sigma)_{ca}n^c &= \r^{(2)}_a  -\dfrac{1}{2}n(\elltwo)\big(\r-\s\big)_a-2V^c{}_a \r_c,\label{sigmaprop2}\\
({}^{(2,3)}\Sigma)_c n^c & = 0.\label{sigmaprop3}
\end{align}
	\begin{proof}
From \eqref{apcextra2} for $m=1$ (recall that the equality is exact)
		\begin{equation*}
	\uwh\Sigma_{\mu\alpha\beta}=		\Sigma_{\mu\alpha\beta} = \nabla_{(\alpha}\mc{K}_{\beta)\mu} -\dfrac{1}{2} \nabla_{\mu}\mc{K}_{\alpha\beta}.
		\end{equation*} 
Pulling this back onto $\mc H$ and using \eqref{identity10} with $T=\mc K$,
\begin{equation*}
\Sigma_{cab} = \nablacero_{(a}\mc{K}_{b)c} + \Y_{ab}\mc{K}_{cd} n^d  +  \U_{ab}({}^{(1)}\mc{K})_{c} -  \dfrac{1}{2}\nablacero_c \mc{K}_{ab}  .
\end{equation*}	
Identity \eqref{Sigmaabc} follows after inserting $\mc{K}_{ab}=2\Y_{ab}$ and $({}^{(1)}\mc{K})_c = \dfrac{1}{2}\nablacero_c \elltwo$ (see \eqref{Kxi}). The contraction with $n^c$ given in \eqref{sigmaprop1} is obtained from this after using $n^c\nablacero_c\Y_{ab} =  \lie_n\Y_{ab}- 2\Y_{c(a}\nablacero_{b)}n^c$ and \eqref{derivadannull}. To prove the second and third identities we repeat the same strategy as in Prop. \ref{propB3}. First we compute the terms of the form ${}^{(i)}\nabla\mc K$ and ${}^{(i,j)}\nabla\mc K$ that appear in \eqref{3Sigma}-\eqref{23Sigma} for $m=1$ (which recall are exact in this case). To do this we use \eqref{identity20}-\eqref{a2id1}, \eqref{a2id3} and recall $\mc{K}^{(m)}_{ab}=2\Y^{(m)}_{ab}$ as well as \eqref{Kxi},
\begin{gather*}
\big({}^{(1)}\nabla\mc{K}\big)_{ac}=2\Y^{(2)}_{ac}-(\r-\s)_{(a}\nablacero_{c)}\elltwo -4V^d{}_{(a}\Y_{c)d},\hspace{0.3cm} \big({}^{(2)}\nabla\mc{K}\big)_{[ac]}=-\dfrac{1}{2}(\r-\s)_{[a}\nablacero_{c]}\elltwo - 2V^b{}_{[a}\Y_{c]b},\\
\big({}^{(1,2)}\nabla\mc{K}\big)_{c} = -\dfrac{1}{2}V^b{}_c \nablacero_b\elltwo,\qquad \big({}^{(2,3)}\nabla\mc{K}\big)_{c}=-V^b{}_c \nablacero_b\elltwo.
\end{gather*}
Expressions \eqref{3Sigma1ca} and \eqref{23Sigma1c} are obtained by replacing this into \eqref{3Sigma}-\eqref{23Sigma} respectively, and noting that $\nabla\mc{K}^{(0)}=\nabla g=0$. Finally, \eqref{sigmaprop2} and \eqref{sigmaprop3} are the contraction of \eqref{3Sigma1ca}-\eqref{23Sigma1c} with $n^c$.
	\end{proof}
\end{prop}

\section{Quasi-Einstein equations at null infinity}
\label{appendixC}

In this appendix we write down explicitly the tensors \eqref{Cm}-\eqref{Lm} for every $m\ge 1$ making the leading order terms explicit. The definition \eqref{defi_Sch} of the Schouten tensor in terms of the Ricci yields, after recalling \eqref{derivada},
\begin{equation}
	\label{lieSch}
	L^{(m)}_{\alpha\beta} = \dfrac{1}{\mf n-1}\left(R_{\alpha\beta}^{(m)}-\dfrac{1}{2\mf n}\sum_{i=0}^{m-1}\binom{m-1}{i} R^{(i+1)} \lie_{\xi}^{(m-1-i)}g_{\alpha\beta}\right).
\end{equation}
From Remark \ref{remarkricci} one has 
\begin{align}
	L_{\alpha\beta}^{(m)} &\st{[m-1]}{=} \dfrac{1}{\mf n-1} \left(R_{\alpha\beta}^{(m)}-\dfrac{1}{2\mf n}\big(R^{(m)}g_{\alpha\beta}+(m-1)R^{(m-1)}\mc{K}_{\alpha\beta} + \dfrac{(m-1)(m-2)}{2}R^{(m-2)} \mc{K}^{(2)}_{\alpha\beta}\big)\right),\label{Lm+1ordm}\\
	L_{\alpha\beta}^{(m)} &\st{[m]}{=} \dfrac{1}{\mf n-1} \left(R_{\alpha\beta}^{(m)}-\dfrac{1}{2\mf n}\big(R^{(m)}g_{\alpha\beta}+(m-1)R^{(m-1)}\mc{K}_{\alpha\beta}\big)\right),\label{Lm+1}
\end{align}
and

\begin{minipage}{0.6\textwidth}
\noindent
\begin{equation}
\label{Lmapp}
L_{\alpha\beta}^{(m)}  \st{[m+1]}{=}\dfrac{1}{\mf n-1} \left(R_{\alpha\beta}^{(m)}-\dfrac{1}{2\mf n}R^{(m)}g_{\alpha\beta}\right),
\end{equation}
\end{minipage}
\begin{minipage}{0.4\textwidth}
	\noindent
	\begin{equation}
		\label{Lm-1m+1}
L_{\alpha\beta}^{(m-1)} \st{[m+1]}{=} 0.
	\end{equation}
\end{minipage}
In the following proposition we compute the contractions of $L^{(m)}_{\alpha\beta}$ that will be needed below. Notice that some contractions are computed up to order $m$, and some others just to order $m+1$, depending on our needs for the rest of the appendix.
\begin{prop}
	\label{prop_schouten}
Let $\{\H,\bg,\bm\ell,\elltwo\}$ be null metric hypersurface data $(\Phi,\xi)$-embedded in $(\mc M,g)$ and extend $\xi$ off $\Phi(\H)$ by $\nabla_{\xi}\xi=0$. Let $L_{\alpha\beta}$ be the Schouten tensor, $L_{\alpha\beta}^{(m)}=\lie_{\xi}^{(m-1)}L_{\alpha\beta}$ and $m\ge 2$. Then,
\begin{align}
L_{ab}^{(m)}&\st{(m)}{=} \dfrac{1}{\mf{n}-1}\Big(-2\lie_n\Y^{(m)}_{ab}-\big(2m\kappa_n+\tr_P\bU\big)\Y^{(m)}_{ab}-\big(\tr_P\bY^{(m)}\big)\U_{ab} + 4P^{cd}\U_{c(a}\Y^{(m)}_{b)d}\nonumber\\
&\qquad\qquad\,  +4(\s-\r)_{(a}\r^{(m)}_{b)} + 2\nablacero_{(a}\r^{(m)}_{b)} +\dfrac{2(m-\mf n-1)}{\mf n}\kappa^{(m)}\Y_{ab}\Big)\nonumber\\
&\quad\, +\dfrac{1}{\mf n(\mf n-1)}\Big(\kappa^{(m+1)}+2\lie_n\big(\tr_P\bY^{(m)}\big)+\big(2m\kappa_n+\tr_P\bU\big)\tr_P\bY^{(m)}+P^{ef}P^{cd}\U_{ec}\Y^{(m)}_{fd}\nonumber\\
&\quad\, -2\div_P\br^{(m)} + 2P\big((2m+1)\br+3\bs,\br^{(m)}\big)+\big(2\tr_P\bY-(m+2) n(\elltwo)\big)\kappa^{(m)}\Big)\gamma_{ab},\label{Lmm}\\
\dot{L}_a^{(m)}&\st{(m)}{=}\dfrac{1}{\mf n-1}\Big(\r^{(m+1)}_a + P^{bc} \nablacero_b\Y^{(m)}_{ac}  -2m P^{bc}\r_b \Y^{(m)}_{ac} - \nablacero_a\big(\tr_P\bY^{(m)}\big) + \big(\tr_P\bY^{(m)}\big)(\r_a-\s_a)\nonumber\\
&\quad\, + (P^{bc}\Y_{ab}-3V^c{}_a)\r_c^{(m)} + \left(\tr_P\bY-\dfrac{m}{2}n(\elltwo)\right)\r^{(m)}_a +\dfrac{m-1}{2\mf n}\kappa^{(m)}\nablacero_a\elltwo \Big)\nonumber\\
&\quad\, +\dfrac{1}{\mf n(\mf n-1)}\Big(\kappa^{(m+1)}+2\lie_n\big(\tr_P\bY^{(m)}\big)+\big(2m\kappa_n+\tr_P\bU\big)\tr_P\bY^{(m)}+P^{ef}P^{cd}\U_{ec}\Y^{(m)}_{fd}\nonumber\\
&\quad\, -2\div_P\br^{(m)} + 2P\big((2m+1)\br+3\bs,\br^{(m)}\big)+\big(2\tr_P\bY-(m+2) n(\elltwo)\big)\kappa^{(m)}\Big)\ell_{a}, \label{dLmm}\\
\ddot{L}^{(m)}&\st{(m+1)}{=}\dfrac{1}{\mf n-1}\left(-\tr_P\bY^{(m+1)} + \dfrac{\kappa^{(m+1)}}{\mf n}\elltwo\right).\label{ddLmm+1}
\end{align}
In particular,
\begin{align}
L^{(m)}_{ab}n^b &\st{(m)}{=} \dfrac{1}{\mf n-1}\left(-\lie_n\r_a^{(m)}-(2(m-1)\kappa_n+\tr_P\bU)\r_a^{(m)} - \nablacero_a\kappa^{(m)}+\dfrac{2(m-1)}{\mf n}\kappa^{(m)}\r_a\right),\label{Lmmn}\\
\dot{L}_a^{(m)}n^a &\st{(m)}{=} \dfrac{1}{\mf n(\mf n-1)}\Bigg((1-\mf n)\kappa^{(m+1)} + (2-\mf n)\lie_n\big(\tr_P\bY^{(m)}\big)+\big((2m-\mf n)\kappa_n+\tr_P\bU\big)\tr_P\bY^{(m)}\nonumber\\
&\quad\, +(1-\mf n)P^{bc}P^{ad} \U_{bd}\Y^{(m)}_{ac}  + 2P\big(((2-\mf n)m+1-\mf n)\br + (3-2\mf n)\bs,\br^{(m)}\big)\nonumber\\
&\quad\,+ (\mf n-2) \div_P\br^{(m)} +\left((2-\mf n)\tr_P\bY+\dfrac{3\mf n+m(\mf n-1)-5}{2}n(\elltwo)\right)\kappa^{(m)}\Bigg),\label{dLmn}
\end{align}
and

\begin{minipage}{0.4\textwidth}
\noindent
\begin{equation}
\label{Lmm+1}
L_{ab}^{(m)}\st{(m+1)}{=} \dfrac{\kappa^{(m+1)}}{\mf n(\mf n-1)}\gamma_{ab},
\end{equation}
\end{minipage}
\begin{minipage}{0.6\textwidth}
	\noindent
	\begin{align}
\label{dLmm+1}
\dot{L}_a^{(m)}&\st{(m+1)}{=}\dfrac{1}{\mf n-1}\left(\r_a^{(m+1)}+\dfrac{\kappa^{(m+1)}}{\mf n}\ell_a\right),
	\end{align}
\end{minipage}

\begin{minipage}{0.4\textwidth}
	\noindent
	\begin{equation}
		\label{dLmnm+1}
\dot{L}_a^{(m)}n^a \st{(m+1)}{=} -\dfrac{\kappa^{(m+1)}}{\mf n},
	\end{equation}
\end{minipage}
\begin{minipage}{0.6\textwidth}
	\noindent
\begin{equation}
\label{PLmm+1}
P^{bc}L_{ab}^{(m)} \st{(m+1)}{=} \dfrac{\kappa^{(m+1)}}{\mf n(\mf n-1)}(\delta^c_a-n^c\ell_a).
\end{equation}
\end{minipage}

\begin{proof}
The contraction of \eqref{Lm+1} with $e_a^{\alpha}e_b^{\beta}$ gives \eqref{Lmm} after replacing \eqref{R}, \eqref{scalar} and \eqref{scalarm+1} and recalling that $\mc{K}_{ab}=2\Y_{ab}$. Similarly, \eqref{dLmm} is obtained by contracting \eqref{Lm+1} with $e_a^{\alpha}\xi^{\beta}$ and using \eqref{dotR}, \eqref{scalar}-\eqref{scalarm+1} and $({}^{(1)}\mc K)_a=\frac{1}{2}\nablacero_a\elltwo$ (cf. \eqref{Kxi}). To prove \eqref{ddLmm+1} we simply contract \eqref{Lmapp} with $\xi$ twice and insert \eqref{ddotR} and \eqref{scalar}. Relations \eqref{Lmmn} and \eqref{dLmn} are respectively the contractions of \eqref{Lmm} and \eqref{dLmm} with $n$. They are obtained after using $\gamma_{ab}n^b = \U_{ab}n^b=0$, $\s_b n^b = 0$ and $\ell_an^a=1$ together with the definitions of $\kappa_n$, $\r_a$, $\kappa^{(m)}$ and $\r^{(m)}_a$. In addition, \eqref{Lmmn} uses identity \eqref{nnablaomega} applied to $\omega_a=\r_a$, and \eqref{dLmn} employs \eqref{derivadannull} and \eqref{Vn}. Finally, relations \eqref{Lmm+1}-\eqref{PLmm+1} are immediate from the previous ones by simply keeping the quantities of order $m+1$.
\end{proof}
\end{prop}

The tensors $L^{(1)}_{ab}$, $\dot{L}^{(1)}_a$ and $\ddot{L}^{(1)}$ require a separate analysis. The starting point is \eqref{lieSch} with $m=1$ (this is an exact relation), which we contract with $e_a^{\alpha}e_b^{\beta}$, $e_a^{\alpha}\xi^{\beta}$ and $\xi^{\alpha}\xi^{\beta}$. Taking into account \eqref{constraint}, \eqref{ricxixi} and \eqref{ricxiX} (we do not replace $R$ from \eqref{scal} since this will not be needed) and recalling $g_{\alpha\beta}e_a^{\alpha}e_b^{\beta}=\gamma_{ab}$, $g_{\alpha\beta}e_a^{\alpha}\xi^{\beta}=\ell_{a}$ and $g_{\alpha\beta}\xi^{\alpha}\xi^{\beta}=\elltwo$, the result is
\begin{align}
	L_{ab}^{(1)} & = \dfrac{1}{\mf n-1}\Bigg(\Rcero_{(ab)} -2\lie_n \Y_{ab} - (2\kappa_n+\tr_P\bU)\Y_{ab} + \nablacero_{(a}\left(\s_{b)}+2\r_{b)}\right) -2\r_a\r_b + 4\r_{(a}\s_{b)} - \s_a\s_b \nonumber\\
	& \qquad\qquad - (\tr_P\bY)\U_{ab}+ 2P^{cd}\U_{d(a}\left(2\Y_{b)c}+\F_{b)c}\right) - \dfrac{R}{2\mf n}\gamma_{ab}\Bigg).\label{L1}\\
	\dot L_{a}^{(1)} & = \dfrac{1}{\mf n-1}\Bigg( \r^{(2)}_a - P^{bc}A_{bca}  - P^{bc}(\r+\s)_b(\Y+\F)_{ac} +\dfrac{1}{2}\kappa_n \nablacero_a\elltwo -\dfrac{1}{2} n(\elltwo)(\r-\s)_{a} - \dfrac{R}{2\mf n}\ell_{a}\Bigg),\label{dL1}\\
	\ddot L^{(1)} & = \dfrac{1}{\mf n-1}\Bigg( -P^{ab}\Y^{(2)}_{ab} + P^{ab}P^{cd}(\Y+\F)_{ac}(\Y+\F)_{bd} + P^{ab}(\r-\s)_a \nablacero_b\elltwo - \dfrac{R}{2\mf n}\elltwo\Bigg).\label{ddL1}
\end{align}

Using these identities we can now compute $\mc{L}^{(m+1)}_a$ and $\dot{\mc{L}}^{(m+1)}$ up to order $m+1$. Taking $m$ transverse derivatives of $\L_{\alpha}\d (\mf n-1)L_{\alpha\beta}\nabla^{\beta}\Omega$ and using \eqref{derivada},
\begin{align}
	\mc{L}_{\alpha}^{(m+1)} &= (\mf n-1)\sum_{i=0}^m\binom{m}{i} L^{(i+1)}_{\alpha\beta} \lie_{\xi}^{(m-i)}g^{\beta\mu}\nabla_{\mu}\Omega \nonumber\\
	&= (\mf n-1)\sum_{i=0}^m\sum_{j=0}^{m-i}\binom{m}{i}\binom{m-i}{j} L^{(i+1)}_{\alpha\beta} ( \lie_{\xi}^{(j)}g^{\beta\mu} )\nabla_{\mu}\lie_{\xi}^{(m-i-j)}\Omega,\label{mcLalpham}
\end{align}
so (recall $L_{\alpha\beta}^{(k)}\st{[k+2]}{=}0$, by \eqref{Lm-1m+1})
\begin{equation*}
	\mc{L}_{\alpha}^{(m+1)} \st{[m+1]}{=} (\mf n-1)\left(L^{(m+1)}_{\alpha\beta} g^{\beta\mu}\nabla_{\mu}\Omega + mL^{(m)}_{\alpha\beta} \big(g^{\beta\mu}\nabla_{\mu}\lie_{\xi}\Omega + (\lie_{\xi}g^{\beta\mu})\nabla_{\mu}\Omega\big)\right).
\end{equation*}
Now we note that $g^{\beta\mu}\nabla_{\mu}\Omega \st{\scri}{=} \sigma^{(1)}\nu^{\beta}$ (by \eqref{inversemetric}) and
\begin{align*}
g^{\beta\mu}\nabla_{\mu}\lie_{\xi}\Omega + (\lie_{\xi}g^{\beta\mu})\nabla_{\mu}\Omega & \st{\scri}{=} g^{\beta\mu}\nabla_{\mu}\lie_{\xi}\Omega - \mc{K}^{\beta\mu}\nabla_{\mu}\Omega\\
&\stb[\eqref{inversemetric}]{\scri}{=} \big(P^{bc}e_b^{\beta}e_c^{\mu}+\xi^{\beta}\nu^{\mu}+\xi^{\mu}\nu^{\beta}\big)\nabla_{\mu}\lie_{\xi}\Omega -\sigma^{(1)} \nu^{\mu} \mc{K}_{\mu}{}^{\beta}\\
&\stb[\eqref{Kupnu}]{\scri}{=} \left(P^{bc}\big(\nablacero_c\sigma^{(1)}-2\sigma^{(1)}\r_c\big) +n^b \left(\sigma^{(2)}-\dfrac{1}{2}\sigma^{(1)} n(\elltwo)\right)\right)e_b^{\beta}\\
&\qquad +\big(\lie_n\sigma^{(1)}+2\sigma^{(1)}\kappa_n\big)\xi^{\beta}.
\end{align*}
Hence,
\begin{align}
	\mc{L}_{\alpha}^{(m+1)} &\stb[\scri]{[m+1]}{=} (\mf n-1)\Bigg(\sigma^{(1)}L^{(m+1)}_{\alpha\beta} n^b e_b^{\beta} +m \big(\lie_n\sigma^{(1)}+2\sigma^{(1)}\kappa_n\big)\xi^{\beta}L^{(m)}_{\alpha\beta}\nonumber\\
	& \qquad\qquad +m\left( P^{bc}\big(\nablacero_c\sigma^{(1)}-2\sigma^{(1)}\r_c\big)+n^b\left(\sigma^{(2)}-\dfrac{1}{2}\sigma^{(1)} n(\elltwo)\right)\right)e_b^{\beta}L^{(m)}_{\alpha\beta}\Bigg).\label{mcLalpham[]}
\end{align}

\begin{prop}
	\label{propL}
	Let $\mc{L}^{(m+1)}_a$ and $\dot{\mc{L}}^{(m+1)}$ be defined as in \eqref{Lm}. Then, 
\begin{align}
\dfrac{	2\mf n}{\sigma^{(1)}} \dot{\mc{L}}^{(1)} &=-2(\mf n-1)\kappa^{(2)} -2(\mf n-2)\lie_n\big(\tr_P\bY\big) + 2\big(\tr_P\bU-(\mf n-2) \kappa_n \big)\tr_P\bY  \nonumber\\
&\quad\, -2(\mf n-1)P^{ab}P^{cd}\U_{ac}\Y_{bd} + 2(\mf n-2)\div\br +(2\mf n-3)\div\bs  \nonumber\\
&\quad\, +\big(2(2\mf n-3)\kappa_n+(\mf n-1)\tr_P\bU\big)n(\elltwo) - 2(2\mf n-3)P(\br,\br) -(4\mf n-5)P(\bs,\bs) \nonumber\\
&\quad\, - 4(2\mf n-3)P(\br,\bs)- \tr_P\Rcero,\label{dotLgeneral}\\
\dfrac{1}{\sigma^{(1)}}\mc{L}^{(1)}_a = \mc{R}_{ab}n^b &= -\lie_n(\r_b-\s_b) - \nablacero_b \kappa_n - (\tr_P\bU) (\r_b-\s_b) - \nablacero_b (\tr_P\bU)  \nonumber\\
&\quad\, + P^{cd}\nablacero_c\U_{bd} - 2P^{cd}\U_{bd}s_c,\label{La1general}
\end{align}	
and for every $m\ge 1$,
	\begin{align}
\mc{L}^{(m+1)}_a  &\st{(m+1)}{=} - \sigma^{(1)} \lie_n \r_a^{(m+1)} + \left(m \lie_n\sigma^{(1)} - \tr_P\bU\right)\r_a^{(m+1)} - \sigma^{(1)}\nablacero_a \kappa^{(m+1)}  + \dfrac{m\kappa^{(m+1)}}{\mf n}\nablacero_a\sigma^{(1)} ,	\label{Lam}\\
\mf n \dot{\mc{L}}^{(m+1)} &\st{(m+1)}{=} -(\mf n-1)\sigma^{(1)}\kappa^{(m+2)} -(\mf n-2)\sigma^{(1)} \lie_n\big(\tr_P\bY^{(m+1)}\big) -(\mf n-1)\sigma^{(1)} P^{ab}P^{cd}\U_{ac}\Y^{(m+1)}_{bd} \nonumber\\
&\, +\left(\big(2m(1-\mf n)+2-\mf n\big)\sigma^{(1)}\kappa_n -m\mf{n} \lie_n\sigma^{(1)} +\sigma^{(1)}\tr_P\bU\right)\tr_P\bY^{(m+1)} +\dot{\mathscr R}_{\mc L}^{(m)},	\label{dotLm}
	\end{align}
where $\dot{\mathscr R}_{\mc L}^{(m)}$ is an explicit tensor that depends linearly on $\br^{(m+1)}$ as well as on lower order transverse derivatives and that we do not write for simplicity (it can be easily read out by making explicit all the calculations in the proof).
	\begin{proof}
The definition $\L_{\alpha}\d \left(R_{\alpha\beta}-\frac{1}{2\mf n} R g_{\alpha\beta}\right)\nabla^{\beta}\Omega$ and $\nabla_{\mu}\Omega\st{\scri}{=} \sigma^{(1)}\nu_{\mu}$ gives $$\mc{L}^{(1)}_{\alpha} =\sigma^{(1)}\left(R_{\alpha\beta}-\dfrac{1}{2\mf n}R g_{\alpha\beta}\right)n^b e_b^{\beta}.$$

 The contraction with $\xi^{\alpha}$ yields \eqref{dotLgeneral} after replacing \eqref{dotRn} and \eqref{scal}, and that with $e_a^{\alpha}$ gives, after inserting \eqref{constraintn}, \eqref{La1general}. To obtain \eqref{Lam} we contract \eqref{mcLalpham[]} with $e^{\alpha}_a$ and use $L_{ab}^{(m)}n^b \st{(m+1)}{=}0$ (by \eqref{Lmm+1}),
\begin{align*}
\dfrac{1}{\mf n-1}\mc{L}^{(m+1)}_a &\st{(m+1)}{=} \sigma^{(1)}L_{ab}^{(m+1)}n^b + m\big(\lie_n\sigma^{(1)}+2\sigma^{(1)}\kappa_n\big) \dot{L}^{(m)}_a +m P^{bc} \big(\nablacero_c\sigma^{(1)}-2\sigma^{(1)}\r_c\big) L^{(m)}_{ab}.
\end{align*}	
This becomes \eqref{Lam} after inserting \eqref{Lmmn}, \eqref{dLmm+1} and \eqref{PLmm+1}. In order to prove \eqref{dotLm} we contract \eqref{mcLalpham[]} with $\xi^{\alpha}$, 
		\begin{equation*}
			\begin{aligned}
\dfrac{1}{\mf n-1} \dot{\mc{L}}^{(m+1)}&\st{(m+1)}{=} \sigma^{(1)}\dot{L}^{(m+1)}_b n^b +m \big(\lie_n\sigma^{(1)}+2\sigma^{(1)}\kappa_n\big)\ddot{L}^{(m)} +m P^{bc}\big(\nablacero_c\sigma^{(1)}-2\sigma^{(1)}\r_c\big)\dot{L}_b^{(m)} \\
&\qquad\, + \left(\sigma^{(2)}-\dfrac{1}{2}\sigma^{(1)} n(\elltwo)\right)\dot{L}^{(m)}_b n^b.
			\end{aligned}
		\end{equation*}
Inserting \eqref{dLmn}, \eqref{ddLmm+1}, \eqref{dLmm+1} and \eqref{dLmnm+1} and using $P^{bc}\ell_b =-\elltwo n^c$,	\eqref{dotLm} is established.
	\end{proof}
\end{prop}


Next we compute the tensors $\mc{Q}_{ab}^{(m)}$, $\dot{\mc Q}_a^{(m)}$ and $\ddot{\mc Q}^{(m)}$ to the leading order. To do that, we apply $\lie_{\xi}^{(m)}$ to \eqref{defC} and use Proposition \ref{derivadashess},
\begin{equation}
	\label{Calphabeta}
	\begin{aligned}
		\dfrac{1}{\mf n-1}\Q^{(m+1)}_{\alpha\beta} &= \nabla_{\alpha}\nabla_{\beta}\lie_{\xi}^{(m)}\Omega - \sum_{k=0}^{m-1}\binom{m}{k+1} \nabla_{\sigma}\big(\lie_{\xi}^{(k)}\Omega\big) \Sigma^{(m-k)}{}^{\sigma}{}_{\beta\alpha} + \sum_{k=0}^{m}\binom{m}{k} \big(\lie_{\xi}^{(k)}\Omega\big) L_{\alpha\beta}^{(m-k+1)}.
	\end{aligned}
\end{equation}
To perform the calculation it is convenient to introduce the following symmetric tensors (recall that $\Sigma^{\sigma}{}_{\beta\alpha}$ is symmetric in $\alpha,\beta$, so the same holds for $\Sigma^{(k)\sigma}{}_{\beta\alpha}$) $$\text{I}^{(m+1)}_{\alpha\beta} \d  \nabla_{\alpha}\nabla_{\beta}\lie_{\xi}^{(m)}\Omega,\qquad \text{II}^{(m+1)}_{\alpha\beta}\d -\sum_{k=0}^{m-1}\binom{m}{k+1} \nabla_{\sigma}\big(\lie_{\xi}^{(k)}\Omega\big) \Sigma^{(m-k)}{}^{\sigma}{}_{\alpha\beta},$$ $$\text{III}^{(m+1)}_{\alpha\beta}\d \sum_{k=0}^{m}\binom{m}{k} \big(\lie_{\xi}^{(k)}\Omega\big) L_{\alpha\beta}^{(m-k+1)}.$$ 

Given that $\Sigma^{(k)}$ involves at most $m+1$ derivatives of $g$ and that $\Omega=0$, $\nabla_{\mu}\Omega\st{\scri}{=}\sigma^{(1)}\nu_{\mu}$ at $\scri$, we can write
\begin{align}
\text{II}^{(m+1)}_{\alpha\beta} &\st{[m]}{=} -\big(\nabla^{\sigma}\Omega\big) \Sigma^{(m)}_{\sigma\alpha\beta} - m \big(\nabla^{\sigma}\lie_{\xi}\Omega \big) \Sigma^{(m-1)}_{\sigma\alpha\beta}\stb[\scri]{[m]}{=} - \sigma^{(1)}\nu^{\sigma}\Sigma^{(m)}_{\sigma\alpha\beta} - m \big(\nabla^{\sigma}\lie_{\xi}\Omega \big) \Sigma^{(m-1)}_{\sigma\alpha\beta},\label{II0}\\
\text{II}^{(m+1)}_{\alpha\beta} &\stb[\scri]{[m+1]}{=} - \sigma^{(1)}\nu^{\sigma} \Sigma^{(m)}_{\sigma\alpha\beta}.\label{II}
\end{align}
Similarly, $L^{(m)}$ involves at most $m+1$ derivatives of $g$, so
\begin{align}
\text{III}^{(m+1)}_{\alpha\beta} &\st{[m]}{=} \Omega L_{\alpha\beta}^{(m+1)}+m(\lie_{\xi}\Omega) L_{\alpha\beta}^{(m)} + \dfrac{m(m-1)\lie_{\xi}^{(2)}\Omega}{2} L_{\alpha\beta}^{(m-1)},\label{III0}\\
\text{III}^{(m+1)}_{\alpha\beta} &\st{[m+1]}{=} \Omega L_{\alpha\beta}^{(m+1)} + m(\lie_{\xi}\Omega) L_{\alpha\beta}^{(m)}.\label{III}
\end{align}

Using this and the formulas in Propositions \ref{proppullback}, \ref{propB3} and \ref{prop_schouten}, the computation of $\mc{Q}_{ab}^{(m)}$, $\dot{\mc Q}_a^{(m)}$ and $\ddot{\mc Q}^{(m)}$ will be straightforward. Observe again that $\ddot{\Q}^{(m)}$ is computed to one order less, since this is all that will be needed.

\begin{prop}
	\label{propQ}
	Let $\mc{Q}_{ab}^{(m)}$, $\dot{\mc Q}_a^{(m)}$ and $\ddot{\mc Q}^{(m)}$ be as in \eqref{Cm} and $m\ge 2$. Then,
\begin{align}
\Q_{ab}^{(m+1)} &\st{(m)}{=} (\mf n-1)\sigma^{(m+1)}\U_{ab} +  (\mf n-1-2m)\sigma^{(1)}\lie_n\Y^{(m)}_{ab}+ 4m\sigma^{(1)}P^{cd}\U_{c(a}\Y_{b)d}^{(m)}\nonumber \\
&+ m\left((\mf n-1)\big(\lie_n\sigma^{(1)}+\sigma^{(1)}\kappa_n\big) + (\mf n-1-2m)\sigma^{(1)}\kappa_n -\sigma^{(1)}\tr_P\bU \right)\Y_{ab}^{(m)}  \nonumber\\
&+\dfrac{m}{\mf n}\Big(\sigma^{(1)}\Big(\kappa^{(m+1)}+2\lie_n\big(\tr_P\bY^{(m)}\big) + \big(2m\kappa_n+\tr_P\bU\big)\tr_P\bY^{(m)}+ P^{ab}P^{cd}\U_{ac}\Y^{(m)}_{bd} \Big)\gamma_{ab}\nonumber\\
&\qquad +\dfrac{m-1}{2}\sigma^{(2)}\kappa^{(m)}\gamma_{ab}\Big)-m\sigma^{(1)}\big(\tr_P\bY^{(m)}\big)\U_{ab} +\wt{\mathscr{R}}_{ab}^{(m)},\label{Cmab}\\
\dot\Q_{a}^{(m+1)} &\st{(m)}{=} (\mf n-1)\big(\nablacero_a \sigma^{(m+1)} - \sigma^{(m+1)}(\r-\s)_a -mP^{bc}\Y_{ab}^{(m)}\nablacero_c\sigma^{(1)}\big)- (\mf n-1-m)\sigma^{(1)}\r_a^{(m+1)} \nonumber\\
&  + m\sigma^{(1)}\left(P^{bc}\nablacero_b\Y^{(m)}_{ac}-\nablacero_a\big(\tr_P\bY^{(m)}\big)+2P^{bc}\r_b\Y_{ac}^{(m)} + \big(\tr_P\bY^{(m)}\big) (\r-\s)_a\right)\nonumber\\
& + \dfrac{m\sigma^{(1)}}{\mf n} \left(\kappa^{(m+1)}+2\lie_n\big(\tr_P\bY^{(m)}\big)+\big(2m\kappa_n+\tr_P\bU\big)\tr_P\bY^{(m)} + P^{bc}P^{df}\U_{bd}\Y^{(m)}_{cf}\right)\ell_a \nonumber\\
&+ \wt{\mathscr R}^{(m)}_a,\label{dotCam}\\
\ddot\Q^{(m+1)} &\st{(m+1)}{=} (\mf n-1)\sigma^{(m+2)} - m\sigma^{(1)}\tr_P\bY^{(m+1)} + \dfrac{m\sigma^{(1)}\elltwo}{\mf n}\kappa^{(m+1)},\label{ddotCm}
\end{align}
	where $\mc{\wt O}^{(m)}_{ab}$ and $\mc{\wt O}^{(m)}_a$ are tensors that depend on $\br^{(m)}$, $\sigma^{(m)}$ and lower order terms and we do not write for simplicity (they can be easily read out by performing all the calculations in the proof explicitly).
		%
		\begin{proof}
To prove each identity in \eqref{Cmab}-\eqref{ddotCm} we contract the three pieces $\text{I}^{(m+1)}_{\alpha\beta}$, $\text{II}^{(m+1)}_{\alpha\beta}$ and $\text{III}^{(m+1)}_{\alpha\beta}$ with $e_a^{\alpha}e_b^{\beta}$, $e_a^{\alpha}\xi^{\beta}$ and $\xi^{\alpha}\xi^{\beta}$. The contraction of $\text{I}^{(m+1)}_{\alpha\beta}$ with $e_a^{\alpha}e_b^{\beta}$ gives, after using identity \eqref{identity1}, $$\text{I}^{(m+1)}_{ab} =\nablacero_a\nablacero_b\sigma^{(m)} +\lie_n(\sigma^{(m)})\Y_{ab} + \sigma^{(m+1)}\U_{ab}.$$

 Contracting \eqref{II0} with $e_a^{\alpha}e_b^{\beta}$, inserting \eqref{inversemetric} and recalling $\Sigma^{(m-1)}_{cab} \st{(m)}{=} 0$ (by \eqref{Sigmamcab}) gives
			\begin{align*}
				e^{\alpha}_a e_b^{\beta}\text{II}^{(m+1)}_{\alpha\beta} & \st{(m)}{=}  -\sigma^{(1)}\Sigma^{(m)}_{cab}n^c - m (\lie_n\sigma^{(1)}) \big({}^{(1)}\Sigma^{(m-1)}\big)_{ab}.
			\end{align*}
Inserting \eqref{apC1}, \eqref{apC5} and taking into account \eqref{apCextra}, $$\text{II}^{(m+1)}_{ab} \st{(m)}{=} \sigma^{(1)}\lie_n\Y_{ab}^{(m)} + \big(2m\sigma^{(1)}\kappa_n +m \lie_n\sigma^{(1)}\big)\Y_{ab}^{(m)}-2\sigma^{(1)}\nablacero_{(a}\r_{b)}^{(m)}+2\sigma^{(1)}\kappa^{(m)}\Y_{ab}.$$ 

Finally, the contraction of \eqref{III0} with $e_a^{\alpha}e_b^{\beta}$ is (recall $\Omega\st{\scri}{=}0$) $$\text{III}^{(m+1)}_{ab} \st{(m)}{=} m\sigma^{(1)}L_{ab}^{(m)}+\dfrac{m(m-1)\sigma^{(2)}}{2}L_{ab}^{(m-1)},$$ 

where the first term is given by \eqref{Lmm} and the second by \eqref{Lmm+1}. Since $\Q^{(m)}_{ab}=(\mf n-1)\big(\text{I}_{ab}^{(m)}+\text{II}_{ab}^{(m)}+\text{III}_{ab}^{(m)}\big)$ we simply need to add the three terms and simplify to get \eqref{Cmab}.\\

 To prove \eqref{dotCam} we repeat the same steps. Firstly, we contract $\text{I}_{\alpha\beta}^{(m+1)}$ with $e_a^{\alpha}\xi^{\beta}$ and use \eqref{identity2} to get
			\begin{equation}
				\label{I}
e_a^{\alpha}\xi^{\beta}\text{I}^{(m+1)}_{\alpha\beta} = \nablacero_a\sigma^{(m+1)}-\sigma^{(m+1)}(\r-\s)_a-V^b{}_a \nablacero_b \sigma^{(m)}.
			\end{equation} 
Secondly, the contraction of \eqref{II0} with $e_a^{\alpha}\xi^{\beta}$ gives, after inserting \eqref{inversemetric} and recalling that $\big({}^{(1,3)}\Sigma^{(m-1)}\big)_a\st{(m)}{=} 0$ (by \eqref{apC6} and \eqref{apCextra}),
			\begin{align*}
e_a^{\alpha}\xi^{\beta}\text{II}^{(m+1)}_{\alpha\beta} &\st{(m)}{=} - \sigma^{(1)} \big({}^{(3)}\Sigma^{(m)}\big)_{ca}n^c - mP^{cd}\big({}^{(3)}\Sigma^{(m-1)}\big)_{ca} \nablacero_d\sigma^{(1)} - m\sigma^{(2)}\big({}^{(3)}\Sigma^{(m-1)}\big)_{ca} n^c\\
				&\st{(m)}{=} -\sigma^{(1)}\r^{(m+1)}_a+2m\sigma^{(1)} P^{bc}\r_b\Y_{ac}^{(m)} - m P^{cd}\Y^{(m)}_{ca}\nablacero_d\sigma^{(1)} \\
				&\qquad + \mbox{terms depending on } \br^{(m)} \mbox{ and hypersurface data},
			\end{align*}
where in the second line we used \eqref{3Sigmaca} and \eqref{apC10}. To compute $e^{\alpha}_a\xi^{\beta}\text{III}^{(m+1)}_{\alpha\beta} $ we contract \eqref{III0} with $e^{\alpha}_a\xi^{\beta}$, $$e^{\alpha}_a\xi^{\beta}\text{III}^{(m+1)}_{\alpha\beta} \st{(m)}{=}m\sigma^{(1)}\dot{L}^{(m)}_a + \dfrac{m(m-1)\sigma^{(2)}}{2}\dot{L}^{(m-1)}_a. $$ 

The explicit forms of $\dot{L}^{(m)}_a$ and $\dot{L}^{(m-1)}_a$ are given in \eqref{dLmm} and \eqref{dLmm+1}. Using as before $\dot\Q^{(m)}_a=(\mf n-1)e_a^{\alpha}\xi^{\beta}\big(\text{I}_{\alpha\beta}^{(m)}+\text{II}_{\alpha\beta}^{(m)}+\text{III}_{\alpha\beta}^{(m)}\big)$ we find \eqref{dotCam} after adding the three terms and simplifying. \\

 The proof of \eqref{ddotCm} is analogous. Contracting $\text{I}^{(m+1)}_{\alpha\beta}$ with $\xi^{\alpha}\xi^{\beta}$ and using \eqref{identity3} gives $$\xi^{\alpha}\xi^{\beta}\text{I}^{(m+1)}_{\alpha\beta} = \sigma^{(m+2)}.$$ 
 
 The corresponding term with $\text{II}^{(m+1)}_{\alpha\beta}$ vanishes because the contraction of \eqref{II} with $\xi^{\alpha}\xi^{\beta}$ is $$\xi^{\alpha}\xi^{\beta}\text{II}^{(m+1)}_{\alpha\beta} \st{(m+1)}{=} - \sigma^{(1)}\big({}^{(2,3)}\Sigma^{(m)}\big)_c n^c \st{(m+1)}{=} 0,$$ where in the last equality we used \eqref{apCextra} and \eqref{apC6}. Finally, the term with $\text{III}^{(m+1)}_{\alpha\beta}$ is obtained by contracting \eqref{III} with $\xi^{\alpha}\xi^{\beta}$ and inserting \eqref{ddLmm+1}, which gives $$\xi^{\alpha}\xi^{\beta}\text{III}^{(m+1)}_{\alpha\beta} \st{(m+1)}{=} m\sigma^{(1)} \ddot{L}^{(m)} \st{(m+1)}{=} \dfrac{m\sigma^{(1)}}{\mf n-1} \left(-\tr_P\bY^{(m+1)} + \dfrac{\kappa^{(m+1)}}{\mf n}\elltwo\right).$$ 
 
 The expression $\ddot\Q^{(m)}=(\mf n-1)\xi^{\alpha}\xi^{\beta}\big(\text{I}_{\alpha\beta}^{(m)}+\text{II}_{\alpha\beta}^{(m)}+\text{III}_{\alpha\beta}^{(m)}\big)$ yields \eqref{ddotCm} at once.
		\end{proof}
	\end{prop}

As in other cases, the lowest orders $\Q^{(1)}$ and $\Q^{(2)}$ require a specific treatment (note that they are excluded from Proposition \ref{propQ} by the condition $m\ge 2$). To obtain the former we simply contract $\Q_{\alpha\beta} \st{\scri}{=}(\mf n-1)\nabla_{\alpha}\nabla_{\beta}\Omega$ with $e_a^{\alpha}e_b^{\beta}$, $e_a^{\alpha}\xi^{\beta}$ and $\xi^{\alpha}\xi^{\beta}$ and use \eqref{identity1}-\eqref{identity3},
\begin{equation}
	\label{Q1general}
\Q_{ab}^{(1)} =\sigma^{(1)}\U_{ab},\qquad \dot{\Q}^{(1)}_a = \nablacero_a\sigma^{(1)} -\sigma^{(1)}(\r-\s)_a,\qquad 	\ddot{\Q}^{(1)} = \sigma^{(2)} .
\end{equation}
To compute the latter we evaluate the exact expression \eqref{Calphabeta} with $m=1$ and use $\nabla_{\sigma}\Omega \st{\scri}{=} \sigma^{(1)}\nu_{\sigma}$ to get $$\dfrac{1}{\mf n-1}\Q^{(2)}_{\alpha\beta} \st{\scri}{=} \nabla_{\alpha}\nabla_{\beta}\lie_{\xi}\Omega - \sigma^{(1)}\Sigma_{\sigma\alpha\beta} \nu^{\sigma} + \sigma^{(1)} L_{\alpha\beta}.$$ 

The contraction of this expression with $e_a^{\alpha}e_b^{\beta}$, $e_a^{\alpha}\xi^{\beta}$ and $\xi^{\alpha}\xi^{\beta}$ can be evaluated from \eqref{identity1}-\eqref{identity3}, \eqref{sigmaprop1}-\eqref{sigmaprop3} and \eqref{L1}-\eqref{ddL1}. The result is
\begin{align}
\hspace{-0.25cm}\Q^{(2)}_{ab} & = (\mf n-3)\sigma^{(1)}\lie_n\Y_{ab} + \big((\mf n-1)\big(\lie_n\sigma^{(1)}+\sigma^{(1)}\kappa_n\big) + \sigma^{(1)}\big((\mf n-3)\kappa_n-\tr_P\bU\big)\big)\Y_{ab}- (\tr_P\bY)\U_{ab}\nonumber\\
&+ 2\sigma^{(1)}P^{cd}\U_{d(a}\left(2\Y_{b)c}+\F_{b)c}\right)-2\sigma^{(1)}(\mf n-2)\nablacero_{(a}\r_{b)} +2\sigma^{(1)}\big(2\r_{(a}\s_{b)}-\r_a\r_b\big)- \dfrac{\sigma^{(1)} R}{2\mf n}\gamma_{ab}\nonumber\\
&+\sigma^{(1)}\big(\nablacero_{(a}\s_{b)}- \s_a\s_b+\Rcero_{(ab)}\big)+(\mf n-1)\left(\nablacero_a\nablacero_b\sigma^{(1)}+ \left(\sigma^{(2)}-\dfrac{1}{2}\sigma^{(1)}n(\elltwo)\right)\U_{ab}\right),\label{Q2general}\\
\hspace{-0.25cm}\dot\Q^{(2)}_{a} & =-(\mf n-2)\sigma^{(1)}\r_a^{(2)} +(\mf n-1) V^c{}_a\big(2\sigma^{(1)}\r_c-\nablacero_c\sigma^{(1)}\big)+\left(\dfrac{\mf n-2}{2}\sigma^{(1)} n(\elltwo)-(\mf n-1)\sigma^{(2)}\right)(\r-\s)_a\nonumber\\
&+(\mf n-1)\nablacero_a\sigma^{(2)}-\sigma^{(1)}\left(P^{bc}\big(A_{bca}+(\r+\s)_b(\Y+\F)_{ac}\big)+\dfrac{R}{2\mf n}\ell_a-\dfrac{1}{2}\kappa_n\nablacero_a\elltwo\right),\label{dQ2general}\\
\hspace{-0.25cm}\ddot\Q^{(2)} & = (\mf n-1)\sigma^{(2)} - \sigma^{(1)}\tr_P\bY^{(2)} + \sigma^{(1)}P^{ab}P^{cd}(\Y+\F)_{ac}(\Y+\F)_{bd} + \sigma^{(1)}P(\br-\bs,d\elltwo) \nonumber\\
&\quad\, -\dfrac{R}{2\mf n}\sigma^{(1)}\elltwo,\label{ddQ2general}
\end{align}
where $R$ is given by \eqref{scal} and needs not be written out explicitly.\\

We conclude this appendix by computing the scalars $f^{(m+1)}$, obtained after applying $\lie_{\xi}^{(m)}$ to $f=g^{\alpha\beta}\nabla_{\alpha}\Omega\nabla_{\beta}\Omega$ and using \eqref{derivada},
\begin{equation}
	\label{fm+1}
	\lie_{\xi}^{(m)}f = \sum_{k=0}^m \binom{m}{k} \big(\lie_{\xi}^{(k)}g^{\alpha\beta}\big) \sum_{j=0}^{m-k}\binom{m-k}{j} \nabla_{\alpha}\big(\lie_{\xi}^{(j)}\Omega\big)\nabla_{\beta}\big(\lie_{\xi}^{(m-k-j)}\Omega\big).
\end{equation}
Again, the lowest orders $f^{(1)}$, $f^{(2)}$ and $f^{(3)}$ need a separate treatment. Obviously $f^{(1)}=0$ because $\scri$ is null. To compute $f^{(2)}$ we particularize \eqref{fm+1} to $m=1$ so that $$\lie_{\xi}f = \big(\lie_{\xi} g^{\alpha\beta}\big)
\nabla_{\alpha} \Omega \nabla_{\beta} \Omega + 2 \nabla^{\alpha} \Omega \nabla_{\alpha} \lie_{\xi} \Omega.$$ 

The first term is $(\lie_{\xi}g^{\alpha\beta})\nabla_{\alpha}\Omega\nabla_{\beta}\Omega= 2(\sigma^{(1)})^2 \kappa_n$ (by \eqref{liegnu} and $\Omega\st{\scri}{=}0$) and the second term is $2\sigma^{(1)}\lie_n\sigma^{(1)}$ because $\nabla^{\alpha}\Omega \st{\scri}{=} \sigma^{(1)}\nu^{\alpha}$. Thus,
\begin{equation}
	\label{f2general}
	f^{(2)} = 2\sigma^{(1)} \big(\sigma^{(1)}\kappa_n +\lie_n\sigma^{(1)}\big).
\end{equation}
We now compute $f^{(3)}$. The expanded form of \eqref{fm+1} for $m=2$ is $$f^{(3)} = 2 \nabla^{\alpha}\Omega\nabla_{\alpha}(\lie_{\xi}^{(2)}\Omega) +2 g^{\alpha\beta} \nabla_{\alpha}(\lie_{\xi}\Omega)\nabla_{\beta}(\lie_{\xi}\Omega) + 4 (\lie_{\xi}g^{\alpha\beta})\nabla_{\alpha}(\lie_{\xi}\Omega)\nabla_{\beta}\Omega+(\lie^{(2)}_{\xi}g^{\alpha\beta}) \nabla_{\alpha}\Omega\nabla_{\beta}\Omega.$$ 

Now, using $\nabla^{\alpha}\Omega \st{\scri}{=} \sigma^{(1)}\nu^{\alpha}$, together with \eqref{liegnu} and
\begin{align*}
	g^{\alpha\beta}\nabla_{\alpha}(\lie_{\xi}\Omega)\nabla_{\beta}(\lie_{\xi}\Omega) &\stb[\eqref{inversemetric}]{\scri}{=} P^{ab}\nablacero_a\sigma^{(1)}\nablacero_b\sigma^{(1)} +2\sigma^{(2)}\lie_n\sigma^{(1)},\\
	(\lie^{(2)}_{\xi}g^{\alpha\beta})\nabla_{\alpha}\Omega\nabla_{\beta}\Omega &\st{\scri}{=} (\sigma^{(1)})^2\nu^{\alpha}\nu^{\beta}\big(2g^{\mu\nu}\mc{K}_{\alpha\mu}\mc{K}_{\beta\nu}-\mc{K}^{(2)}_{\alpha\beta}\big)\\
	&\hspace{-0.3cm}\stb[\eqref{inversemetric},\eqref{Kupnu}]{\scri}{=} (\sigma^{(1)})^2\big( 8P(\br,\br) -4\kappa_n n(\elltwo) +2\kappa^{(2)}\big),
\end{align*}
it follows that
\begin{align}
	f^{(3)} &=2\sigma^{(1)}\big(\sigma^{(1)}\kappa^{(2)} + \lie_n\sigma^{(2)}\big) + 2\big(2\sigma^{(2)}-\sigma^{(1)}n(\elltwo)\big)\big(\lie_n\sigma^{(1)}+2\sigma^{(1)}\kappa_n\big) \nonumber\\
	&\quad\,  +8\sigma^{(1)}P\big(\br,\sigma^{(1)}\br-d\sigma^{(1)}\big) + 2P^{ab}\nablacero_a\sigma^{(1)}\nablacero_b\sigma^{(1)}.\label{f3general}
\end{align}
We now deal with the rest of the cases in the following proposition.
\begin{prop}
	\label{propf}
	Let $f\d g^{\alpha\beta}\nabla_{\alpha}\Omega\nabla_{\beta}\Omega$. Then, for any $m\ge 2$,
	\begin{equation}
		\label{fm2}
		f^{(m+1)}\st{(m)}{=} 2\sigma^{(1)} \big(\lie_n\sigma^{(m)} + \sigma^{(1)} \kappa^{(m)}\big) +2m\big(\lie_n\sigma^{(1)}+2\sigma^{(1)}\kappa_n\big) \sigma^{(m)}.
	\end{equation}
	\begin{proof}
The case $m=2$ is \eqref{f3general}, which after ignoring terms of order one and zero becomes \eqref{fm2}. For $m\ge 3$, the only terms in \eqref{fm+1} that have a chance to depend on $\bY^{(m)}$ or $\sigma^{(m)}$ are $k=0$ (and $j=0,1,m-1,m$), $k=1$ (and $j=0,m-1$) and $k=m$, namely 
\begin{align*}
	f^{(m+1)} &\st{[m]}{=} \nabla^{\beta}\Omega \nabla_{\beta}\big(\lie^{(m)}_{\xi}\Omega\big)+2mg^{\alpha\beta}\nabla_{\alpha}\big(\lie_{\xi}\Omega\big)\nabla_{\beta}\big(\lie_{\xi}^{(m-1)}\Omega\big)+ 2m\big(\lie_{\xi}g^{\alpha\beta}\big)\nabla_{\alpha}\big(\lie_{\xi}^{(m-1)}\Omega\big)\nabla_{\beta}\Omega \\
	&\quad\, + \big(\lie_{\xi}^{(m)}g^{\alpha\beta}\big) \nabla_{\alpha}\Omega\nabla_{\beta}\Omega,
\end{align*}
where we used the fact that $m\ge 3$ because we have used that the terms with $k=0$ and $j=0,1,m-1,m$ are all different. When $m=2$ we would be overcounting (but it turns out that the extra factor compensates, so \eqref{fm2} is valid also for $m=2$). Now one uses $\nabla^{\beta}\Omega \st{\scri}{=} \sigma^{(1)}\nu^{\beta}$ in the first term, $\nabla^{\alpha}\big(\lie_{\xi}^{(m-1)}\Omega\big) \st{(m)}{=}\sigma^{(m)}\nu^{\alpha}$ in the second, and in the third and fourth 
\begin{align}
(\lie_{\xi}g^{\alpha\beta})\nabla_{\beta}\Omega &\st{\scri}{=} -\sigma^{(1)}\mc{K}^{\alpha}{}_{\beta}\nu^{\beta} \st{\eqref{Kupnu}}{=} -\sigma^{(1)}\left(2P^{cd}\r_c + \dfrac{1}{2}n(\elltwo)n^d\right)e_d^{\alpha} +2\sigma^{(1)}\kappa_n\xi^{\alpha},\label{liegnu}\\
\big(\lie_{\xi}^{(m)}g^{\alpha\beta}\big) \nabla_{\alpha}\Omega\nabla_{\beta}\Omega &\st{\scri}{=}(\sigma^{(1)})^2 \big(\lie_{\xi}^{(m)}g^{\alpha\beta}\big)\nu_{\alpha}\nu_{\beta} \st{(m)}{=} -(\sigma^{(1)})^2 \mc{K}^{(m)}_{ab} n^a n^b = 2(\sigma^{(1)})^2\kappa^{(m)}.\nonumber
\end{align}
After adding up the four terms one gets $$f^{(m+1)}\st{(m)}{=} 2\sigma^{(1)}\lie_n\sigma^{(m)}+2m(\lie_n\sigma^{(1)})\sigma^{(m)}+4m\sigma^{(1)}\kappa_n\sigma^{(m)}+2(\sigma^{(1)})^2\kappa^{(m)},$$ 

which is \eqref{fm2} for $m\ge 3$. 
	\end{proof}
\end{prop}

\section{Higher order Raychaudhuri equation}
\label{appendixD}	

In this appendix we compute the contraction $\L^{(m+1)}_a n^a$ to one order less, i.e. up to order $m$. This computation requires knowing $\mc R^{(m+1)}_{ab} n^a n^b$ up to order $m$ and turns out to be quite laborious using the general formalism of hypersurface data. To keep the computations to a reasonable length we have decided to perform the computation using a particular gauge and assuming that the hypersurface is totally geodesic and admits a foliation by cross-sections. The result presented here will be sufficient for the purposes of this paper.

\begin{prop}
Let $\mc H$ be a null hypersurface embedded in $(\mc M,g)$ and assume the existence of a foliation $\{\mc S_u\}_{u\in\real}$ of $\mc H$ by cross-sections. Let $\xi$ be the (uniquely defined) rigging satisfying (i) $g(\xi,\xi)\st{\mc H}{=}0$ and (ii) $\Phi^{\star}\bm\xi=du$. Suppose in addition that $\lie_n\bg=0$. Then,
	\begin{equation}
		\label{Raychaudhuri2}
		\begin{aligned}
			\mc{R}^{(m+1)}_{ab}n^a n^b & \st{(m)}{=} 2m\kappa_n\kappa^{(m+1)} -\lie_n^{(2)}\big(\tr_P\bY^{(m)}\big) +\kappa_n\lie_n\big(\tr_P\bY^{(m)}\big) \\
			&\qquad + \big(2\kappa_n^2-\lie_n\kappa_n\big)m\tr_P\bY^{(m)}+ \mc{O}_{\mc R},
		\end{aligned}
	\end{equation}
	where $\mc{O}_{\mc R}$ is a scalar that depends on $\br^{(m)}$ and lower order terms.
	\begin{proof}
We construct Gaussian null coordinates \cite{moncrief1983symmetries} $\{t,u,x^A\}$ from $\mc H=\{t=0\}$ in which the metric takes the form $$g = 2 dt du + \phi du^2 + 2\beta_A dx^A du + \mu_{AB} dx^A dx^B,$$ 

where $\phi$, $\bm\beta$ and $\mu$ are, respectively, a function, a one-form and a 2-covariant, symmetric tensor satisfying $\phi(t=0)=0$, $\bm\beta(t=0)=0$ and $\mu_{AB}(t=0)=\gamma_{AB}$. In these coordinates, the rigging is $\xi=\partial_t$ and $\nu=\partial_u$. Denoting the derivative w.r.t. $t$ at $t=0$ with a dot, it is easy to check that, at $t=0$, $\kappa_n=-\frac{1}{2}\dot\phi$, $\r_A=\frac{1}{2}\dot{\bm\beta}$ and $\Y_{AB}=\frac{1}{2}\dot\mu_{AB} $. If we denote by $\dot{f}^{(m)}$ the $m$-th $t$-derivative at $t=0$, one also cheeks easily that $\kappa^{(m)}=-\frac{1}{2}\dot\phi^{(m)}$, $\r^{(m)}_A=\frac{1}{2}\dot{\bm\beta}^{(m)}$ and $\Y_{AB}^{(m)}=\frac{1}{2}\dot\mu_{AB}^{(m)}$. We also note that $P^{ab}=\mu^{AB}\delta_A^a\delta_B^b$ and $n^a=\delta^a_u$. In these coordinates, the $(u,u)$ component of the Ricci tensor is given by \cite{moncrief1983symmetries}
\begin{align*}
	\mc{R}_{uu} & = \dfrac{1}{2}\dotr{\left( \phi \dot{\phi}+\beta^A\big(D_A\phi -2\beta_A'-\dot{\phi}\beta_A \big)\right)} + \dfrac{1}{4}\left( \phi' +\phi \dot{\phi}+\beta^A\big(D_A\phi -2\beta_A'-\dot{\phi}\beta_A\big) \right) \mu^{AB}\dot{\mu}_{AB} \\
	&\quad\, +\dfrac{1}{2}D^A\left(2\beta_A'-D_A\phi + \dot{\phi}\beta_A\right) -\dfrac{1}{4}\dot{\phi}\mu^{AB}\mu_{AB}' -\dfrac{1}{2}\mu^{AB}\mu_{AB}'' + \dfrac{1}{4}\mu^{AC}\mu^{BD}\mu_{AB}'\mu_{CD}' \\
	&\quad\, -\dfrac{1}{2}\big(\dot{\phi}-\beta^A\dot{\beta}_A\big)^2 - \mu^{AC}\mu^{BD}D_{[B}\beta_{A]}D_{[C}\beta_{D]}\\
	&\quad\, -\dfrac{1}{2}\dot{\beta}^C\big(2D_C\phi +\phi\dot{\beta}_C -\beta^A \big(4D_{[C}\beta_{A]}+\beta_A \dot{\beta}_C\big) - 2\beta_C' \big),
\end{align*}
where the prime means derivative w.r.t. $u$ and $D$ is the Levi-Civita connection of the induced metric in the codimension-two surfaces $u=const.$, $t=const.$. Obviously this connection depends on $u$ and $t$ (we do not reflect this dependence for notational simplicity). Note however that at $t=0$, the connection is independent of $u$ because $\partial_u\gamma_{AB}=0$ (given that
$\lie_{n} \gamma =0$). So, when an expression is evaluated at $t=0$ (e.g. all expressions with an $(m)$ on top of the equal sign), $D$ will mean the covariant derivative associated to the metric $\gamma_{AB}$. The strategy now is to take $m$ derivatives along $t$ and keep only the terms depending on $\dot{\phi}^{(m+1)}$, $\dot{\bm\beta}^{(m+1)}$, $\dot{\mu}^{(m+1)}$ and $\dot{\mu}^{(m)}$. This requires commuting $\partial_t$ with $D$ using a formula analogous to \eqref{propMarc},
\begin{align*}
\big[\partial_t^{(m)},D_A\big] T^{B_1\cdots B_q}{}_{C_1\cdots C_p} =  \sum_{k=0}^{m-1}\binom{m}{k+1}&\left(\sum_{j=1}^{q}\big(\dot T^{(m-k-1)}\big)^{B_1\cdots B_{j-1}D B_{j+1}\cdots B_q}{}_{C_1\cdots C_p}\Xi^{(k+1)}{}^{B_j}{}_{DA}\right. \\
&\left.- \sum_{i=1}^{p}\big(\dot T^{(m-k-1)}\big)^{B_1\cdots B_q}{}_{C_1\cdots C_{i-1}D C_{i+1}\cdots C_p}\Xi^{(k+1)}{}^{D}{}_{C_i A}\right),
\end{align*}
where $\Xi^A{}_{BC} = D_{(B}\dot\mu_{C)}{}^A - \dfrac{1}{2}D^A\dot\mu_{BC}$. Then, $\Xi^{(k)A}{}_{BC}$ will depend on $\dot\mu^{(j)}$ with $j=1,...,k$. Consequently, for any tensor field $T$ satisfying $T|_{t=0}=0$, the previous identity implies $$\partial_t^{(m)}D_A T^{B_1\cdots B_q}{}_{C_1\cdots C_p} \st{(m)}{=} D_A  \big(\dot T^{(m)}\big)^{B_1\cdots B_q}{}_{C_1\cdots C_p}.$$ 

This is the case, in particular, of $D_A\phi$, $\beta_A$, $\beta_A'$, $\mu_{AB}'$ and $\mu_{AB}''$. Using this rule and that $\mu_{AB}'(t=0)=\beta_A(t=0)=\phi(t=0)=0$, the $m$-th $t$-derivative of $\mc{R}_{uu}$ at $t=0$ is
		\begin{align*}
\mc{R}^{(m+1)}_{uu}=\dot{\mc R}^{(m)}_{uu} & \st{(m)}{=} \dfrac{1}{2}m\dot{\phi}\dot{\phi}^{(m+1)} + \dfrac{1}{4}m\big(\dot{\phi}'+\dot{\phi}^2\big)\mu^{AB}\dot{\mu}^{(m)}_{AB}  -\dfrac{1}{4}\dot{\phi}\mu^{AB}\dot{\mu}_{AB}^{(m)}{}'-\dfrac{1}{2}\mu^{AB}\dot{\mu}^{(m)}_{AB}{}''\\
			&\quad\,  +\mbox{terms depending on } \dot{\bm\beta}^{(m)} \mbox{ and } \dot\phi^{(m)},
		\end{align*}
which is \eqref{Raychaudhuri2} after replacing $\dot\phi=-2\kappa_n$, $\dot\phi^{(m+1)}=-2\kappa^{(m+1)}$, $\dot\mu^{(m)}_{AB} = 2\Y^{(m)}_{AB}$ and taking into account $\gamma_{AB}'=0$.
	\end{proof}
\end{prop}

We now use the previous proposition to compute $\L^{(m+1)}_a n^a$ up to order $m$ under the same assumptions on $\scri$ and the rigging.

\begin{prop}
	\label{cor_Ln}
Assume $\scri$ admits a foliation of cross-sections $\{\mc S_u\}_{u\in\real}$ and let $\xi$ be the rigging satisfying $g(\xi,\xi)\st{\scri}{=}0$ and $\Phi^{\star}\bm\xi= du$. Suppose in addition that $\lie_n\bg=0$. Then,
	\begin{equation}
		\label{Lnunordenmenos}
		\mc{L}_a^{(m+1)}n^a  \st{(m)}{=}  H_1\kappa^{(m+1)} - \sigma^{(1)}\lie_n^{(2)}\big(\tr_P\bY^{(m)}\big) +H_2\lie_n\big(\tr_P\bY^{(m)}\big) + H_3\tr_P\bY^{(m)} + \mc{O}_{\mc L},
	\end{equation}
	where $\mc{O}_{\mc L}$ is a scalar that depends on $\br^{(m)}$ and lower order terms and whose explicit form is not relevant for this paper, and 
\begin{align*}
H_1&\d -\dfrac{m(\mf n-1)\lie_n\sigma^{(1)}}{\mf n},& H_2&\d -(2m-1)\sigma^{(1)}\kappa_n - m(\mf n-2)\lie_n\sigma^{(1)},\\
H_3&\d \dfrac{(2m-\mf n)m}{\mf n}\lie_n\sigma^{(1)} -m\sigma^{(1)}\lie_n\kappa_n.&&
\end{align*}
	\begin{proof}
We begin by computing $\L_{\alpha}^{(m+1)}$ to order $[m]$. From \eqref{mcLalpham} and Proposition \ref{prop_schouten}, the only terms that could depend on $\bY^{(m)}$ are $i=m$, $i=m-1$ ($j=0,1$) and $i=m-2$ ($j=0,1,2$), namely
\begin{align*}
\L_{\alpha}^{(m+1)} & \st{[m]}{=} (\mf n-1) \Bigg(L^{(m+1)}_{\alpha\beta}\nabla^{\beta}\Omega + m L_{\alpha\beta}^{(m)}\left( g^{\beta\mu}\nabla_{\mu}\lie_{\xi}\Omega+(\lie_{\xi}g^{\beta\mu})\nabla_{\mu}\Omega\right)\\
&\qquad\qquad\quad + \dfrac{m(m-1)}{2} L^{(m-1)}_{\alpha\beta}\big(g^{\beta\mu}\nabla_{\mu}\lie^{(2)}_{\xi}\Omega + 2(\lie_{\xi}g^{\beta\mu})\nabla_{\mu}\lie_{\xi}\Omega + (\lie_{\xi}^{(2)}g^{\beta\mu})\nabla_{\mu}\Omega\big)\Bigg)
\end{align*}
Contracting with $\nu^{\alpha}$, evaluating at $\scri$ and using as usual $\nabla^{\beta}\Omega\st{\scri}{=}\sigma^{(1)}\nu^{\beta}$ we get
\begin{align*}
	\mc{L}_a^{(m+1)}n^a & \st{(m)}{=} (\mf n-1)\left(\sigma^{(1)}L^{(m+1)}_{\alpha\beta}\nu^{\alpha}\nu^{\beta} + m \nu^{\alpha}L_{\alpha\beta}^{(m)}\left( g^{\beta\mu}\nabla_{\mu}\lie_{\xi}\Omega+(\lie_{\xi}g^{\beta\mu})\nabla_{\mu}\Omega\right) \right.\\
	&\qquad \qquad + \left. \dfrac{m(m-1)}{2} \nu^{\alpha}L^{(m-1)}_{\alpha\beta}\big(g^{\beta\mu}\nabla_{\mu}\lie^{(2)}_{\xi}\Omega + 2(\lie_{\xi}g^{\beta\mu})\nabla_{\mu}\lie_{\xi}\Omega + (\lie_{\xi}^{(2)}g^{\beta\mu})\nabla_{\mu}\Omega\big)\right).
\end{align*}
It is easier to split the computation into three pieces, namely $\mbox{I} = (\mf n-1)\sigma^{(1)}L^{(m+1)}_{\alpha\beta}\nu^{\alpha}\nu^{\beta}$, $\mbox{II} = (\mf n-1) m \nu^{\alpha}L_{\alpha\beta}^{(m)}\left( g^{\beta\mu}\nabla_{\mu}\lie_{\xi}\Omega+(\lie_{\xi}g^{\beta\mu})\nabla_{\mu}\Omega\right)$ and $$\mbox{III} = \dfrac{m(m-1)(\mf n-1)}{2} \nu^{\alpha}L^{(m-1)}_{\alpha\beta}\Big(g^{\beta\mu}\nabla_{\mu}\lie^{(2)}_{\xi}\Omega + 2(\lie_{\xi}g^{\beta\mu})\nabla_{\mu}\lie_{\xi}\Omega + (\lie_{\xi}^{(2)}g^{\beta\mu})\nabla_{\mu}\Omega\Big),$$ 

and compute each term separately. Those involving only $\br^{(m)}$ or $\kappa^{(m)}$ are not needed in explicit form. We shall gather all of them into the tensor $\mc{O}_{\mc L}$ in the statement of the proposition. For the calculation we shall write $\cdots$ to mean additional terms of this form. For the first term $\mbox{I}$ we contract \eqref{Lm+1ordm} (with $m+1$ instead of $m$) with $\nu^{\alpha}\nu^{\beta}$ and use $\mc{K}_{\alpha\beta}\nu^{\alpha}\nu^{\beta}=-2\kappa_n$ and $\mc{K}^{(2)}_{\alpha\beta}\nu^{\alpha}\nu^{\beta}=-2\kappa^{(2)}$ to get $$\mbox{I} \st{(m)}{=} \sigma^{(1)}\left(\mc R_{ab}^{(m+1)}n^{a}n^{b}+\dfrac{1}{\mf n}\left(m\kappa_n R^{(m)} + \dfrac{m(m-1)}{2}\kappa^{(2)} R^{(m-1)} \right)\right).$$ 

Inserting \eqref{Raychaudhuri2} as well as $$R^{(m)} \st{(m)}{=} -2\kappa^{(m+1)} -4\lie_n(\tr_P\bY^{(m)})-4m\kappa_n\tr_P\bY^{(m)} + \cdots,\qquad R^{(m-1)} \st{(m)}{=}\cdots,$$ 

we arrive at 
\begin{align*}
\mbox{I} &\st{(m)}{=} \dfrac{2m(\mf n-1)\sigma^{(1)}\kappa_n}{\mf n}\kappa^{(m+1)}- \sigma^{(1)}\lie_n^{(2)}\big(\tr_P\bY^{(m)}\big) + \dfrac{\mf n-4m}{\mf n}\kappa_n\sigma^{(1)}\lie_n\big(\tr_P\bY^{(m)}\big)\nonumber\\
	&\qquad + \left(\dfrac{2(\mf n-2m)\kappa_n^2}{\mf n} - \lie_n\kappa_n\right)m\sigma^{(1)} \tr_P\bY^{(m)} + \cdots.
\end{align*} 
To compute $\mbox{II}$ we now use that (cf. \eqref{Lmmn}-\eqref{dLmn})
\begin{align}
L_{ab}^{(m)}n^a&\st{(m)}{=}\cdots\label{piece2}\\
\dot{L}^{(m)}_a n^a &\st{(m)}{=} \dfrac{1}{\mf n(\mf n-1)}\Big((1-\mf n)\kappa^{(m+1)} +(2-\mf n)\lie_n(\tr_P\bY^{(m)}) + (2m-\mf n)\kappa_n \tr_P\bY^{(m)}\Big)+\cdots.\label{piece3}
\end{align}
Inserting these into the definition of $\mbox{II}$ and using \eqref{inversemetric} and \eqref{liegnu} gives
\begin{align*}
\mbox{II} & = (\mf n-1) m \nu^{\alpha}L_{\alpha\beta}^{(m)}\left( g^{\beta\mu}\nabla_{\mu}\lie_{\xi}\Omega+(\lie_{\xi}g^{\beta\mu})\nabla_{\mu}\Omega\right)\\
& \st{(m)}{=} (\mf n-1) m \big(\lie_n\sigma^{(1)}+2\sigma^{(1)}\kappa_n\big)\dot{L}^{(m)}_a n^a\\
& \st{(m)}{=}\dfrac{m\big(\lie_n\sigma^{(1)}+2\sigma^{(1)}\kappa_n\big)}{\mf n} \Big((1-\mf n)\kappa^{(m+1)} +(2-\mf n)\lie_n(\tr_P\bY^{(m)}) + (2m-\mf n)\kappa_n \tr_P\bY^{(m)}\Big)+\cdots.
\end{align*}
Finally, it is clear from \eqref{piece2}-\eqref{piece3} that $\mbox{III}=\cdots$. To show \eqref{Lnunordenmenos} we only need to add $\L_a^{(m+1)}n^a\st{(m)}{=}\mbox{I}+\mbox{II}+\mbox{III}$ and simplify.
	\end{proof}
\end{prop}

\end{appendices}

	\begingroup
	\let\itshape\upshape
	
	\renewcommand{\bibname}{References}
	\bibliographystyle{acm}
	\bibliography{library} 
	
\end{document}